\documentclass[11pt]{article}
\usepackage{amssymb,amsfonts,amsmath,amsthm,amscd,dsfont,mathrsfs}
\usepackage{graphicx,float,psfrag,epsfig,color,url}
\usepackage{latexsym}

\newtheoremstyle{ex}
 {8pt}
{8pt}
 {}
{} 
{\bf}
{.}
{.5em}
 {}

\definecolor{Red}{rgb}{1,0,0}
\definecolor{Blue}{rgb}{0,0,1}
\definecolor{Olive}{rgb}{0.41,0.55,0.13}
\definecolor{Green}{rgb}{0,1,0}
\definecolor{MGreen}{rgb}{0,0.8,0}
\definecolor{DGreen}{rgb}{0,0.55,0}
\definecolor{Yellow}{rgb}{1,1,0}
\definecolor{Cyan}{rgb}{0,1,1}
\definecolor{Magenta}{rgb}{1,0,1}
\definecolor{Orange}{rgb}{1,.5,0}
\definecolor{Violet}{rgb}{.5,0,.5}
\definecolor{Purple}{rgb}{.75,0,.25}
\definecolor{Brown}{rgb}{.75,.5,.25}
\definecolor{Grey}{rgb}{.5,.5,.5}
\definecolor{Pink}{rgb}{1,0,1}
\definecolor{DBrown}{rgb}{.5,.34,.16}

\definecolor{Black}{rgb}{0,0,0}

\def\Pr{{\sf P}}

\def\E{{\mathbb E}}

\def\bx{{\bf X}}
\def\by{{\bf Y}}
\def\bz{{\bf Z}}
\def\cS{{\cal S}}

\def\os{\overline{s}}

\def\onu{\overline{\nu}}
\def\E{\mathbb{E}}

\def\de{{\rm d}}
\def\ind{{\mathbb I}}

\def\id{{\rm I}}
\def\supp{{\rm supp}}
\def\tR{\widetilde{R}}

\newcommand{\hx}{\widehat{x}}
\newcommand{\hp}{\widehat{p}}

\newcommand{\halpha}{\widehat{\alpha}}
\newcommand{\hPT}{\widehat{{\rm PT}}}
\newcommand{\hbeta}{\widehat{\beta}}

\newcommand{\normal}{\hat{x}_{1.\lambda}}
\def\normal{{\sf N}}

\def\sT{{+}}

\def\de{{\rm d}}

\def\reals{{\bf R}}
\def\naturals{{\bf N}}
\def\E{{\mathbb E}}
\def\<{\langle}
\def\>{\rangle}

\def\prob{{\mathbb P}}

\def\sign{{\rm sign}}

\def\mfp{{m_{\rm fp}}}
\def\mono{{\rm mono}}
\newcommand{\cMono}{{\cal M}}

\newcommand{\beqa}{\begin{eqnarray}}
\newcommand{\eeqa}{\end{eqnarray}}

\newcommand{\bitem}{\begin{itemize}}
\newcommand{\eitem}{\end{itemize}}
\newcommand{\beq}{\begin{equation}}
\newcommand{\eeq}{\end{equation}}
\newcommand{\goto}{\rightarrow}

\newcommand{\cA}{{\cal A}}
\newcommand{\cP}{{\cal P}}

\newcommand{\bR}{{\bf R}}

\newcommand{\bZ}{{\bf Z}}
\newcommand{\cE}{{\cal E}}

\newcommand{\cH}{{\cal H}}
\newcommand{\cV}{{\cal V}}
\newcommand{\eps}{{\varepsilon}}
\newcommand{\ve}{{\varepsilon}}

\newcommand{\argmin}{\mbox{argmin}}

\definecolor{cardinal}{rgb}{.64,0.,.11}

\numberwithin{equation}{section}
\newtheorem{theorem}{Theorem}[section]

\newtheorem{lemma}{Lemma}[section]
\newtheorem{proposition}{Proposition}[section]
\newtheorem{definition}{Definition}[section]

\theoremstyle{ex}
\newtheorem{example}{Example}[section]
\def\eex{\hfill $\blacksquare$}

\def\MSE{{\rm MSE}}

\def\cJ{{\cal J}}

\def\HFP{{\rm HFP}}

\def\MSE{{\rm MSE}}

\def\onu{\overline{\nu}}

\def\id{{\bf I}}

\def\sign{{\rm sgn}}
\def\sT{{\sf T}}
\def\onsager{{\sf b}}
\def\Soft{{\rm Soft}}
\def\SoftPos{{\rm SoftPos}}
\def\Pos{{\rm Pos}}
\def\Cap{{\rm Cap}}
\def\Firm{{\rm Firm}}
\def\Hard{{\rm Hard}}
\def\Minimax{{\rm Minimax}}
\def\Blocksoft{{\rm BlockSoft}}
\def\JamesStein{{\rm JamesStein}}
\def\Mono{{\rm MonoReg}}
\def\Oracle{{\rm Oracle}}
\def\efirm{\eta^{firm}}
\def\esoft{\eta^{soft}}
\def\ecap{\eta^{cap}}
\def\esoftpos{\eta^{softpos}}
\def\ehard{\eta^{hard}}
\def\eJS{\eta^{JS}}
\def\emono{\eta^{mono}}
\def\etv{\eta^{tv}}
\def\meas{{\cal L}}
\def\eall{\eta^{all}}

\def\div{{\rm div}}

\def\HFP{{\rm HFP}}
\def\oB{\overline{B}}
\def\ox{\overline{x}}
\def\omu{\overline{\mu}}
\def\oz{\overline{z}}
\def\eplus{\eta^{+}}
\def\tgamma{\widetilde{\gamma}}

\newcommand{\cF}{{\cal F}}

\title{Accurate Prediction of Phase Transitions in Compressed Sensing\\
via a Connection to Minimax Denoising}

\author{David  Donoho\thanks{Department of Statistics, Stanford
    University},\;\;\;\;\; Iain  Johnstone${}^*$\;\; and\;\;
Andrea Montanari${}^{*,}$\thanks{Department of Electrical Engineering, Stanford
University}}

\date{\today}
\begin{document}
\maketitle

\begin{abstract}
Compressed sensing posits that, within
limits, one can undersample a sparse signal
and yet reconstruct it accurately. Knowing the precise limits
to such undersampling is important both for theory and practice.

We present a formula that characterizes the allowed
undersampling of  generalized sparse objects. The formula applies to
Approximate Message Passing (AMP) algorithms for compressed sensing, 
which are here  generalized to employ  denoising operators besides the traditional 
scalar soft thresholding denoiser. This paper 
gives several examples
including  scalar denoisers not derived from  convex penalization -- 
the \emph{firm} shrinkage nonlinearity and  the \emph{minimax}  
nonlinearity  --  and also
nonscalar denoisers --  block thresholding,  monotone regression, and total variation minimization.

Let the variables $\eps = k/N$ and $\delta = n/N$
  denote the generalized sparsity and undersampling fractions for
  sampling the $k$-generalized-sparse $N$-vector $x_0$
  according to $y=Ax_0$.
Here $A$ is an $n\times N$  measurement matrix whose entries are  iid 
standard Gaussian.
The formula states that the
phase transition curve $ \delta = \delta(\eps)$ separating successful 
from unsuccessful reconstruction of $x_0$ by AMP
is given by:
  \[
     \delta = M(\eps|  \mbox{Denoiser}),
  \]
  where $M(\eps| \mbox{Denoiser})$ denotes the per-coordinate minimax
  mean squared error (MSE)
  of the specified, optimally-tuned denoiser in the \emph{directly 
observed} problem $y = x + z$.
  In short, the phase transition of a noiseless \emph{undersampling} 
problem is identical
  to the minimax MSE in a \emph{denoising} problem.
We prove that this formula follows from  state evolution and present numerical results
validating it in a wide range of settings.

The above formula generates numerous new insights, both in the
scalar and in the nonscalar cases.
  \end{abstract}

\vspace{.1in}
Lasso. Group Lasso,  Joint Sparsity, James-Stein,
Minimax Risk over Nearly-Black Objects.
Minimax Risk of Soft Thresholding.
Minimax Risk of Firm Thresholding. Minimax Shrinkage.
Nonconvex penalization. State Evolution.
Total Variation Minimization. Monotone Regression.
\vspace{.1in}

\vspace{.1in}
{\bf Acknowledgements.}
NSF DMS-0505303, NSF DMS-0906812, NSF CAREER CCF-0743978.
\vspace{.1in}

\tableofcontents 

\newpage

\newcommand{\peeTwoOne}{(P_{2,1})}
\newcommand{\peeTwoOneL}{(P_{2,1,\lambda})}
\newcommand{\peeOne}{(P_{1})}
\newcommand{\peeOneL}{(P_{1,\lambda})}

\section{Introduction}

In the noiseless compressed sensing problem,
we are given a collection of  linear measurements of an unknown
vector $x_0$:
\beq
\label{eq:obsdata} y=Ax_0.
\eeq
Here the measurement matrix $A$ is $n$ by $N$, $n < N$,
and the $N$-vector $x_0$ is the object we wish to recover.
Both $y$ and $A$ are known, while $x_0$ is unknown and we
seek to recover an approximation to $x_0$.

Since  $n < N$, the equations are underdetermined.
It seems hopeless to recover $x_0$ in general, but
in compressed sensing one also assumes that the object is
\emph{sparse} in the appropriate sense. 
Suppose that the object 
is known to be \emph{$k$-sparse}, i.e. to have $k$ nonzero entries. 
If the problem dimensions $(k,n,N)$ are large, many recovery algorithms
exhibit the phenomenon
of \emph{phase transition}.

Explicitly, let $\eps = k/N$ and $\delta = n/N$ denote the sparsity and undersampling parameters,
respectively. Hence  $(\eps,\delta) \in [0,1]^2$ 
defines a phase space for the different kinds of limiting situations we may encounter
as $(k,n,N)$ grow large.  For a variety of algorithms and Gaussian
matrices $A$ with iid entries, one finds that
this phase space can be partitioned into two phases: ``success'' and ``failure''.
Namely, for a given algorithm $\cA$ and given sparsity fraction $\eps$, there exists a critical fraction
$\delta(\eps| \cA)$ such that if the sampling rate $\delta$ is larger
than the critical value, $\delta > \delta(\eps|\cA)$,
then the algorithm is successful in recovering the underlying object
$x_0$ with high probability\footnote{Throughout the paper, we will
  write that an event holds with high probability (w.h.p.) if its
  probability  converges to $1$ in the large system limit
  $N,n\to\infty$ with $\delta =n/N$ and $\eps=k/N$ fixed.}, while
if $\delta <  \delta(\eps|\cA)$ the algorithm is unsuccessful, also
with high probability.
In particular,  $\delta(\eps|\cA) < 1$ means that it is indeed possible
to undersample and still recover the unknown signal. In fact $\delta(\eps|\cA)$
shows precisely the limits of allowable undersampling.
By now a large amount of empirical and theoretical knowledge has been compiled about the
phase transitions exhibited by different algorithms: we refer the
reader to
 \cite{DoTa10,Do,DoTa08,HassibiXu,BlCaTa11,StojnicMMV10,KabashimaTanaka,DMM09,DMM-NSPT-11,Wainwright2009}.
In a parallel line of work, a number of sufficient conditions under
which undersampling is possible using deterministic matrices have been
studied, see e.g. \cite{CandesTao,BickelEtAl,Indyk,CandesReview}.

It is however fair to say that the research focused so far on
`unstructured' notions of sparsity whereby $k$ simply counts the
number on non-zero entries in $x_0$. (We refer to Section \ref{sec:Related} for
an overview of related literature.) On the other hand, applications
naturally lead to `structured' notions of sparsity.
This paper applies an algorithm framework \--  \emph{Approximate Message Passing} (AMP) \--
 to construct specific algorithms applicable to a variety of compressed sensing settings, 
including block and structured sparsity, convex and nonconvex penalization, 
and  develops a single unifying formula that, specialized
to each instance, gives the actual phase transition that we observe in practice.
To give a preview of our results, we first recall some facts about
statistical decision theory and AMP reconstruction. For the sake of
illustration, the classical case of simple sparse vectors will be used as a
running example.

\subsection{Signal models}
\label{sec:SignalModels}

Throughout this paper, we will consider  estimation of unknown
structured signals $x\in\bR^N$ from a minimax point of view. 
Various notions of structures can be formalized by considering a
family $\cF_N$ of probability measures over $\bR^N$. One such
probability measure will be denoted by $\nu_N\in\cF_N$ and a signal
with distribution $\nu_N$ will often be
denoted as $\bx\sim \nu_N$. 

The family $\cF_N$ will typically include degenerate distributions,
i.e. point masses $\nu_N = \delta_{x_0}$ for some $x_0\in\bR^N$. 
\begin{example}
The case of simple sparse vectors corresponds to the family 
\begin{eqnarray}
\cF_{N,\eps}\equiv\Big\{\nu_N\in\cP(\bR^N)\, :\;\;\;
\E_{\nu_N}\big\{\|\bx\|_0\big\}\le N\eps \Big\}\, ,\label{eq:SimpleSparse}
\end{eqnarray}
where $\cP(\bR^N)$ is the set of Borel probability measures on $\bR^N$
and, as usual, $\|v\|_0$ denotes the number of nonzero entries of the
vector $v$. \eex
\end{example}
As exemplified in this case $\cF_N$ is often indexed by a sparsity
parameter $\eps$, with $N\eps$ corresponding to the number of non-zero
entries. We will sometimes use the notation $\cF_{N,\eps}$ to
indicate this dependency, also beyond the last example. Two further
common properties that will 
always hold unless otherwise stated are the following. 
\begin{enumerate}
\item \emph{Nestedness.} If $\eps_1\le\eps_2$ then  $\cF_{N,\eps_1}\subseteq\cF_{N,\eps_2}$. 
\item \emph{Scale invariance.} If $\nu_N\in\cF_{N,\eps}$ then any
  scaled version of $\nu_N$ (defined by letting $\nu^a_N(B) =
  \nu(aB)$ for some $a>0$) is also in $\cF_{N,\eps}$.
\end{enumerate}
We will often omit the subscript $N$ if $N=1$.
%
%
\subsection{Denoising and minimax MSE}

The denoising problem requires to reconstruct a signal
$x\in\bR^N$ from observations $\by = x+\bZ$ whereby $\bZ\sim
\normal(0,\sigma^2\id_{N\times N})$ is a noise vector of known
variance. (Here and below $\id_{m\times m}$ denoted the identity
matrix in $m$ dimensions.)
A denoiser is a mapping 
\begin{align*}
\eta(\, \cdot\, ;\tau,\sigma):&\bR^N\to \bR^N\, ,\\
&y\mapsto \eta(y;\tau,\sigma)\, ,
\end{align*}
that returns an estimate of $x$ when applied to observations
$y=\by$. The denoiser is parametrized by the noise scale $\sigma$
and additional tuning parameters $\tau\in\Theta$. Often 
denoisers have the property $\|\eta(y;\tau,\sigma)\|_2\le \|y\|_2$ and
are hence called `shrinkers'. We will often have $\Theta = \bR_+$,
i.e. the denoiser depends on a single non-negative parameter, but more
complex choices of the parameter space $\Theta$ fit in the formalism
as well.

Following the minimax formulation in the previous section, we evaluate denoisers
on signals  $\bx\sim\nu_N\in\cF_{N,\eps}$, for specific class of
distributions $\cF_{N,\eps}$. Because of the scale invariance
property of $\cF_{N,\eps}$, it is sufficient to consider scale
invariant denoisers:
\begin{align}
\eta(y;\tau,\sigma) = \sigma\,\eta\Big(\frac{y}{\sigma};\tau,1\Big)
\equiv \sigma\,\eta\Big(\frac{y}{\sigma};\tau\Big)\, ,\label{eq:ScalingRelation}
\end{align}
Hence we omit the last argument when this is $\sigma=1$.
We evaluate a denoiser $\eta$ through its minimax mean square error (MSE)
per coordinate 
\begin{eqnarray}
 \label{generalM}
     M(\cF_{N,\eps} | \eta) =\frac{1}{N} \inf_{\tau\in\Theta} \sup_{\nu_N \in \cF_{N,\eps}} \E_{\nu_N} \Big\{\big\| \bx - \eta(\by; \tau) \big\|_2^2\Big\} ,
\end{eqnarray}
where expectation is taken with respect to $\bx\sim\nu_N$ and
$\by=\bx+\bz$, $\bz\sim\normal(0,\id_{N\times N})$.  
In words, we tune the denoiser optimally to control the
(per-coordinate) mean square error
for typical signals from even the most unfavorable
choice within our class $\cF_{N,\eps}$.

In the following we will be particularly interested in the
high-dimensional limit of the minimax MSE. 
It will be implicitly understood that we are given a sequence of probability
distributions classes $\{\cF_{N,\eps}\}_{N\ge 1}$ indexed by the dimension and a
sequence of denoisers $\eta = \{\eta_N\}_{N\ge 1}$ also indexed by the
dimension (the subscript $N$ will be omitted on $\eta$). 
We define the asymptotic minimax MSE through the
following limit (whenever it exists)
\begin{eqnarray}
 M(\eps | \eta) =\lim_{N\to\infty}\frac{1}{N} \inf_{\tau\in\Theta}
 \sup_{\nu_N \in \cF_{N,\eps}} \E_{\nu_N} \Big\{\big\| \bx - \eta(\by;
 \tau) \big\|_2^2\Big\} . \label{eq:generalMAsymptotic}
\end{eqnarray}

We say that a denoiser is \emph{separable} if, for
$v=(v_1,\dots,v_N)\in\bR^N$, we have 
\begin{eqnarray}
\eta_N(v;\tau) = \big(\eta_1(v_1;\tau),\eta_1(v_2;\tau),\dots,\eta_1(v_N;\tau)\big)\, . \label{eq:Separable}
\end{eqnarray}
\begin{example}\label{example:SoftThresholding}
A well studied denoiser is coordinatewise  soft-thresholding, that we
will denote by $\esoft$. This is a separable denoiser with a unique
parameter $\tau\in \Theta=\bR_+$ (the threshold). On each coordinate
$y\in\bR$ this acts as
\begin{align*}
\esoft(y;\tau) = \begin{cases}
y-\tau& \mbox{ if $\tau\le y$,}\\
0& \mbox{ if $-\tau\le y\le \tau$,}\\
y+\tau & \mbox{ if $ y<-\tau$.}
\end{cases}
\end{align*}
Soft thresholding is well suited for sparse signals from the class
$\cF_{N,\eps}$ defined in Eq.~(\ref{eq:SimpleSparse}).
It turns out that the resulting minimax MSE $M(\cF_{N,\eps} | \eta)$  can be characterized in
terms of a scalar estimation problem, namely for all $N$,
$M(\eps|\esoft)=M(\cF_{N,\eps}|\esoft) = M(\cF_{1,\eps}|\esoft)$. 
Explicitly, all these quantities are given by
\begin{align*}
 M(\eps | \Soft) = \inf_{\tau\in\bR_+} \sup_{\nu \in \cF_\eps }
 \E_\nu\big\{ \big[X - \eta^{soft}(X+Z;\tau)\big]^2\big\} \, ,
\end{align*}
where expectation is taken with respect to $X\sim \nu$ and
$Z\sim\normal(0,1)$ independent of $X$.
We refer to \cite{DJ94a,DMM09,DMM-NSPT-11} for an explicit
characterization of this quantity (a summary being provided in Section
\ref{sec-Scalar}). In particular $M(\eps|\Soft)$ can be explicitly
evaluated.\eex
\end{example}
In several other  examples  $M(\cF_{N,\eps}| \eta)$ has been
explicitly evaluated (see \cite{DMM09}, Supplementary Information).
\begin{example} \label{example:Softpos}
 The positive-constrained case, where $x_0 \geq 0$ can be modeled by 
considering the family of probability distributions 
$\cF_{N,\eps,+} \equiv \{\nu_N\in\cP(\bR^N_+):\;
\E_{\nu_N}\{\|\bx\|_0\}\le N\eps\,\}$ supported in the positive
orthant. A natural denoiser is positive soft-thresholding $\esoftpos$. This is
again separable with, for $y\in\bR$,
$\esoftpos(y;\tau) = (y-\tau)_+\equiv \max(y-\tau,0)$. We have, again,
 $M(\eps|\esoftpos)=M(\cF_{N,\eps,+}|\esoftpos) =
 M(\cF_{1,\eps,+}|\esoftpos)$. 
 These quantities will be denoted by $M(\eps | \SoftPos)$.\eex
\end{example}
\begin{example} \label{example:Box}
The box-constrained case where $x_0 \in [0,1]^N$, 
can be modeled through the class $\cF_{N,\eps,\Box} \equiv \{\nu_N\in\cP([0,1]^N):\;
\E_{\nu_N}\{\sum_{i=1}^N 1_{\{x_i\in (0,1)\}}\}\le N\eps\,\}$.
A natural denoiser is coordinatewise capping. Namely, for $y\in\reals$ 
$\ecap(y) = \min(1,\max(y,0))$ (in this case there is no tuning
parameter, $\Theta=\emptyset$).
Notice that, in this case, the signal class is not scale invariant, and
hence the present framework does not apply directly. 
We discuss in Appendix \ref{Appendix:classical} how to modify it.\eex
 \end{example}
In this paper we will give several other calculations of
$M(\cF_{N,\eps}|\eta)$,
for signal structures and denoisers going considerably beyond these examples.

%
%
\subsection{Compressed sensing and AMP reconstruction}

Consider now the noiseless compressed sensing problem, i.e. the problem
of recovering a signal $x_0\in\bR^N$ from $n<N$ linear observations
$y=Ax_0$,  cf. Eq. (\ref{eq:obsdata}).
The key intuition is that this can be done exploiting the structure of
$x_0$, sparsity being a special example.
Approximate Message Passing (AMP) is an iterative scheme that allows
to exploit richer types of structure in a flexible way.
Given a denoiser $\eta(\,\cdot\,;\tau,\sigma):\bR^N\to\bR^N$ that is well suited for reconstructing $x_0$
from observations $x_0+\bz$, the AMP framework turns it into a scheme
for solving the compressed sensing problem.

The AMP iteration starts from $x^{0} = 0$,
and proceeds for iterations $t=1$, $2$, \dots  by maintaining
a \emph{current reconstruction} $x^{t}\in\bR^N$ and a \emph{current working residual} $z^{t}\in\bR^n$, and adjusting these
iteratively. At iteration $t$, it forms a vector of \emph{current pseudo-data} $ y^t = x^{t} + A^{\sT} z^{t}$
and the next iteration's estimate is obtained by applying $\eta$ to 
the current pseudo-data:
\begin{eqnarray}
        y^{t}  &=& x^{t} + A^{\sT} z^{t}\, ,                         \label{AMPA} \\
        x^{t+1}   &=& \eta(y^{t};\tau,\sigma_t)\, ,   \label{AMPB} \\
        z^{t+1} &=& y - Ax^{t+1} +  \onsager_t z^t\, .  \label{AMPC}  
\end{eqnarray}
Here $\onsager_t$ is a scalar determined by 
\begin{align}
\onsager_t \equiv \left.\frac{1}{n}
  \,\div\, \eta(y;\tau,\sigma_{t-1})\right|_{y=y^{t-1}}\, .
\end{align}
The rationale for this specific choice of $\onsager_t$ is discussed in
\cite{DMM09,DMM_ITW_I,BM-MPCS-2011}:
a justification goes betond the scope of the present paper.
The parameter $\sigma_t$ is can be interpreted as the noise standard
deviation for the pseudo-data $y^t$. This can be estimated from $y^t$
or $z^t$ as explained in Appendix \ref{Appendix:Computation}.

Conceptually, AMP  constructs an artificial denoising problem at each iteration
and solves it using the denoising defined by  $\eta$. In other words, it solves a compressed
sensing problem by successive denoising.
For the purpose of this paper, this description should be sufficient, save for two remarks.

\emph{First}, the specifics of the construction are
absolutely crucial for the  results of  this paper.
These are embedded in the specification of the scale factors $\onsager_t$ and $\sigma_t$.

\emph{Second}, the above algorithm framework was originally proposed in
  \cite{DMM09,DMM_ITW_I} in the case of a \emph{separable} denoiser $\eta$,
i.e. a denoiser acting independently on each coordinate.
In that paper the algorithm was derived by constructing a proper belief propagation message
passing algorithm,  and then obtaining the above algorithm as a
first-order approximation. Specific separable denoisers corresponded to
different choices of the prior in belief propagation.

A central point of this paper is  that  the form of the algorithm 
(\ref{AMPA}), (\ref{AMPB}), (\ref{AMPC}) is really more general and can be used in settings outside
the original definition.

\subsection{Phase transition for AMP}

A recurring property of AMP algorithms is that they undergo a
\emph{phase transition}. When the undersampling ratio $\delta$ decreases below
a certain threshold (that depends on the signal class $\cF_{\eps,N}$
and the denoiser $\eta$), the algorithm behavior changes from being
successful most of the times, to failing most of the times. In order
to formalize this notion, we introduce the following terminology.
\begin{definition}
We say that \emph{AMP succeeds with high probability} for the signal class
$\cF_{N,\eps}$, and denoising procedure $\eta$ if there exist a choice
of the tuning parameter $\tau\in\Theta$ such that the following
happens. For each $\xi>0$,
there exists a function $o(t)$ with $\lim_{t\to\infty} o(t) = 0$ such
that, for any $\nu_N\in \cF_{N,\eps}$,
\begin{align}
\lim\sup_{N\to\infty}\prob_{\nu_N}\big\{\|x^t-x_0\|_2^2\ge N\xi
\big\} \le o(t)\, .
\end{align}
Here probability is taken with respect to $x_0=\bx\sim \nu_N$ and the
sensing matrix $A$. Further, the limit $N\to\infty$ is taken with $n/N\to\delta$.

Viceversa we say that \emph{AMP fails with high probability}  for the signal class
$\cF_{N,\eps}$, and denoising procedure $\eta$ if for any
$\tau\in\Theta$ the following happens.  There exists
$\xi>0$ and a sequence $\nu_N\in\cF_{N,\eps}$ such that, for all $t\ge 0$
\begin{align}
\lim\sup_{N\to\infty}\prob_{\nu_N}\big\{\|x^t-x_0\|_2^2\ge N\xi
\big\} =1\, .
\end{align}
\end{definition}
Note that we could have chosen other, slightly different, notions of
convergence, e.g. requiring
$\prob_{\nu_N}\{\lim_{t\to\infty}\|x^t-x_0\|_2\}\to 1$ as $N\to\infty$.
The notion of convergence adopted here corresponds instead to
achieving arbitrarily small MSE per coordinate \emph{in a constant
number of iterations} (independent of $N$). This notion is more
appropriate for practical applications and better suited to the theory
of AMP algorithms (see Section \ref{sec:DerivePT}).

Our main result is the following general relation between denoising
and compressed sensing.
\begin{quotation}
\noindent{\bf Phase Transition Formula for AMP.}   
Consider compressed sensing reconstruction over the signal class
$\cF_{N,\eps}$, using AMP with the denoiser $\eta$. Denote by $M(\eps|
\eta)$ the asymptotic minimax MSE per coordinate using denoiser $\eta$.

Then AMP succeeds with high probability if
\begin{equation} \label{generalPT}
    \delta > M(\eps| \eta) .
\end{equation}
Viceversa AMP fails with high probability for   $\delta < M(\eps|
\eta)$. 
\end{quotation}
\begin{example}\label{example:Classical}
Let $\cF_{N,\eps}$ be the class of signals with at most $N\eps$
non-zero entries (in expectation) and consider AMP with soft
thresholding $\esoft(\,\cdot\,;\tau)$. Then the above formula states 
that reconstruction will succeed if $\delta>M(\eps|\Soft)$ and
fail for $\delta>M(\eps|\Soft)$. This result was first proved in
\cite{DMM09} to follow from  state evolution. State evolution  was subsequently
established as a rigorous tool in \cite{BM-MPCS-2011}.

The same paper \cite{DMM09} studied AMP with positive soft
thresholding and showed that it succeeds for  $\delta > M(\eps|\SoftPos)$,
AMP with capping, proving that it succeeds for  $\delta > M(\eps|\Cap)$.

Appendix \ref{Appendix:classical} spells out how these existing results fall under the aegis of Eq.~(\ref{generalPT}).\eex
\end{example}

\emph{Comparison to $(\rho,\delta)$ phase diagrams.} In prior literature on phase transitions in compressed sensing,
\cite{DoTa10,Do,DoTa08,BlCaTa11,DMM09,DMM-NSPT-11}, the authors
considered a different phase diagram, based on variables $\delta$ and $\rho = \eps/\delta$.
 The relation $\eps = \rho \delta$ makes for a 1-1 relationship between the diagrams,
 so all information in the two diagrams can be presented  in either format.

\subsection{This Paper}

Our aim in this paper is to show that formula 
(\ref{generalPT}) is correct in settings extending far beyond the
three cases just mentioned in Example \ref{example:Classical}. We lay out several
denoising problems, and in each one verify the
general formula.  This requires in each case $(a)$ calculating the
minimax MSE for a problem of statistical decision theory;
$(b)$ implementing AMP for compressed sensing with the given
denoising family; and $(c)$ verifying empirically 
that the phase transition does indeed occur at the precise
sparsity/undersampling tradeoff indicated by the formula.

In particular, we consider the following denoising tasks, and
corresponding compressed sensing problems.
\begin{description}
 \item[Firm Shrinkage.]
Again we consider the class od sparse vectors $\cF_{N,\eps}$
but instead of soft-thrresholding, we use the \emph{firm shrinkage}
denoiser $\efirm(\,\cdot\,;\tau)$. This is again a separable denoiser with two tuning parameters
$\tau=(\tau_1,\tau_2)$ with $\tau\in\Theta\equiv\{(\tau_1,\tau_2):\;
0\le\tau_1<\tau_2<\infty\}$. It acts on each coordinate by setting $\efirm(y;\tau)=0$
for $|y| < \tau_1$,  $\efirm(y;\tau) =y$ for $|y| >
\tau_2$ and interpolating linearly.

Denoting by $M(\eps|\Firm)$ the associated asymptotic minimax MSE, we
will show that $M(\eps | \Firm) < M(\eps |\Soft) $ strictly.
By verifying the general formula, we show that the
phase transition curve for optimally-tuned AMP firm shrinkage is slightly better than the phase transition
for  optimally tuned AMP soft shrinkage.  
 \item[Minimax Shrinkage.]
For the same class of sparse vectors $\cF_{N,\eps}$.
we consider the separable  denoiser $\eta$ applies coordinatewise shrinkage 
using a \emph{minimax} shrinkage. In other words
implicitly we are optimizing the mean square error over 
$\Theta \equiv \{ \mbox{ all scalar nonlinearities } \}$. We calculate
the minimax MSE function $M(\eps | \Minimax)$,
and  show that $M(\eps | \Minimax) < M(\eps | \Firm) $ strictly. By
verifying the general formula we show  that
 the phase transition curve for  AMP minimax shrinkage is slightly better than the phase transition
for  both AMP soft or firm shrinkage. 
 \item[Block Thresholding.] 
Here we consider the  class of block sparse vectors $\cF_{N,\eps,B}$
(see Section \ref{sec:blockThresh} for a formal definition).
 We use two block-separable denoisers: 
either block
soft thresholding (for block length $B$, the $B$-variate nonlinearity
 obeys $\eta_{B,\lambda}(y) = y \cdot ( 1 - \|y\|_2/\lambda)_+$) or block
 James-Stein denoiser.
 We will compute the minimax MSE function $M_B(\eps | \Blocksoft)$,
 and bound the minimax MSE function $M_B(\eps |\JamesStein)$. 
We will verify that  the phase transition curve for optimally-tuned
AMP with  block-separable denoisers
follows the general formula.
 \end{description}
Notice that, as demonstrated numerically in \cite{DMM09}, and proved
in \cite{BayatiMontanariLASSO} in the case of Gaussian sensing
matrices, soft-thresholding AMP reconstruction coincides with LASSO
reconstruction (in the large system limit). By the above results,
firm-shrinkage AMP and minimax AMP both outperform LASSO
reconstruction. Correspondingly, it can be argued that blocksoft
thresholding AMP coincides asymptotically with group LASSO, and hence 
James-Stein AMP outperforms the latter.

In all of the above examples, the denoisers are coordinatewise or
at least blockwise separable.
We next consider examples where the denoiser has more subtle structure.
We find that formula (\ref{generalPT}) applies more generally.
\begin{description}
 \item[Monotone Regression.] We consider the class $\cF_{N,\eps,\mono}$
   of vectors that are monotone with at most $N\eps$ points of increase.
As denoiser, we use the least-squares projection 
$\eta$ onto the cone of monotone increasing functions.  
 \item[Total variation minimization.] 
We consider the class $\cF_{N,\eps,TV}$
   of vectors that have at most $N\eps$ points of change.
The denoiser $\eta$  minimizes the residual sum of squares penalized by $\tau$ times the
 total variation of the signal.  
\end{description}
In these cases, evaluating the asymptotic minimax MSE is more
challenging
than for separable denoisers and simpler classes of signals. 
Nevertheless, we will show that it can be done quite explicitly.
We find well-defined phase transitions for AMP reconstruction, precisely at the location
predicted by the general formula  (\ref{generalPT}).
%
%
\subsection{Contributions}

We list eight contributions, beginning with the two most obvious:
\begin{enumerate}
\item \emph{Application of the AMP framework to a wider range of
    shrinkers $\eta(\, \cdot\,)$.} 
We implement and study AMP algorithms that don't correspond
to $\ell_1$ penalization (e.g. the Firm and Minimax scalar shrinkers)
and also  that don't correspond to scalar separable nonlinearities:
both the block separable case and the general non-separable cases.
\item \emph{A formula for phase transitions of AMP algorithms. } 
We confirm that formula (\ref{generalPT}) accurately describes the sparsity-undersampling
tradeoff under which AMP algorithms successfully recover a sparse
structured signal from underdetermined measurements. We prove that this
relation follows from the state evolution formalism.
\item  \emph{A formula predicting the phase transitions of many convex optimization problems. }
Much work on compressed sensing
establishes the possibility of recovery under sparsity by solving convex optimization problems.
Unfortunately,considerable
 work was required to obtain sharp phase transition results for \emph{one} convex optimization algorithm: $\ell_1$ minimization
\cite{Do,DoTa05,DoTa08,DoTa10}. The arguments needed to attack --for example--
the block-sparsity case seemed to be quite different
\cite{StojnicMMV10}. 

As demonstrated in \cite{DMM-NSPT-11} and proved in
\cite{BM-MPCS-2011} in the case of the LASSO, there exists a correspondence between convex
optimization methods and specific AMP algorithms. We will show
that this correspondence is considerably more general.
This provides  a unified approach which yields
sharp phase transition predictions in numerous cases.
\item \emph{Reconstruction approaches not based on convex penalization, with sharp guarantees.}
We introduced three new AMP algorithms, 
based on Firm, Minimax, and James-Stein shrinkage,
which do not correspond to any obvious convex penalization. These
methods  have better phase transitions
than the corresponding convex optimization problems,
in their domains (e.g. Firm and Minimax outperform $\ell_1$ minimization, while James-Stein outperforms block soft shrinkage for large $B$).
We show that these algorithms are in correspondence with penalized
least square problems, but
that the implied penalties are nonconvex\footnote{Fornasier and Rauhut \cite{FornasierRauhut} show
  that some denoisers corresponding to  non-convex penalties can be
  implemented via convex optimization by adding suitable auxiliary
  variables. This is the case, for instance, for firm thresholding
  denoising. Unfortunately, the same method does not apply --in
  general-- to the
  compressed sensing problem (i.e. for non-diagonal matrices $A$).}. Nevertheless, AMP
appear to converge to the correct solution with high probability, as
long as the undersampling ratio is above the phase
transition boundary.
 In the interior of the success phase,
these methods typically converge exponentially fast.
\item \emph{Limited benefit of nonconvex penalization for ordinary sparsity.}
Within the class of scalar separable AMP algorithms, the best achievable
phase transition is obtained by the minimax shrinker. Unfortunately the 
improvement in the transition is relatively small.
\item \emph{Existence of algorithms for the block sparse case approaching ``ideal'' behavior}.
Most attention in the group sparse case concerns block soft
thresholding and the corresponding $\ell_2-\ell_1$ penalized regression, a.k.a. group
LASSO \cite{YuanLin}.
We show here that the phase transitions for block thresholding do not tend as $B \goto \infty$
to the ideal transition, i.e. that compressed sensing reconstruction
is possible as soon as $\delta>\eps$ (i.e. from, roughly, as many measurements
as nonzeros).
On the other hand, we show that positive-part James-Stein shrinkage
\emph{does} tend to such an ideal limit.
\item \emph{Identification of combinatorial geometry phase transitions
    with minimax Mean-Square Error.}
The phase transitions for $\ell_1$ optimization, and for positivity-constrained $\ell_1$ optimization
are determined by combinatorial geometry, see \cite{DoTa08}.
By our general formula (\ref{generalPT}), these transitions are
the same as the minimax MSE in problems of scalar denoising.
\item \emph{Calculation of the minimax MSE of monotone regression and total variation denoising}.
We are not aware of any previous work computing the minimax MSE of these denoising procedures
under the condition of $\eps$-sparse first differences. We prove here
a characterization for each of these cases
and show that it  agrees with the phase transition of both AMP and convex optimization algorithms.
\end{enumerate}

A conjectures flow naturally from this work:
\begin{itemize}
 \item[] \emph{State Evolution accurately describes the behavior of a wide range of AMP algorithms, for large system sizes $N$}.
State evolution is a formalism that allows to characterize the
asymptotic behavior of AMP as the number of dimension tend to infinity
\cite{DMM09}. We show in Section \ref{sec:DerivePT} that the general
relation (\ref{generalPT})  can be proved by assuming state evolution
to hold.

In the case of separable denoisers, under suitable regularity conditions, the correctness of state evolution
as a description of AMP is proved by \cite{BM-MPCS-2011}.  
Since formula (\ref{generalPT}) is apparently successful beyond the
separable case, it is natural to conjecture that state evolution applies much more generally
than to the cases proven so far. 
\end{itemize}
Our study supports the general conclusion that AMP provides a general
tool in compressed sensing, that is applicable beyond simple sparse signal. If one knows that a certain
shrinker is appropriate for denoising a certain type of signal, then the corresponding
AMP algorithm provides an efficient reconstruction method for the
associate  compressed sensing problem. The denoising  minimax MSE then
maps to the sparsity undersampling tradeoff.

An interesting research direction is the study of the noisy linear
model  $y=Ax_0+w$, whereby $w$ is a noise vector
(e.g. $w\sim\normal(0,\sigma^2\id_{m\times m})$). In analogy
\cite{DMM-NSPT-11}, we expect reconstruction to be stable with respect
to noise for $\delta < M(\eps|\eta)$ and instable for $\delta>M(\eps|\eta)$.

\subsection{Related literature}
\label{sec:Related}

Approximate message passing algorithms for compressed sensing
reconstruction were introduced in \cite{DMM09}. They were largely
motivated by the connection with message passing algorithms in
iterative decoding systems \cite{RiU08}, and with mean field methods
in statistical physics \cite{MezardMontanari} (in particular the cavity
method and TAP equations). We refer to \cite{DMM-NSPT-11} for a
discussion of these connections.

The original AMP framework \cite{DMM09,DMM_ITW_I} included 
iterations of the form defined in Eqs.~(\ref{AMPA}), (\ref{AMPB}), (\ref{AMPC}) 
whereby the denoiser is separable.
While this covers the $\Firm$ and $\Minimax$ shrinkage rules studied
in this paper, it did not include the various non-separable denoisers we discuss below,
namely the block, monotone and total variation denoisers.  Further,
in \cite{DMM09}, the phase transition behavior was validated
numerically only for $\Soft$, $\SoftPos$ and $\Cap$ denoisers, that
are in correspondence with well-studied convex optimization problems.
The extension to a noisy linear model $y=Ax_0+w$, with $w\in\reals^n$ 
a vector of iid random entries was carried out in \cite{DMM-NSPT-11}.
We also refer to \cite{MontanariChapter} for an overview of this work.

Several papers investigate generalizations of the original framework
put forward in \cite{DMM09}.  The paper \cite{BM-MPCS-2011} defines a
general class of approximate message passing algorithms for which the
state evolution  was proved to be correct.  This include in
particular all separable Lipschitz-continuous
denoisers. Generalizations of this result were proved in
\cite{BLM-Universality,JM-BlockCase}.
Notice that all the separable denoisers treated in this papers are
Lipschitz continuous with the exception of hard thresholding. 
While the last case is not covered by \cite{BM-MPCS-2011}, we expect
 state evolution to hold for hard thresholding AMP as well, by a suitable
approximation argument.

Rangan
\cite{RanganGAMP} introduces a class of generalized approximate
message passing (G-AMP) algorithms that cope with --roughly-- two
extensions of the basic noisy linear model. First, the noisy
measurement vector $y$ can be a non-linear (random) function of the  
noiseless measurement $Ax_0$. Second, each of the `coordinates' of
$x_0$ can itself be a --low dimensional-- vector. Interesting
applications of  this framework were developed in
\cite{RanganQuantized,SchniterOFDM}. Let us notice that G-AMP does not
cover any of the non-separable cases treated here (even the block
sparse example), and hence provides a generalization in an
`orthogonal' direction.

In a parallel line of work, Schniter  applied AMP to a number of
examples in which the signal $x_0$ has a structured prior
\cite{SchniterTurbo,Schniter-NonUniform-2010,SchniterTree}.
Inference with respect to the prior is carried out using belief
propagation, and this is combined with AMP to compute \emph{a
  posteriori} estimates. This type of application fits within the
class of problems studied here, by choosing the denoiser $\eta_t$
in Eq.~(\ref{AMPB}) be given by the appropriate conditional expectation
with respect to the signal prior. 
Note however that the general scheme provided by 
Eqs.~(\ref{AMPA}), (\ref{AMPB}), (\ref{AMPC}) encompasses cases in
which the denoiser is not the Bayes estimator for a specific prior.

A special case of known prior is the one in which $x_0=\bx\sim \nu_N$
is
distributed according to the (known) product measure
$\nu_N=\nu\times\cdots\times\nu$
(i.e. the coordinates of $\bx$ are iid with known distribution
$\nu$). The fundamental limits for compressed sensing  reconstruction
were established in \cite{WuVerdu}. The natural AMP algorithm uses in
this case a posterior expectation denoiser \cite{DMM_ITW_I}. It was proved in 
\cite{OurSpatial} that, for suitable sensing matrices with
heteroscedastic entries, this approach achieves the fundamental limits
of  \cite{WuVerdu} (this approach was put forward in
\cite{KrzakalaEtAl} on the basis of a statistical
physics argument).  This case fits within the general
philosophy of the present paper whereby the class $\cF_{N,\eps}$
consists of a single distribution, namely $\cF_{N,\eps}=\{\nu_N\}$. However, we
prefer not treating this example in the present paper because it is a degenerate
case, and the fact that $\cF_{N,\eps}$ is not scale invariant leads to
some technical differences. We refer instead to \cite{OurSpatial}.

Statistical physics methods were also used to study $\ell_1$-based
reconstruction methods in \cite{RanganFletcherGoyal,KabashimaTanaka}.

Maleki, Anitori,  Yang and Baraniuk \cite{MalekiComplex} used methods analogous to the ones 
developed here to study phase transitions for compressed sensing with
complex vectors.
This is closely related to the block-separable
setting considered in Section \ref{sec:blockThresh} (there is however
some difference in the structure of the sensing matrix).

Structured sparsity models are studied from a different
point of view in \cite{ModelCS,CevherMRF,CevherReview}. Those works focus on
deriving  sparsity models that capture a variety applications, and of convex relaxations
that promote the relevant sparsity patters. Reconstruction guarantees
are proved under suitable `isometry' assumptions on the sensing
matrix.

Closer to our approach is a recent series of papers 
\cite{ChandrasekaranRecht,RechtNowakTight,RechtNowakTight},
considering general classes of structured signals under random
measurements. Let us emphasize two important differences with respect
to our work. First, these papers only deal with convex reconstruction
methods, while we shall analyze several approaches that are not
derived from convex optimization and
demonstrate improvements. Second, they establish reconstruction
guarantees using concentration-of-measure arguments, while we propose
\emph{exact} asymptotics (essentially based on weak convergence), which
enables us to unveil the key relation (\ref{generalPT}) between denoising and the compressed
sensing phase transition.

\section{Scalar-separable denoisers}
\label{sec-Scalar}

In this section we study scalar-separable denoisers,
cf. Eq.~(\ref{eq:Separable}), that further satisfy the scaling
relation (\ref{eq:ScalingRelation}). Unless stated otherwise, we will assume that
signals belong to the  simple sparsity class introduced
in Eq.~(\ref{eq:SimpleSparse}), to be denoted as $\cF_{N,\eps}$.

\subsection{Minimax MSE of a separable denoiser}
\label{sec-Scalar-Minimax}

As mentioned in the previous section, the computation of the minimax
MSE is greatly simplified for separable denoisers. We state and prove
the following elementary result in greater generality than necessary
for this section. (In particular $\cF_{N,\eps}$ is here a general
family of probability distributions.)
\begin{lemma}\label{lemma:ScalarMSE}
 Let $\cF_{N,\eps}\subseteq
\cP(\bR^N)$ be any family of probability distributions satisfies the following conditions: $(i)$ If
$\nu_1\in\cF_{1,\eps}$, then defining $\nu_N\equiv \nu_1\times\cdots\times
\nu_1$ ($N$ times), we have $\nu_N\in\cF_{N,\eps}$; $(ii)$ Viceversa, if
$\nu_N\in\cF_N$, then letting $\nu_{N,i}$ denote the $i$-th marginal
of $\nu_N$, we have $\onu_N\equiv N^{-1}\sum_{i=1}^N\nu_{N,i}\in\cF_{1,\eps}$. 

Then, for any separable denoiser $\eta$, and for any $N$,
\begin{eqnarray*}
M(\eps|\eta) = M(\cF_{N,\eps}|\eta) = M(\cF_{1,\eps}|\eta)\, .
\end{eqnarray*}
\end{lemma}
\begin{proof}
Fix $\tau\in\Theta$ and define, for $\bz\sim\normal(0,\id_{N\times
  N})$,
\begin{align*}
M(\cF_{N,\eps} | \eta,\tau)\equiv\frac{1}{N} \sup_{\nu_N \in \cF_{N,\eps}}
\E_{\nu_N} \big\{\big\| \bx - \eta(\bx+\bz; \tau) \big\|_2^2\Big\} \, .
\end{align*}
The lemma then follows immediately if we prove that, for any $N$,
$M(\cF_{N,\eps} | \eta,\tau) = M(\cF_{1,\eps} | \eta,\tau)$. In order
to prove the last statement, first notice that, by property $(i)$:
\begin{align*}
M(\cF_{N,\eps} | \eta,\tau)&\ge\frac{1}{N} \sup_{\nu_1 \in \cF_{1,\eps}}
\E_{\nu_1\times\cdots\times \nu_1} \big\{\big\| \bx - \eta(\bx+\bz;
\tau) \big\|_2^2\Big\} \\
&=\sup_{\nu_1 \in \cF_{1,\eps}}
\E_{\nu_1} \big\{\big\| X_1- \eta(X_1+Z;
\tau) \big\|_2^2\Big\}  = M(\cF_{1,\eps} | \eta,\tau)\, .
\end{align*}
The proof is finished by property $(ii)$, 
since
\begin{align*}
M(\cF_{N,\eps} | \eta,\tau)&=\frac{1}{N} \sup_{\nu_N \in \cF_{N,\eps}}
\sum_{i=1}^N\E_{\nu_N} \big\{\big\| X_i- \eta(X_i+Z;
\tau) \big\|_2^2\Big\} \\
& =  \sup_{\nu_N \in \cF_{N,\eps}}
\E_{\onu_N} \big\{\big\| X- \eta(X+Z;
\tau) \big\|_2^2\Big\}  \le M(\cF_{N,\eps};\eta,\tau)\, .
\end{align*}
\end{proof}
This lemma reduces the problem to solving a minimax scalar estimation problem.
This problem was characterized before for soft thresholding $\eta = \esoft$,
positive soft thresholding $\eta=\esoftpos$,
and also hard thresholding $\eta(y;\tau) = y 1_{\{|y| > \tau\}}$ 
\cite{DJHS92,DJ94a}. Plots of the minimax soft threshold $\tau^*(\eps)$ and
the minimax MSE are available in \cite{DMM09,DMM-NSPT-11}. Such plots also appear later in this paper
as baselines for comparison of  interesting  new families, namely Firm
and  Minimax shrinkage. 

\subsection{Firm shrinkage}

A frequently-voiced criticism of $\ell_1$ minimization and soft thresholding is
the tendency to shrink large values by more than warranted.
In the mid 1990's, \emph{firm} shrinkage was introduced  to correct this tendency
by Gao and Bruce \cite{GaBr97}. As suggested by the  name, this denoiser  is intermediate
between soft and hard thresholding: it is continuous like soft
thresholding, but does not shrink large
values, like hard thresholding. 
Formally, for $\tau=(\tau_1,\tau_2)\in\Theta$,
$\Theta\equiv\{(\tau_1,\tau_2):0\le\tau_1<\tau_2<\infty\}$, we have
\[
   \efirm(y ; \tau_1, \tau_2) =  \left\{ \begin{array}{ll} 
                                                                0 & |y| < \tau_1 ,\\
                                                                \sign(y)\,
                                                                (|y|-\tau_1)\,\tau_2/(\tau_2-\tau_1)
                                                                &
                                                                \tau_1
                                                                < |y|
                                                                <
                                                                \tau_2
                                                                ,\\
                                                                y & |y| > \tau_2 .\\
                                              \end{array} \right. 
\]
Soft and hard thresholding can be recovered as limiting cases:
\begin{equation} \label{eq:SS-limits}
    \esoft(y;\tau_1) = \lim_{\tau_2 \goto \infty} \efirm(y;\tau_1,\tau_2) , \quad
    \ehard(y;\tau_1) = \lim_{\tau_2 \goto \tau_1} \efirm(y;\tau_1,\tau_2) .
\end{equation}
Lemma \ref{lemma:ScalarMSE} yields the following formula for the
minimax MSE of firm shrinkage:
\begin{eqnarray}
    M(\eps | \Firm ) = M( \cF_{1,\eps} | \efirm ) =  \inf_{0 \le
      \tau_1 < \tau_2 } \sup_{\nu \in \cF_{1,\eps}}
\E_{\nu}\big\{[X-\efirm(X+Z;\tau_1,\tau_2)]^2\big\}\, .\label{eq:FirmMMAX}
\end{eqnarray}
Notice that the supremum over $\nu$ is necessarily achieved at a
probability distribution of the form $\nu = \nu_{\eps,\mu} =
(\eps/2)\delta_{\mu}+(1-\eps)\delta_0 +(\eps/2)\delta_{-\mu}$, for
some $\mu\in [0,\infty]$ (indeed
the risk is  an even function of $X$, and the extreme points of the
even subset of $\cF_{1,\eps}$ take this form). Explicit expressions
under this distribution are given in \cite{GaBr97}.  
The computation of $M(\eps|\Firm)$ is therefore reduced to the
calculus problem of finding the saddle point of
$f(\mu,\tau_1,\tau_2)\equiv \E_{\nu_{\eps,\mu}}\big\{[X-\efirm(X+Z;\tau_1,\tau_2)]^2\big\}$.
This can be efficiently solved numerically, yielding the minimax risk
and the optimal thresholds $\tau_1^*(\eps)$ and $\tau_2^*(\eps)$. It
follows immediately from the definition  that
 $0 \leq M(\eps | \Firm )  \leq 1$, that $M$ is monotone
 increasing with $\eps$, that  $M(\eps|\Firm) \goto 1$ as $\eps \goto 1$ and that $M(\eps|\Firm) \goto 0$ as $\eps\goto 0$.

Figure \ref{fig:MMSEscalar}  and Table \ref{table:MMSEscalar} show the minimax MSE for firm shrinkage
as resulting from this calculation.
The  figure also
shows similar results for soft and hard thresholding, for comparative purposes. 
Over the range presented,
the minimax MSE for firm thresholding is \emph{strictly smaller} 
than the MSE for hard or soft thresholding.  Namely, over this range of $\eps$,
\[
    M(\eps | \Firm) < M(\eps| \Soft) < M(\eps | \Hard)\, .
\]
This validates the criticisms of soft thresholding, which is often
said to shrink large values too heavily\footnote{Note, however, that
  the use of hard thresholding instead of soft thresholding leads to a
larger worst case mean square error.}.

Figure \ref{fig:tauScalar} shows the minimax thresholds. 
At least for $\eps < 1/3$ we see  clearly that $\tau_1^*(\eps) < \tau_2^*(\eps) < \infty$,
so  firm thresholding is preferred over the  limiting cases of hard and soft thresholding\footnote{These are numerical results.
It is an open question whether, for $\eps>1/3$  the minimax firm threshold have parameter $\tau_2(\eps) = \infty$
reducing it to  soft thresholding.}. Figure \ref{fig:Nonlinearities}
shows the corresponding minimax denoisers for specific values of $\eps$.
Finally, Figure \ref{fig:LFMu} plots the minimax value of $\mu$ as a
function of $\eps$ (corresponding to the minimax probability
distribution $\nu_{\eps,\mu}$).

\begin{table}
\begin{center}
\begin{tabular}{|l|l|l|l|l|}
\hline
$\eps$ & $M(\eps|\Hard)$ & $M(\eps|\Soft)$ & $M(\eps|\Firm)$ & $M(\eps|\Minimax)$ \\
\hline
0.010 & 0.0729 & 0.0612 & 0.0552 & 0.0533 \\    
 0.025 & 0.1547 & 0.1231 & 0.1137 & 0.1093 \\    
 0.050 & 0.2676 & 0.2039 & 0.1921 & 0.1841 \\    
 0.100 & 0.4497 & 0.3288 & 0.3165 & 0.3025 \\    
 0.150 & 0.5960 & 0.4279 & 0.4171 & 0.3983 \\    
 0.200 & 0.7161 & 0.5111 & 0.5024 & 0.4802 \\    
 0.250 & 0.8141 & 0.5829 & 0.5763 & 0.5516 \\   
\hline
\end{tabular}
\caption{Minimax MSE of various separable denoisers applied to sparse
  vectors from the class $\cF_{N,\eps}$, across sparsity levels}
\label{table:MMSEscalar}
\end{center}
\end{table}

\begin{figure}
\begin{center}
\includegraphics[height=3in]{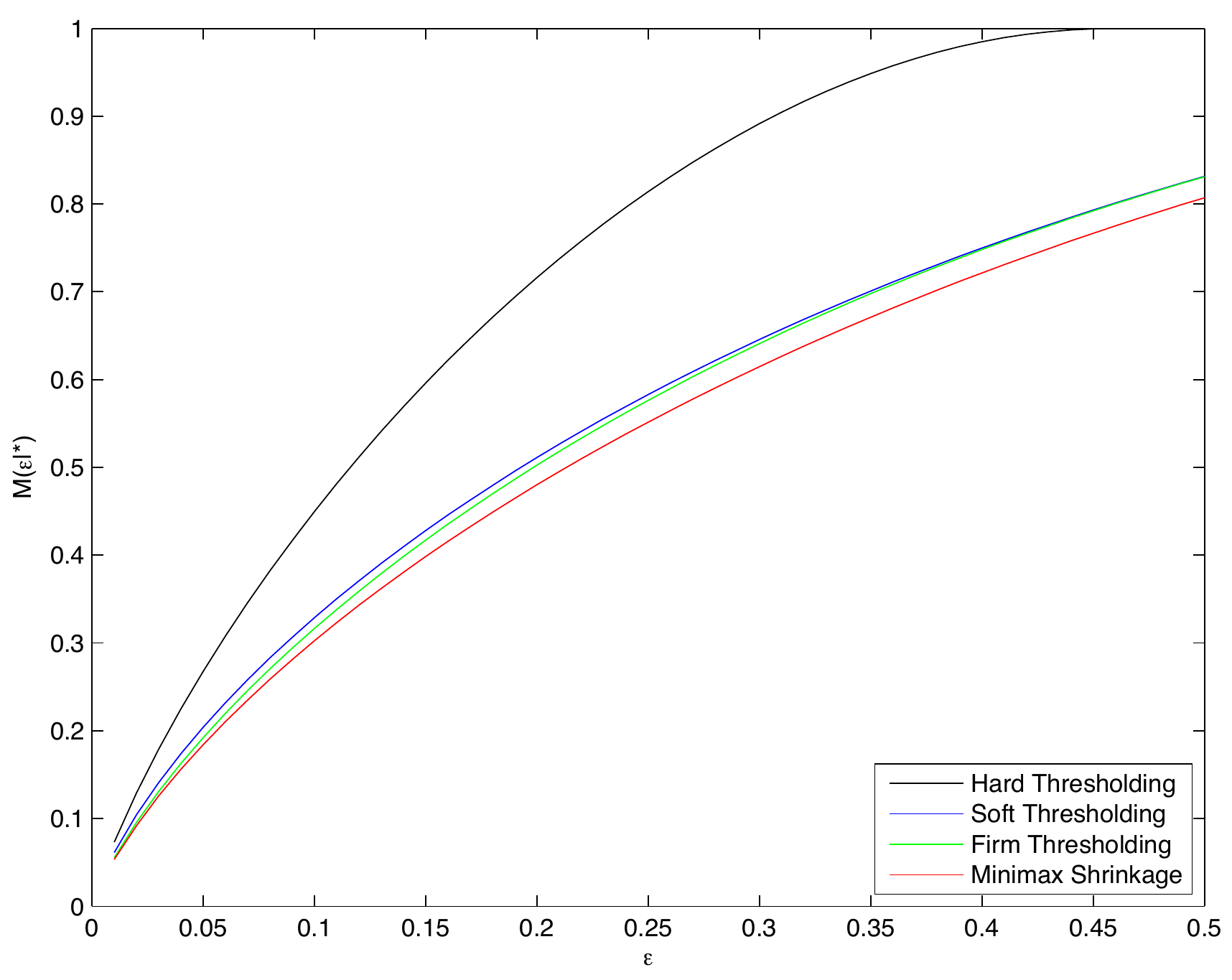}
\caption{Minimax MSE of various separable denoisers, as a function of sparsity parameter $\eps$.
The minimax MSE of firm thresholding is very close to that of soft
thresholding for $\eps > 1/3$.
From top to bottom the curves refer to hard thresholding, soft
thresholding, firm thresholding and minimax denoiser.}
\label{fig:MMSEscalar}
\end{center}
\end{figure}

\begin{figure}
 \begin{center}
\includegraphics[height=3in]{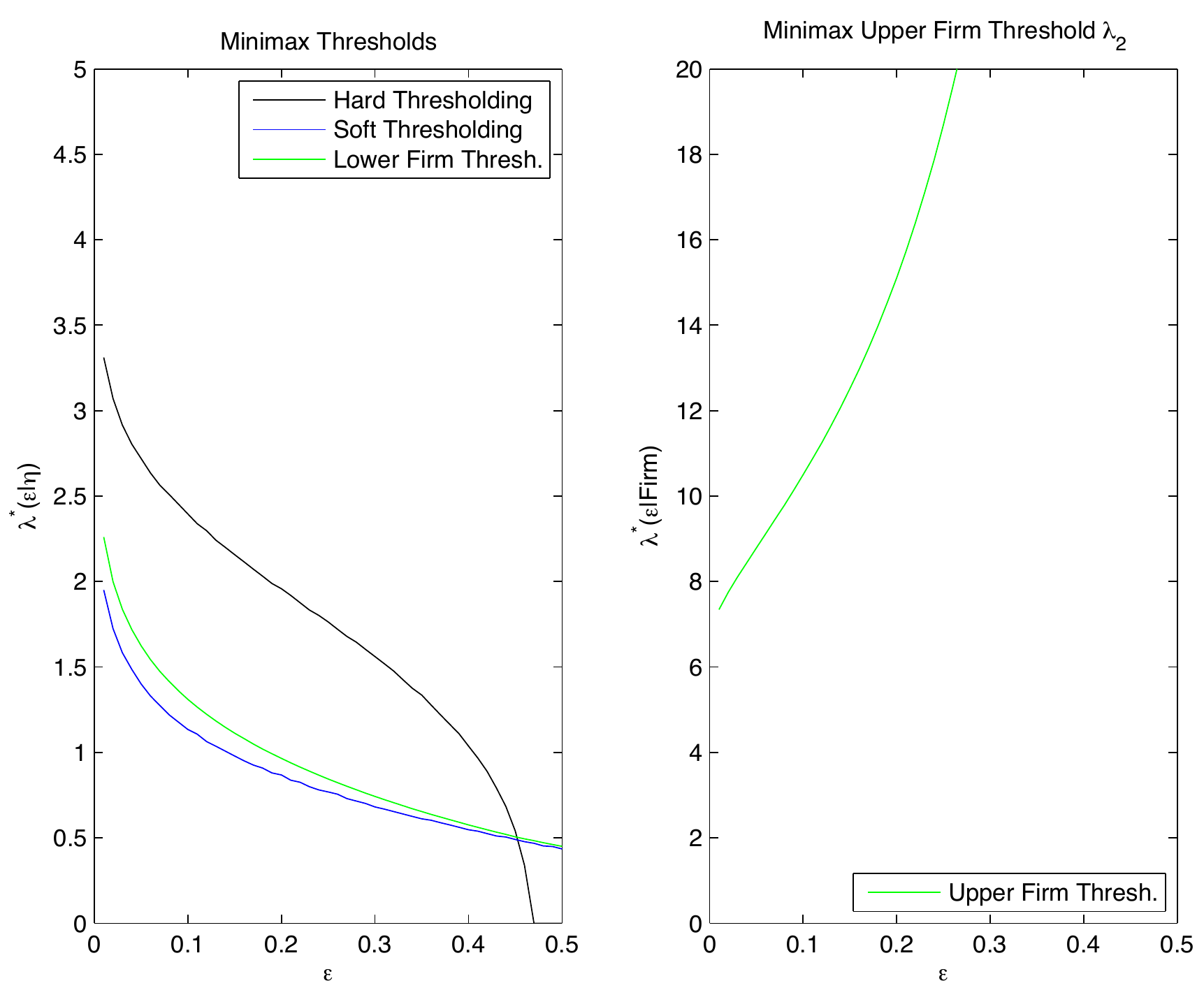}
\caption{Minimax thresholds of various separable denoisers, as a function of sparsity $\eps$. The sudden drop in the
value of the hard threshold to zero near $0.45$ coincides with the value of the minimax MSE
in Figure \ref{fig:MMSEscalar} reaching 1. The numerical approximation to the minimax lower firm threshold approaches 
the minimax soft threshold as $\eps$ increases beyond $0.3$, and the upper firm threshold increases rapidly.}
\end{center}
\label{fig:tauScalar}
\end{figure}

\begin{figure}
 \begin{center}
\includegraphics[height=3in]{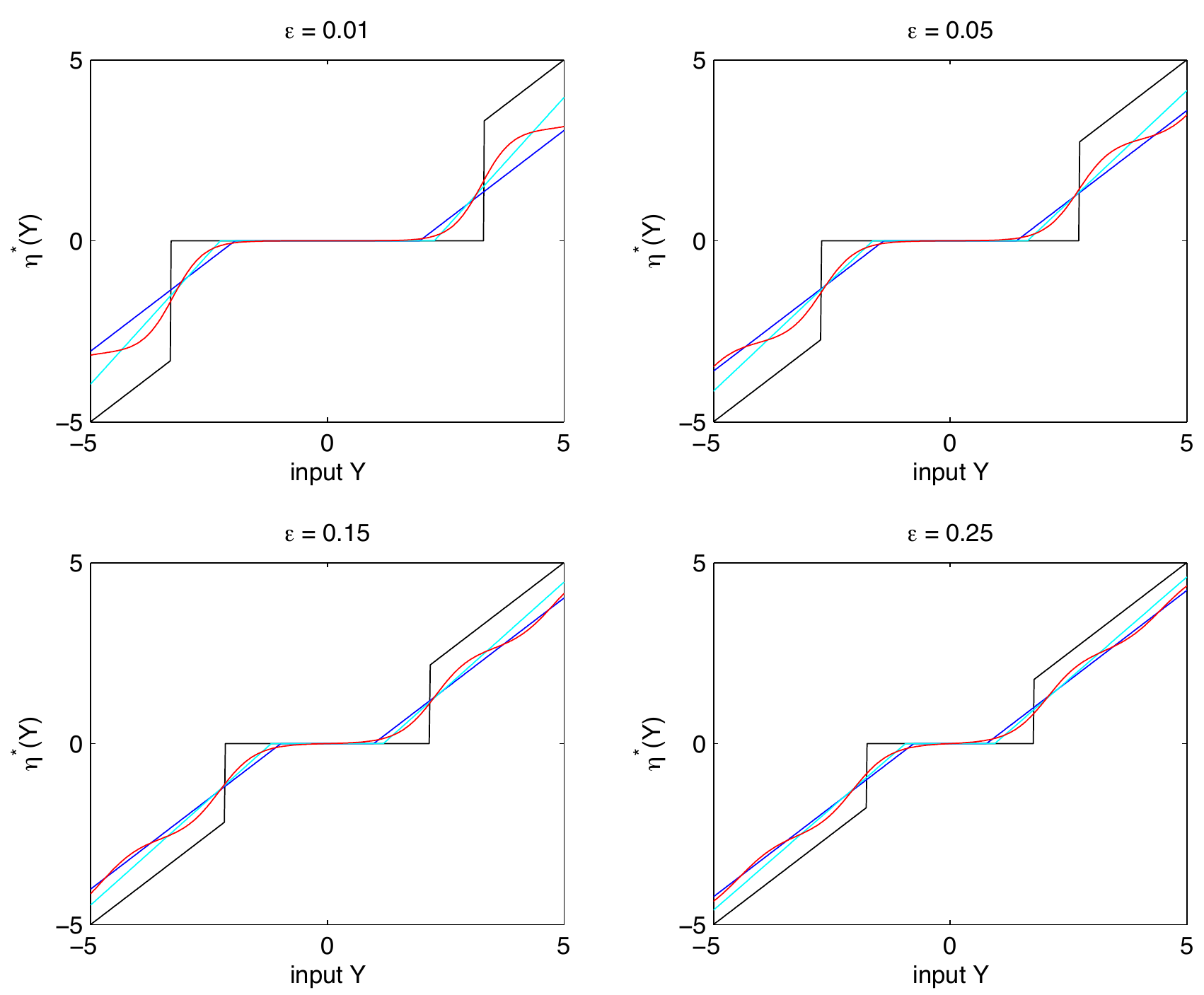}
\caption{Minimax scalar denoisers of various types, at sparsity $\eps
  =0.01$, $0.05$, $0.15$, $0.25$.
Black: hard thresholding;  Blue: soft; Aqua: firm;  Red: minimax shrinkage.}
\end{center}
\label{fig:Nonlinearities}
\end{figure}

\begin{figure}
 \begin{center}
\includegraphics[height=3in]{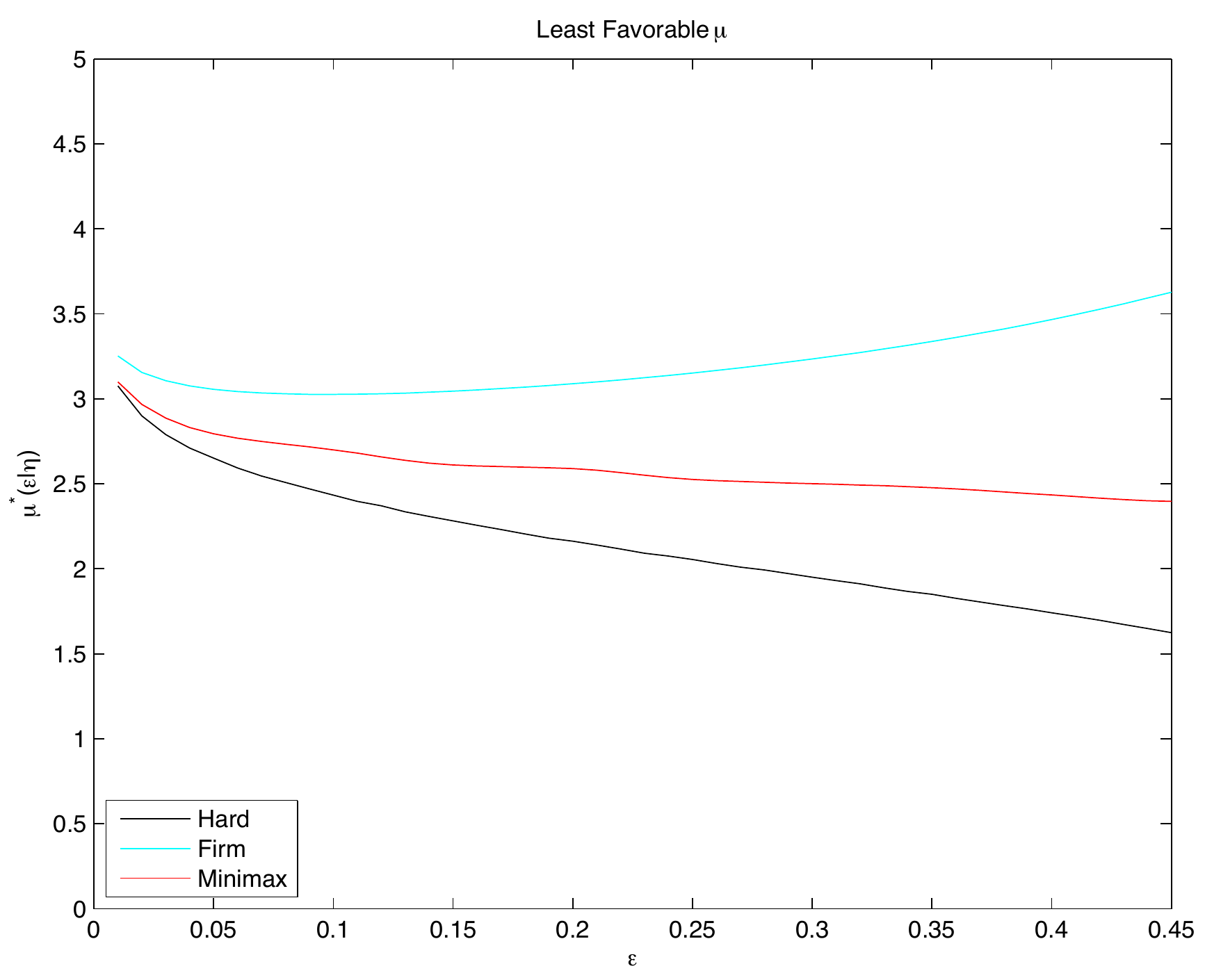}
\caption{Least favorable $\mu$ for various separable denoisers, as a
  function of sparsity $\eps$. 
From top to bottom the curves refer to firm thresholding, minimax
denoiser and hard thresholding.
Soft thresholding has minimax $\mu = \pm \infty$, not shown.}
\end{center}
\label{fig:LFMu}
\end{figure}

\subsection{Minimax shrinkage}
\label{sec:MMAX}

The previous example showed that a parametric family of shrinkers
can improve on soft thresholding, and hence improve the predicted phase transition
according to (\ref{generalPT}). The ultimate improvement one could make in this direction
is to use the \emph{globally} minimax nonlinear shrinker. This is the
separable denoiser $\eta$ that
is minimax not within some parametric family, such as the soft
thresholding or the broader  firm thresholding family,
but minimax over \emph{all} measurable nonlinearities
$\eta:\bR\to\bR$.  While this notion might appear somewhat abstract,
it can be in fact implemented in practice as illustrated in Figure
\ref{fig:Nonlinearities}, that present plots of the more familiar
denoisers (hard, soft, and firm) together with the minimax denoiser.

Formally,  
let $\Theta\equiv \meas(\reals)$ be the set of all measurable functions
$\tau:\reals\to\reals$, for such a $\tau$,
set $\eall(x ; \tau) \equiv \tau(x)$.  The minimax MSE over this class is
\begin{eqnarray}
M(\eps | \Minimax)  &=&  \inf_{\eta\in\meas(\reals)} \sup_{\nu \in \cF_{1,\eps}}
\E\big\{[X-\eta(X+Z)]^2\big\} \, . \label{eq:AllMMAX}
\end{eqnarray}
The calculation of this quantity uses a variety of ideas from minimax
decision theory, developed through several papers
\cite{Bi81,CaSt81,BiCo83,DJHS92,Jo94a,Jo94b,DJ94a}:
details are given in Appendix \ref{Appendix:minimax}.
A key point of this computation  is  the characterization of the
minimax nonlinearity as the minimal MSE Bayes
rule (that is, the conditional expectation) for the so-called
least-favorable prior. The least-favorable prior is the solution of Mallows' classical Fisher information problem \cite{Mallows78},
for which we compute numerical upper and lower bounds that coincide
within the stated precision.

Table \ref{table:MMSEscalar} and Figure \ref{fig:MMSEscalar} present  numerical values
associated with the solution of the minimax problem.  As expected, 
$M(\eps | \Minimax ) < M(\eps | \Firm ) \le M(\eps | \Soft)$,
i.e. optimizing over all nonlinearities yields a smaller mean square
error than soft or firm thresholding. On the other hand, Table \ref{table:MMSEscalar}
shows that the improvements are typically of size 0.01 or smaller over the range
$\eps \in (0.01,0.25)$. For very small $\eps$  it was  pointed out in \cite{DMM09}
 that \cite{DJ94a} implies
 \[
   \lim_{\eps \goto 0} \frac{M(\eps | \Minimax )}{M(\eps | \Soft)}  = 1 .
 \]
In the limit of extreme sparsity, there is nothing to be gained by completely general nonlinearities
over soft thresholding. The improvement is non-vanishing, but moderate
for $\eps$ non-vanishing.

\subsection{Empirical phase transition behavior}
\label{sec-PTExperiments-scalar}

The research hypothesis driving this paper is that Eq.~(\ref{generalPT})
describes the phase transition of AMP algorithms.  
In order to be completely explicit, we need to check the following
predictions, for each nonlinearity $\eta$ of interest
\begin{enumerate}
\item There exists a curve $\eps\mapsto\delta(\eps|\eta)$ such that
  for $\delta > \delta(\eps|\eta)$ the corresponding AMP algorithm will typically succeed
in reconstructing the unknown signal $x_0$, and  for $\delta < \delta(\eps|\eta)$ the algorithm will typically fail.
\item The curve is related to the corresponding scalar denoising
  problem  by $\delta(\eps | \eta ) = M(\eps | \eta)$.
\end{enumerate}
We now test this hypothesis for the firm and
globally minimax nonlinearities $\eta\in\{\efirm, \eall\}$.

Our experiment was conducted along  the same lines
as \cite{MaDo10,DoTa09b,DMM09,DMM-NSPT-11,BlTa11}.
We considered a range of problem sizes $N \in \{ 1000, 2000, 4000 \}$ and
a range of sparsity parameters $\eps \in \{
0.01,0.02,0.05,.10,0.15,0.20,0.25\}$, and  a grid of $\delta$ values
surrounding the predicted phase transition $\delta(\eps|\eta)$. We
ran $N_{\rm sample}= 1000$ Monte Carlo
reconstructions at each parameter combination. We declared ``success''  when the relative mean-squared error  
was below $1\%$:
\[
     \frac{ \| \hx^t - x_0\|_2^2}{ \|x_0\|_2^2} < 0.01 .
\]
We used $t = 300$ iterations of AMP\footnote{In most cases the
  mentioned convergence criterion is reached after a much smaller
  number of iterations (roughly 20) }.  We repeated a subset 
of our simulations with different requirements on $\| \hx^t -
x_0\|_2^2/\|x_0\|_2^2$ and different number of iterations, without
significant changes in the threshold location. This point is further
justified in Appendix \ref{Appendix:Convergence}. In the interest of
reproducibility, a suite of Java classes for carrying out these and
other simulations in the paper is made available as
\cite{AssemblaCode}.  

\begin{table}[hp]
\begin{center}
\begin{tabular}{rrrrrrrr}
\hline
 \multicolumn{8}{|c|}{Firm  Thresholding with Minimax Tuning}\\
    \hline
 $\eps$ & Pred & off.1000 & off.2000 & ci.1000.lo & ci.1000.hi & ci.2000.lo & ci.2000.hi \\ 
  \hline
0.01 & 0.0550 & 0.0057 & 0.0040 & 0.0054 & 0.0060 & 0.0038 & 0.0043 \\ 
  0.02 & 0.0960 & 0.0065 & 0.0039 & 0.0062 & 0.0069 & 0.0036 & 0.0041 \\ 
  0.05 & 0.1920 & 0.0064 & 0.0047 & 0.0060 & 0.0068 & 0.0043 & 0.0050 \\ 
  0.10 & 0.3170 & 0.0067 & 0.0050 & 0.0063 & 0.0071 & 0.0047 & 0.0054 \\ 
  0.15 & 0.4190 & 0.0070 & 0.0049 & 0.0065 & 0.0074 & 0.0045 & 0.0052 \\ 
  0.20 & 0.5070 & 0.0068 & 0.0052 & 0.0064 & 0.0073 & 0.0048 & 0.0055 \\ 
  0.25 & 0.5840 & 0.0069 & 0.0055 & 0.0064 & 0.0073 & 0.0051 & 0.0058 \\ 
   \hline
  \multicolumn{8}{|c|}{Minimax Denoiser}\\
  \hline
 $\eps$ & Pred & off.1000 & off.2000 & ci.1000.lo & ci.1000.hi & ci.2000.lo & ci.2000.hi \\ 
  \hline
0.01 & 0.0530 & 0.0075 & 0.0052 & 0.0072 & 0.0078 & 0.0050 & 0.0055 \\ 
  0.05 & 0.1840 & 0.0126 & 0.0099 & 0.0122 & 0.0130 & 0.0096 & 0.0102 \\ 
  0.10 & 0.3020 & 0.0183 & 0.0164 & 0.0178 & 0.0188 & 0.0160 & 0.0167 \\ 
  0.15 & 0.3980 & 0.0147 & 0.0127 & 0.0142 & 0.0151 & 0.0124 & 0.0131 \\ 
  0.20 & 0.4800 & 0.0148 & 0.0114 & 0.0143 & 0.0152 & 0.0110 & 0.0118 \\ 
  0.25 & 0.5510 & 0.0155 & 0.0124 & 0.0151 & 0.0160 & 0.0120 & 0.0127 \\ 
   \hline
    \multicolumn{8}{|c|}{Soft Thresholding with Minimax Tuning}\\
  \hline
$\eps$  & Pred & off.1000 & off.2000 & ci.1000.lo & ci.1000.hi & ci.2000.lo & ci.2000.hi \\ 
  \hline
0.01 & 0.0610 & 0.0061 & 0.0045 & 0.0058 & 0.0064 & 0.0043 & 0.0048 \\ 
  0.02 & 0.1040 & 0.0065 & 0.0052 & 0.0062 & 0.0069 & 0.0049 & 0.0055 \\ 
  0.05 & 0.2030 & 0.0089 & 0.0072 & 0.0085 & 0.0093 & 0.0068 & 0.0075 \\ 
  0.15 & 0.4270 & 0.0117 & 0.0100 & 0.0112 & 0.0122 & 0.0096 & 0.0104 \\ 
  0.20 & 0.5110 & 0.0118 & 0.0100 & 0.0113 & 0.0124 & 0.0096 & 0.0104 \\ 
  0.25 & 0.5830 & 0.0127 & 0.0106 & 0.0122 & 0.0132 & 0.0102 & 0.0110 \\ 
   \hline    
\end{tabular}
\caption{Empirical phase transition studies for AMP algorithms on the
  class of simple sparse vectors $\cF_{N,\eps}$ based
  on three separable denoisers.
At each value of $\delta,\eps$ we carried out 1000 Monte Carlo repetitions at $N=1000$, $2000$.
Column \emph{Pred} corresponds to  general formula
(\ref{generalPT}) yielding the critical value of $\delta$,  and is
thought to be accurate for large $N$ . 
The columns \emph{off.1000}  and \emph{off.2000}  report values of the
offset $\hPT$   estimated by logistic regression,
 using Eqs.~(\ref{eq:LOGIT}) and (\ref{eq:OffSetDef}).
 (Note that offsets are systematically smaller at $N=2000$ than at $N=1000$, consistent with hypothesis that they vanish as $N \goto \infty$.)
 Columns \emph{ci.1000.lo} , \emph{ci.1000.hi}, \emph{ci.2000.lo}, \emph{ci.2000.lo} give  lower and upper endpoints of formal 95\% confidence intervals
 for the offset $\hPT$.}
 \label{table:EmpiricalPT}
\end{center}
\end{table}

We proceeded to analyse the outcomes of these numerical simulations
as follows, see also Appendix \ref{Appendix:NScaling} (a similar analysis was already carried out in in
\cite{DMM09,DoTa09b}).
The simulations generated a data set, containing, for each algorithm
and each fixed $\eps$,  a list of values $\delta_i$ and empirical success fractions $\hp_i$.
The success fractions observed at $\delta > M(\eps|\eta)$ were indeed
typically better than $50\%$ and at $\delta < M(\eps|\eta)$ were typically worse than $50\%$.

To quantify this tendency, we fit a logistic regression
\begin{eqnarray}
      \log\frac{\hp_i}{1-\hp_i} = \alpha + \beta \big(\delta_i -
      \delta(\eps|\eta)\big) , \label{eq:LOGIT}
\end{eqnarray}
where $\delta(\eps|\eta)=M(\eps|\eta)$ was computed analytically 
using  ideas mentioned
earlier. The choice of the model (\ref{eq:LOGIT}) is motivated by the
observation that the success probability increases rapidly around the
phase transition, and by the common statistical use of logistic
models. Also, similar  models have been proved to be asymptotically
correct in analogous phase transition phenomena \cite{DemboMontanari}.

For each set of data corresponding to given $(N,\eps)$ and each
non-linearity, we estimate $\alpha$ and $\beta$ from the  logit fit,
leading to values $\halpha, \hbeta$. Using these quantities, we
estimate the phase transition location as the value at which the
probability $\hp$ of success is $50\%$. Using Eq.~(\ref{eq:LOGIT})
this corresponds to $\halpha  + \hbeta (\delta-\delta(\eps|\eta)) = 0$,
i.e. $\delta =\delta(\eps|\eta) -(\halpha/\beta)$.
We are therefore led to define the offset between the empirical phase
transition and the prediction $\delta(\eps|\eta) = M(\eps|\eta)$ as
\begin{eqnarray}
\hPT  \equiv - \frac{\halpha}{\hbeta}\, .\label{eq:OffSetDef}
\end{eqnarray}
In order to check the general relation provided by
Eq.~(\ref{generalPT}) we need to show that $\hPT(N,\eps)$ tends to zero as
$N$ gets large, to within the statistical uncertainty.
In Table \ref{table:EmpiricalPT} we report our results on the empirical phase
transition, confirming that indeed the offset is small and decreasing
with $N$.

A few additional remarks on these data are of interest:
\bitem
  \item We calculated formal $95\%$ confidence intervals for $\hPT$,
   indicating the tight control we have of the correct value.
  \item 
As in earlier studies \cite{DoTa09b},
  we expect that  $\hPT(N,\eps)$ tend  $0$ 
at a rate that is inversely proportional to a power of $N$. Namely
  \[
       \hPT(N,\eps)\approx  \frac{{\rm const}}{N^{\gamma}} + \mbox{sampling error},
  \]
for some $\gamma\in (0,1]$.
  Our data supports
   this relationship, with $\gamma \approx1/3$.  See Appendix \ref{Appendix:NScaling}.
  \item Denoting by $\hbeta_N$
the fitted slope coefficient at dimension  $N$, evidence that $\hbeta_N$ is
increasing with larger $N$ indicates that a sharpening of the phase transition is indeed occurring.
Appendix \ref{Appendix:NScaling} shows that $\hbeta_N \sim \sqrt{N}$ is consistent with our data.
\eitem
We refer to Appendices \ref{Appendix:Computation} and \ref{Appendix:NScaling}
for further details.

\section{Block-separable denoisers}
\label{sec:blockThresh}

We now turn to the case of block-structured sparsity,
first introducing some notational conventions.
We partition the vector $x = (x_1, x_2,\dots, x_N)$
into $M$ blocks each of size $B$. Denoting by
$block_m(x) = (x_{(m-1)B + 1}, \dots, x_{mB})$ the $m$-th block, we
hence write
\[
 x =  (block_1(x), \dots, block_M(x))\, ,
\]
with, of course, $N=MB$. 

A block-separable denoiser is a mapping $\eta(\,\cdot\, ;\tau):\reals^N\to\reals^N$
that decomposes according to the above partition:
\begin{eqnarray}
\eta(x;\tau) = \big(\eta(block_1(x);\tau), \dots, \eta(block_M(x);\tau)\big)\, ,
\end{eqnarray}
where, with an abuse of notation, we use the same symbol to denote 
the single-block denoiser $\eta(\,\cdot\,;\tau):\reals^B\to\reals^B$.
The last equation replaces Eq.~(\ref{eq:Separable}) which correspond
to the simple separable case. The above form applies to noise with
variance $\sigma^2=1$. For general variance we adopt again the
scaling relation (\ref{eq:ScalingRelation}).

We will apply these denoiser to signals from the block-sparse class
$\cF_{N,\eps,B}$ defined as follows for $\eps\in [0,1]$, $B\in
\naturals$, $M\equiv N/B$,
\begin{align}
\cF_{N,\eps,B}\equiv \Big\{\;\nu_N\in\cP(\bR^N)\, :\;\;\;
\E_{\nu_N}\big[\#\{i\in [M] : \; block_i(\bx)\neq 0\}\big]\le M\eps\; \Big\}\, .\label{eq:BlockSparseDef}
\end{align}
In words, this is the class of (random) vectors $\bx$ that have (in
expectation) at most $M\eps$ blocks different from $0$.
For simplicity, we will write $\cF_{\eps,B}$ for the $M=1$ case, $\cF_{B,\eps,B}$ 

The same simplifications described in
Section \ref{sec-Scalar-Minimax} applies, 
with obvious modifications, to the present context. 
\begin{lemma}\label{lemma:BlockMSE}
 Let $\cF_{N,\eps}\subseteq
\cP(\bR^N)$ be any family of probability distributions satisfies the following conditions: $(i)$ If
$\nu_B\in\cF_{B,\eps}$, then defining $\nu_N\equiv \nu_B\times\cdots\times
\nu_B$ ($M=N/B$ times), we have $\nu_N\in\cF_{N,\eps}$; $(ii)$ Viceversa, if
$\nu_N\in\cF_N$, then letting $\nu_{N,i}$ denote the marginal of the
$i$-th block under  $\nu_N$, we have $\onu_N\equiv M^{-1}\sum_{i=1}^M\nu_{N,i}\in\cF_{B,\eps}$. 

Then, for any block-separable denoiser $\eta$, and for any $N$
multiple of $B$
\begin{eqnarray*}
M(\eps|\eta) = M(\cF_{N,\eps}|\eta) = M(\cF_{B,\eps}|\eta)\, .
\end{eqnarray*}
\end{lemma}
The proof is omitted since it is an immediate generalization of the one of
Lemma \ref{lemma:BlockMSE}.
The class $\cF_{N,\eps,B}$ to be studied in the rest of this section
clearly satisfy the assumption of this lemma.

\subsection{Block soft thresholding}

Block-soft thresholding $\esoft(\,\cdot\,;\tau):\bR^B\to\bR^B$  is the
nonlinear shrinker defined by letting, for $y\in\bR^B$, and $\tau\in\bR_+$,
\begin{align}
\esoft(y;\tau) =   \Big( 1 - \frac{\tau}{\|y\|_2}
\Big)_+\cdot y\, ,
\end{align}
where $(z)_+ \equiv\max(z,0)$.
The case $B=1$ reduces to  traditional soft  thresholding of Example
\ref{example:SoftThresholding}.
More generally, $\esoft(y;\tau)$ shrinks its argument $y$ to $0$ if
$\|y\|_2\le \tau$ and moves it by an amount $\tau$ towards the origin
otherwise. It can also be regarded as the solution of a penalized
least squares problem, namely
\begin{align*}
\esoft(y;\tau) \equiv \arg\min_{x\in\bR^B}
\left\{\frac{1}{2}\|y-x\|^2_2 + \tau\|x\|_2\right\}\, .
\end{align*}
Block thresholding has previously been considered by  Hall,
Kerkyacharian and Picard \cite{HaKePi98} and
by Cai \cite{Cai99}
although in specific `wavelet' applications. 

In view of Lemma \ref{lemma:BlockMSE}, computing the minimax risk
reduce to solve the block minimax problem (in this section we add the subscript $B$
for greater clarity)
\begin{align}
  M_B(\eps | \Blocksoft) \equiv   \frac{1}{B}\, \inf_{\tau\in\bR_+}
  \sup_{\nu\in\cF_{\eps,B}} 
\E_{\nu}\big\{[\bx-\esoft(\bx+\bz;\tau)]^2\big\}\, ,\label{eq:MMAXBlockSoft}
\end{align}
where expectation is taken with respect to $X\sim \nu$ independent of
$\bz\sim \normal(0,\id_{B\times B})$. Notice that the condition
$\nu\in\cF_{\eps,B}$ simply amount to saying that $\nu$ is a
probability measure on $\bR^B$ with $\nu(\{0\})\ge 1-\eps$.
The calculation of $M_B(\eps|\Blocksoft)$ can be reduced to a calculus
problem. We state the results below deferring calculations to Appendix \ref{Appendix:block}.
\begin{lemma}\label{lemma:MSEBlockSoft}
Let $X_B$ bye a chi-square random variable with $B$ degrees of freedom
and define the functions $g,h:\bR_+\to\bR$ as follows
\begin{align*}
h(\tau^2) \equiv\frac{\tau}
{\E\Big\{ \big(\sqrt{X_B}-\tau\big)\ind_{X_B\ge\tau^2}\Big\}}\, ,\;\;\;\;\;\;\;
g(\tau^2) \equiv\frac{\tau\E\Big\{ \big(\sqrt{X_B}-\tau\big)^2
\ind_{X_B\ge\tau^2}\Big\}}
{\E\Big\{ \big(\sqrt{X_B}-\tau\big)\ind_{X_B\ge\tau^2}\Big\}}\, .
\end{align*}
The minimax risk of block soft thresholding over the class
$\cF_{N,\eps,B}$ is given by 
\begin{align}
M_B(\eps|\Blocksoft) =\frac{B+\tau^2+g(\tau^2)}{B(1+h(\tau^2))}\,
,\;\;\;\;\;\;
\eps = \frac{1}{1+h(\tau^2)}\, .
\end{align}
This is a parametric expression for $\tau\in [0,\infty)$. The
parameter corresponds to the minimax threshold $\tau$.
\end{lemma}
 
In Figure \ref{fig:blockMSE} we present graphs
of $M(\eps) = M_B(\eps|\Blocksoft)$ as a function of $\eps$.
It is immediate to prove the following  structural properties:
$(i)$ $0 \leq M(\eps) \leq 1$ (the upper bound follows from taking
$\tau=0$); $(ii)$ 
$M(\eps)$ is monotone increasing and concave (monotonicity is a
consequence of $\cF_{B,\eps}\subseteq \cF_{B,\eps'}$ for $\eps\le
\eps'$, and concavity follows since any measure in
$\cF_{q\eps_1+(1-q)\eps_2,B}$ can be written as convex combination of
measures in $\cF_{\eps_1,B}$ and in $\cF_{\eps_2,B}$); $(iii)$
$M(\eps) \goto 0$ as $\eps \goto 0$; 
 $M(\eps) \goto 1$ as $ \eps \goto 1$.
(Recall that we are considering the MSE per coordinate.)
Associated with the minimax problem is also an optimal threshold value
$\tau^*(\eps|B)$, that we plot in Figure \ref{fig:blockThresh}.

\begin{figure}
\begin{center}
\includegraphics[height=3in]{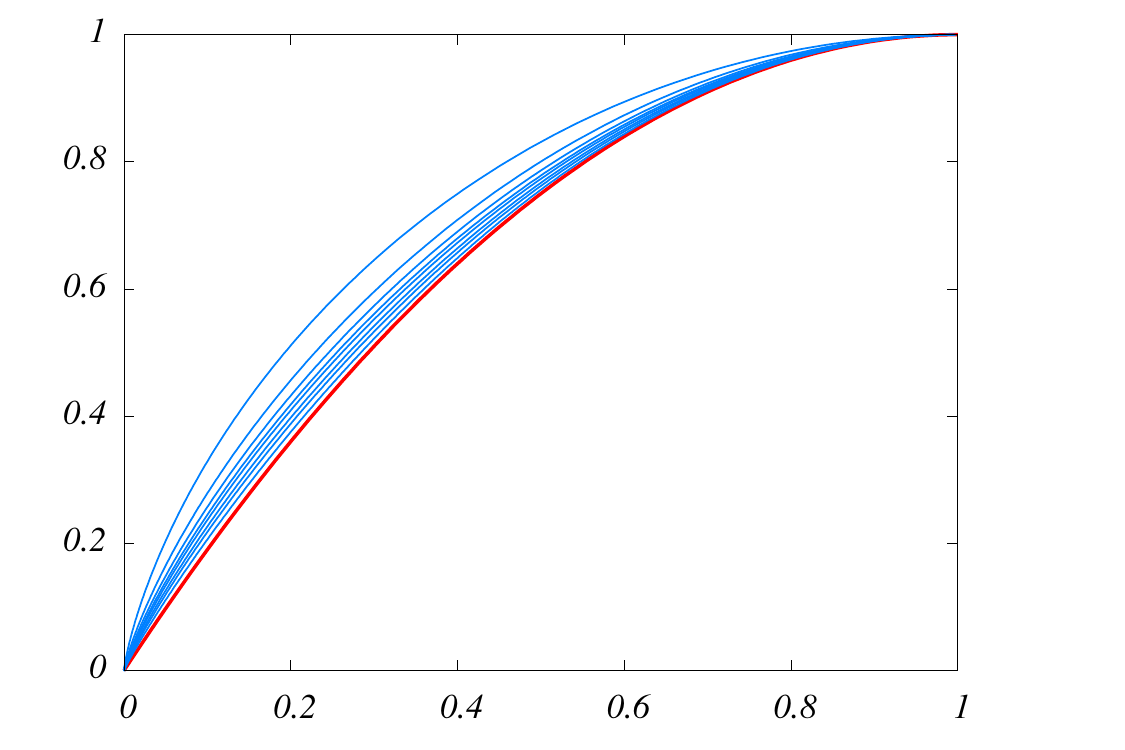}
\put(-170,-5){$\eps$}
\put(-220,80){$M_B(\eps|\Blocksoft)$}
\caption{Minimax MSE $M_B(\eps|\Blocksoft)$ of block soft-thresholding for block sparse
  vectors from the class $\cF_{N,\eps,B}$,
  cf. Eq.~(\ref{eq:MMAXBlockSoft}). 
Thin blue curves refer (from top to bottom) to block sizes  $B =
1$, $2$, $3$, $5$, $7$, $10$, $20$. Thick red line is the large
block size  limit $M_{\infty}(\eps|\Blocksoft)=2\eps-\eps^2$.}
\end{center}
\label{fig:blockMSE}
\end{figure}
\begin{figure}
\begin{center}
\includegraphics[height=2.4in]{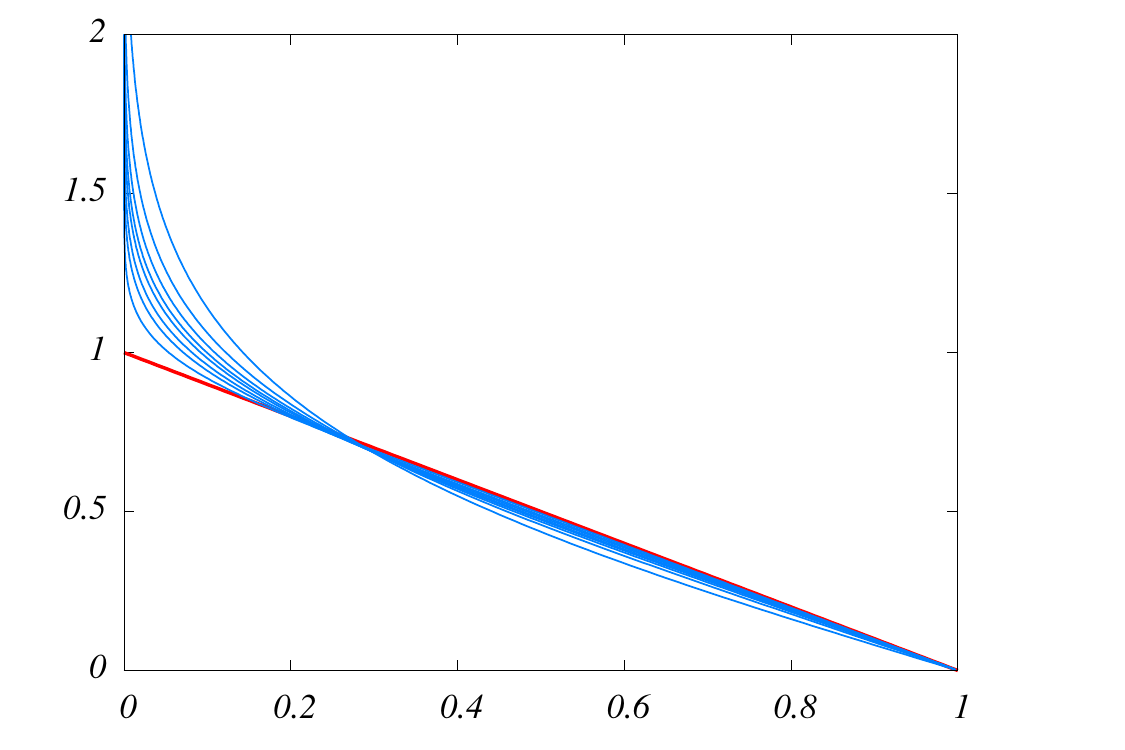}
\put(-135,-5){$\eps$}
\put(-130,70){$\tau^*(\eps|B)/\sqrt{B}$}
\caption{Minimax threshold $\tau^*(\eps|B)$  of block soft-thresholding for block sparse
  vectors from the class $\cF_{N,\eps,B}$. Thin blue lines are the normalized
  threshold  $\tau^*(\eps|B)/\sqrt{B}$
as a function of   $\eps$, for $B = 1$, $2$, $3$, $4$, $5$, $7$, $10$,
$20$ (from top to bottom on the left side of the plot, and from bottom
to top on the right). The thick red line is the large
block size  limit $\tau^*(\eps|\Blocksoft)=1-\eps$.}
\end{center}
\label{fig:blockThresh}
\end{figure}

A particularly interesting case is the one of large blocks.
As $B \goto \infty$ the minimax MSE has a well defined,
and particularly explicit limit.
\begin{lemma}
\label{lem:blockMSElimit}
For large $B$ we have
\[
  \lim_{B\to\infty}  M_B(\eps| \Blocksoft)  = M_\infty(\eps|\Blocksoft) \equiv 2\eps-\eps^2.
\]
Further, when properly normalized, the minimax threshold
converges with increasing block size:  
\[
\lim_{B\to\infty}\tau^*(\eps|B)/\sqrt{B} =\tau^*(\eps| \infty) \equiv
(1 -\eps)\, .
\] 
\end{lemma}
This lemma is proved in Appendix \ref{app:LargeBlocks}.

\subsection{Block James-Stein}
\label{sec:BlockJS}

From the point of view  of  $\ell_1$-penalized estimation, and 
compressed sensing, block soft thresholding seems very natural.
However, the limiting relationship in Lemma \ref{lem:blockMSElimit}
shows that this approach leaves room for improvement at  large
block sizes. 
A simple upper bound on the MSE is the one achieved by a denoiser
which utilizes a special oracle that  tells us without
error which blocks are zero and which are nonzero.
We refer to this as to the \emph{ideal} (or \emph{oracle}) \emph{MSE}. It 
easy to see that the minimax ideal MSE is  $ M_B(\eps |\Oracle) = \eps$. This
is considerably smaller than  $M_B(\eps|\Blocksoft)$, even in the
 $B \goto \infty$ limit characterized by  Lemma
 \ref{lem:blockMSElimit}.
On the other hand. the denoising/compressed sensing problems become 
easier as $B\to\infty$. Can we hope to achieve $M_B(\eps|\Oracle)$?

In order to approach oracle MSE for large $B$, we propose to use
the positive-part James-Stein shrinkage estimator \cite{JamesStein}.
This is again a block-separable denoiser that acts as follows on a block
$y\in \bR^B$:
\[
\eJS(y) =\Big( 1 - \frac{(B-2)}{\|y\|_2^2}\Big)_+  y\, .
\]
Analogously to block soft thresholding, this estimator shrinks to $0$
blocks with small norms.
On the other hand, its bias vanishes as $\|y\|_2\to\infty$.
Using once more Lemma \ref{lemma:BlockMSE} we have 
(notice that in this case there is no tuning parameter)
\begin{align*}
  M_B(\eps | \JamesStein) \equiv   \frac{1}{B}\, 
  \sup_{\nu\in\cF_{\eps,B}} 
\E_{\nu}\big\{[\bx-\eJS(\bx+\bz)]^2\big\}\, .
\end{align*}
Remarkably, the limiting $B  \goto \infty$ behavior  of this denoiser
is ideal, and noticeably better than block soft thresholding, as
shown by  Figure \ref{fig:JSvsBlocksoft} and formally in the next lemma.

\begin{figure}
\begin{center}
\includegraphics[height=2.4in]{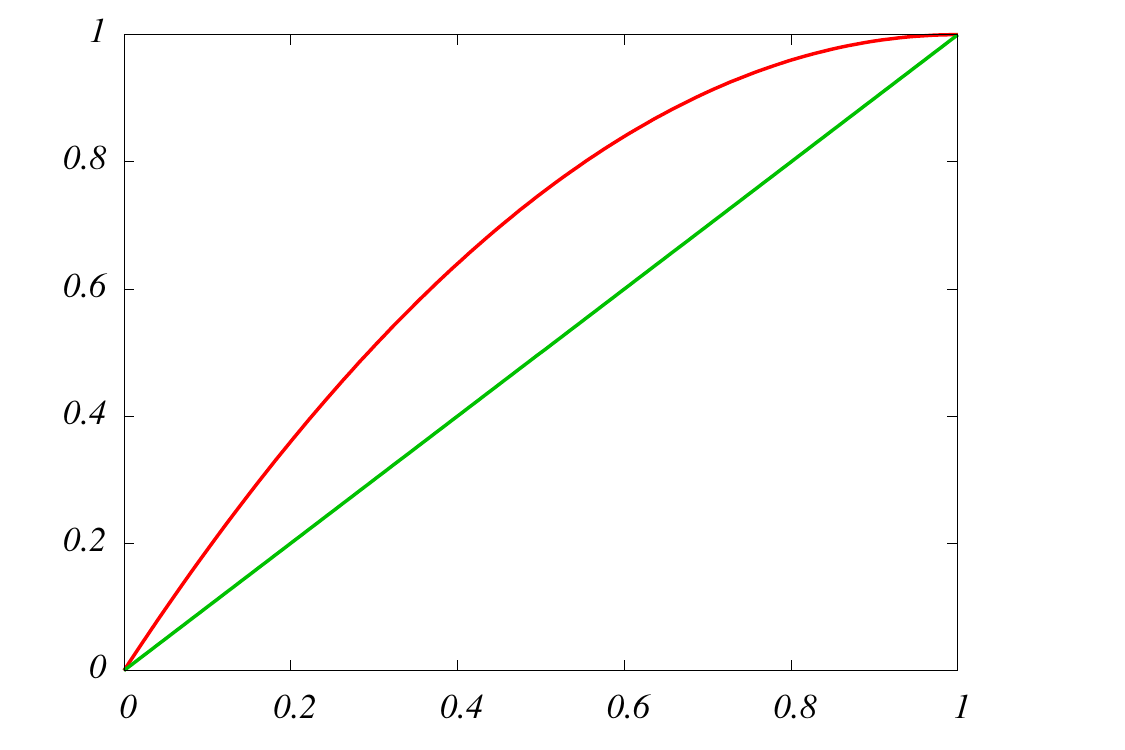}
\put(-130,-5){$\eps$}
\put(-215,135){$M_{\infty}(\eps|\Blocksoft)$}
\put(-150,50){$M_{\infty}(\eps|\JamesStein)$}
\caption{Large block sizes ($B \goto \infty$)  MSE of block soft (upper
  curve) and James-Stein (lower curve) denoisers.
The limit soft thresholding minimax MSE is
$M_{\infty}(\eps|\Blocksoft)=2\eps-\eps^2$.
The limit James-Stein minimax MSE coincides with the ideal MSE $M_{\rm
Oracle}(\eps)= \eps$. }
\end{center}
\label{fig:JSvsBlocksoft}
\end{figure}

\begin{lemma}
Let $M_B(\eps | \JamesStein)$ denote the minimax MSE for 
$\eJS$ over the class of $\eps$-block sparse  sequences.
For any $B > 2$, we have:
\[
M_B(\eps | \JamesStein ) \leq \eps + \frac{2}{B}.
\]
\end{lemma}
\begin{proof}
Consider temporarily the case where the observation is $\by = \mu + \bz$ with
$\bz \sim \normal(0,\id_{B\times B})$, and  $\mu\in\reals^B$  nonrandom
and known.
A simple calculation shows that the optimal linear
estimator of the form $\eta(y) = cy$, $c\in\bR$,  is given by
\begin{align*}
   \eta^{IDL} (y)= c(\mu) \cdot y\, , \;\;\;\; c(\mu) \equiv\frac{\| \mu\|_2^2}{\| \mu^2\|_2^2 + B}\, .
\end{align*}
This estimator uses information about $\| \mu\|_2$ (which
could only be supplied by an oracle) to choose the constant $c$ as a function of $\| \mu\|_2$.
Note in particular that the risk of this estimator is
\begin{align*}
R(\mu;\eta^{IDL} ) \equiv \E\big\{\|\eta^{IDL}(\mu+Z)\|_2^2\big\} = \frac{ B\|\mu\|_2^2 } { \| \mu \|_2^2 + B } ,
\end{align*}
and, in particular,
\begin{equation} \label{eq:idealMSEbounds}
   0 = R(\mu=0,\eta^{IDL}) \leq \sup_{\mu\in\reals^B} R(\eta^{IDL},\mu ) = B
   \, .
\end{equation}
Define the ideal worst-case MSE by
\[
    M_B(\eps | \eta^{IDL}) \equiv\frac{1}{B}   \sup_{\nu \in \cF_{\eps,B} } \E_{\nu}\{ R(\mu=X,\eta^{IDL}) \}  .
\]
Applying (\ref{eq:idealMSEbounds}), and 
keeping in mind that, for $\nu\in\cF_{\eps,B}$, $\nu(\{X=0\})\ge (1-
\eps)$, we have
\begin{equation} \label{eq:MMxIdeal}
   M_B(\eps | \eta^{IDL} ) =  \eps.
\end{equation}
The oracle inequality \cite[Theorem 5]{DoJo95} shows that for $B > 2$,
and for every vector $\mu \in \bR^B$, if $Y \sim
\normal(\mu,\id_{B\times B})$, then
\[
      R(\mu;\eta^{JS}) \leq R(\mu;\eta^{IDL}) +  2\,.
\]
Combined with the previous display this proves the Lemma.
\end{proof}

The argument in the proof leads in fact to a convenient  expression
for  $M_B(\eps | \JamesStein)$. With the notations introduced there, we have
\[
M_B(\eps | \JamesStein )  =   \frac{1-\eps}{B}\, R(0;\eJS)+\frac{\eps}{B}\, \sup_{\mu\in\bR^B} R(\mu;\eJS)\,.
\]
Now $\eJS$ is known to be minimax for the unconstrained problem of
estimating a non-sparse vector $\mu$, i.e. $ \sup_{\mu\in\bR^B} R(\mu;\eta^{JS}) = B$
yielding
\[
M_B(\eps | \JamesStein)  =  \frac{1-\eps}{B}\, R(0;\eJS)+ \eps\, .
\]
Therefore computing the minimax MSE for $\eJS$ reduces to computing
the single quantity  $ R(0;\eJS)$, that can be estimated through
numerical integration.
A good approximation for large $B$ is provided by
the following  formula  $ R(0;\eJS)=B^{-1}
+\kappa B^{-3/2}+O(B^{-2})$ with $\kappa\approx 0.752$ (cf. Appendix \ref{sec:JS_at_0}). Hence we have
\begin{equation} \label{eq:JSM}
  M_B(\eps | \JamesStein )  = (1-\eps) \Big\{\frac{1}{B}    +
  \frac{\kappa}{B^{3/2}} +O(B^{-2})\Big\} + \eps .
\end{equation}
In the next section will use this formula (neglecting 
$O(B^{-2})$ terms)  in comparing the general
prediction of Eq.~(\ref{generalPT}) with the empirical results for the
James-Stein AMP algorithm. Numerical integration reveals that this
formula is accurate enough for such comparison.
  
\subsection{Empirical phase transition behavior}
\label{sec-PTExperiments-vector}

We now turn to the compressed sensing reconstruction problem whereby
the block-sparse vector $x_0$ is reconstructed from observed data
$y=Ax_0$ using the AMP algorithm.
We want to  test the hypothesis that Eq.~(\ref{generalPT})
describes the phase transition of the two 
block shrinkage AMP algorithms, corresponding to the block soft
thresholding, and block James-Stein.

We conducted a set of experiments similar to those 
described in Section \ref{sec-PTExperiments-scalar}
We constructed block-sparse signals at 
different undersampling and sparsity levels and ran tests of
block thresholding AMP. More precisely, we used the update equations
(\ref{AMPA}) to (\ref{AMPC}) with $\eta=\esoft$ (block soft
thresholding AMP) or $\eta =\eJS$ (James-Stein AMP).

It is a straightforward  calculus exercise to compute an explicit
expression for the memory term $\onsager_t$. For block
soft thresholding AMP we get
\begin{align}
\onsager_t\equiv \left.\frac{1}{n} \, \div\,
  \esoft(y;\tau,\sigma_{t-1})\right|_{y=y^{t-1}} =
\frac{1}{n}\sum_{\ell=1}^{M}\left(B-\frac{(B-1)\tau\sigma_{t-1}}{\|block_{\ell}(y^{t-1})\|_2}\right)
\ind_{\{\|block_{\ell}(y^{t-1})\|_2>\tau\sigma_{t-1}\}}\, .
\end{align}
For  James-Stein AMP  we have
\begin{align}
\onsager \equiv \left.\frac{1}{n} \, \div\,
  \eJS(y;\sigma_{t-1})\right|_{y=y^{t-1}} =
\frac{1}{n}\sum_{\ell=1}^{M}\left(B-\frac{(B-2)^2}{\|block_{\ell}(y^{t-1})\|_2^2}\right)\ind_{\{\|block_{\ell}(y^{t-1})\|^2_2>B-2)\}}\, .
\end{align}

Our results show that the curve $\delta = M_B(\eps | \Blocksoft)$  correctly separates
two phases of performance: below this curve success in AMP recovery is atypical and above
it is typical.  Similarly,
the curve $\delta = M_B(\eps | \JamesStein)$ correctly describes the phase transition for block James-Stein shrinkage.
The empirical results are presented in Figure \ref{fig:BlockSoft} (for
block soft thresholding) and Figure \ref{fig:BlockJamesStein} (for
block James-Stein). We refer to Appendix \ref{Appendix:Computation}
for further details.

\begin{figure}
\begin{center}
\includegraphics[height=3in]{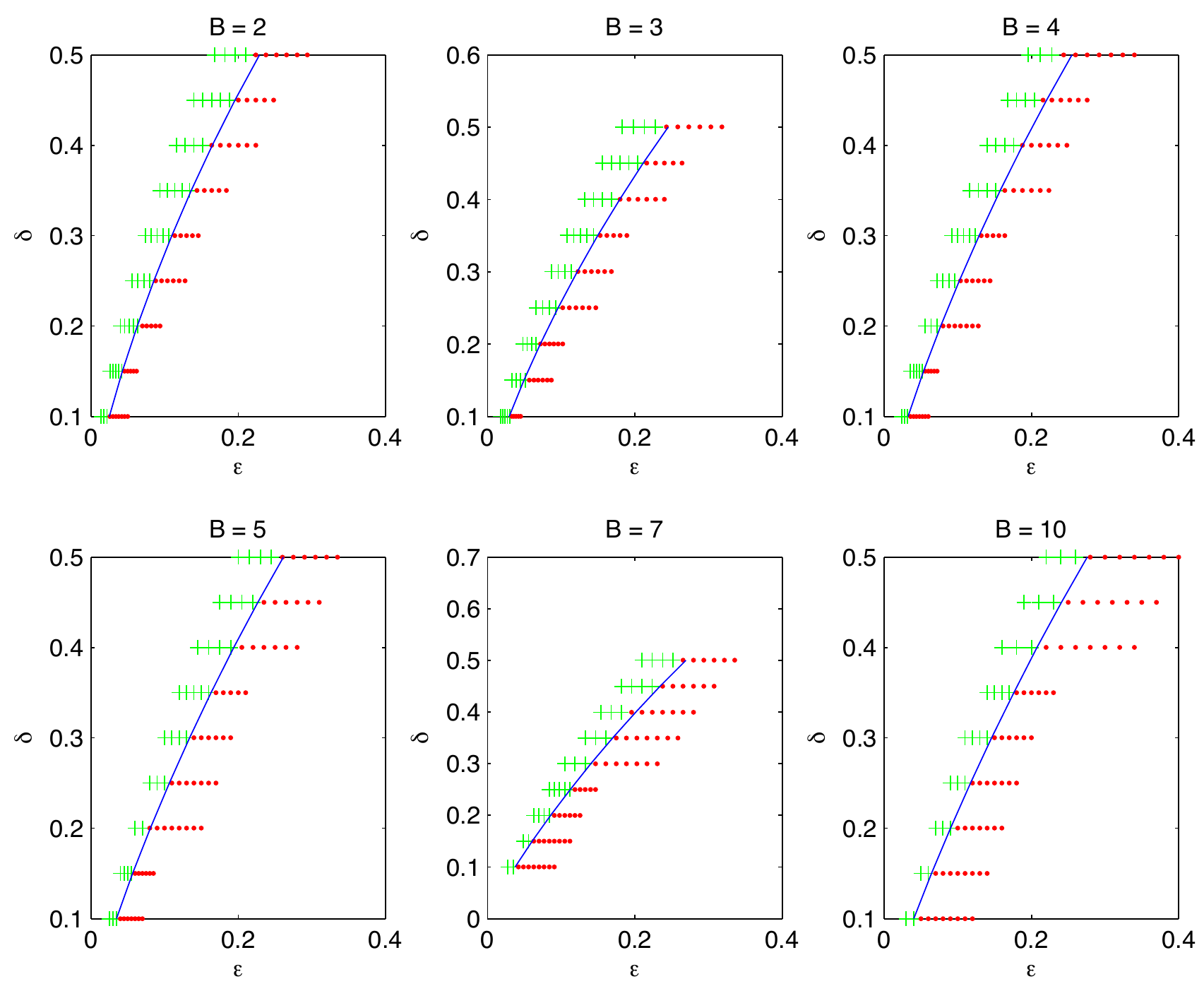}
\caption{Phase transition results for block soft thresholding AMP with minimax
  threshold. Here the signal dimension is $N=1000$,
$\delta = n/N$ 
is the undersampling fraction, and $\eps$ is the sparsity parameter
(fraction of non-zero entries). Red: less than $50\%$ fraction of correct recovery.
Green: greater than $50\%$ fraction of successful recovery. Blue Curve:
$\delta = M_B(\eps| \Blocksoft)$.}
\label{fig:BlockSoft}
\end{center}
\end{figure}

\begin{figure}
\begin{center}
\includegraphics[height=3in]{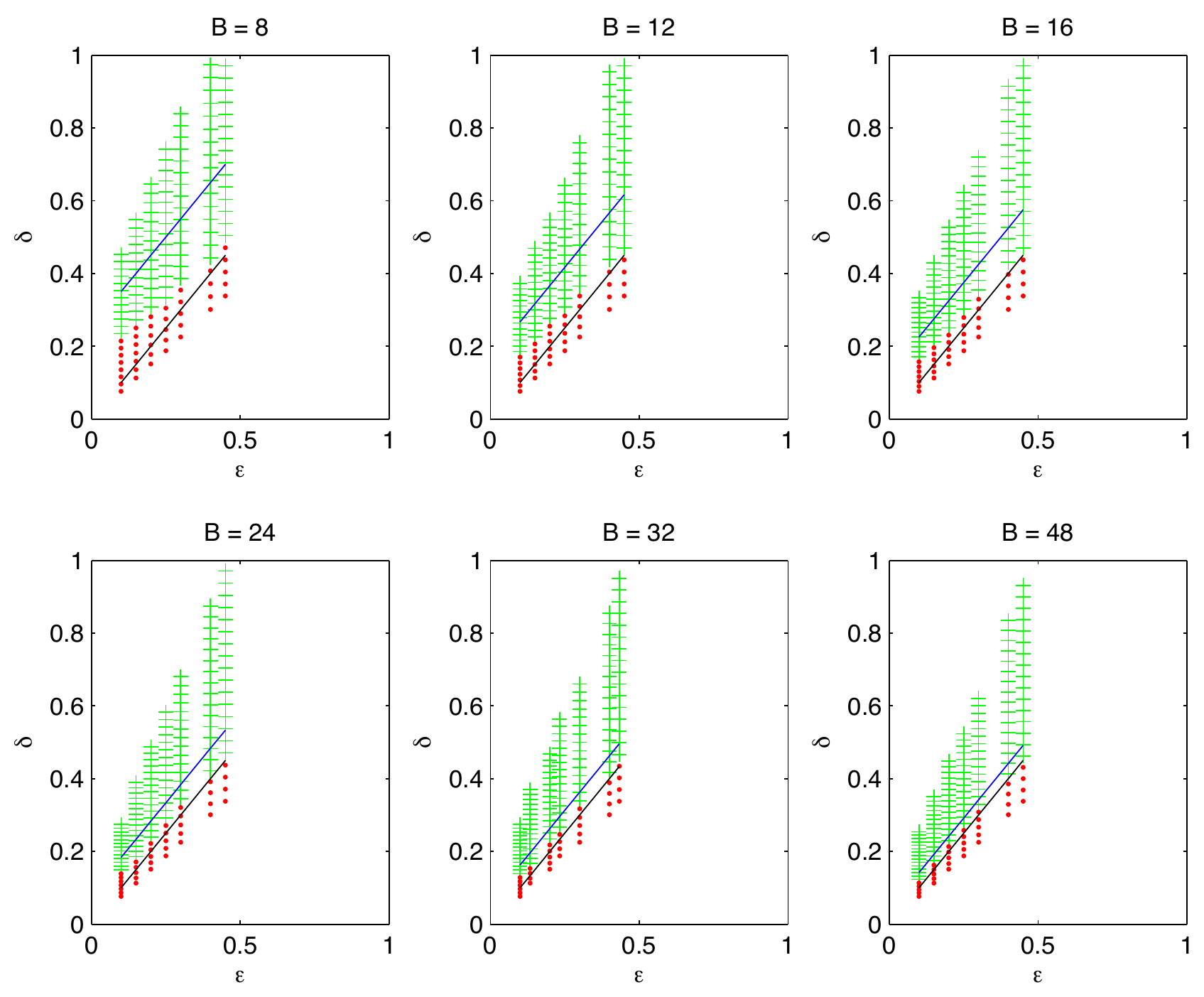}
\caption{Phase transition results for block James-Stein AMP.  
Here the signal dimension is $N=1000$,
$\delta = n/N$ 
is the undersampling fraction, and $\eps$ is the sparsity parameter
(fraction of non-zero entries). 
Red: less than $50\%$ fraction of correct recovery.
Green : greater than $50\%$ fraction of successful recovery. Blue Curve:
asymptotic (large-$B$) formula for
$\delta = M_B(\eps|\JamesStein)$ (\ref{eq:JSM}). 
Black Curve: $\delta= \eps$ (lower bound on minimax risk of
James-Stein Shrinkage).}
\label{fig:BlockJamesStein}
\end{center}
\end{figure}
%
%
\section{Monotone regression}
\label{sec:Mono}

In this section and the next, we show that, quite surprisingly,  the
formula  
(\ref{generalPT}) can be applied also to some highly nontrivial non-separable 
denoisers.

In this section we consider vectors that are monotone, and mostly constant.
Let $\cMono$ denote the cone of nondecreasing sequences:
\begin{align*}
     \cMono_N \equiv \big\{ x  \in \bR^N \;\;\; : \;\;\; x_{t+1} \geq x_t  \;\; \mbox{
       for all } t\in\{1,2,\dots, N-1\} \big\}.
\end{align*}
We then define the class of mostly constant non-decreasing vectors
\begin{align*}
\cF_{N,\eps,\mono}\equiv
\Big\{\nu_N\in\cP(\bR^N)\, :\;\;\;\supp(\nu_N)\subseteq \cMono_N,\;\;
\E_{\nu_N}\big[\#\{t\in [N-1]:\; x_{t+1}>x_t\}\big]\le N\eps \Big\}\, .
\end{align*}
Since vectors from this class are --in general-- not sparse, 
we will occasionally refer to the parameter $\eps$  as to the `simplicity' parameter.

For this problem we will consider the denoiser  $\eta^{mono}:\bR^N\to\bR^N$, that
solves the \emph{monotone regression} problem
\begin{eqnarray}
    \eta^{mono} (y) = \argmin_{x\in\cMono_N} \;\; \| y - x \|_2^2\, .\label{eq:MonoDenoiser}
\end{eqnarray}
In other words, $\emono$ is the (Euclidean) projection on the cone of
monotone sequences.
This denoiser is highly non-separable, as one can understand most clearly
by studying the standard pool-adjacent-violators algorithm for implementing
it (see \cite{Best} for a recent reference).

\subsection{Minimax MSE}

In order to apply formula (\ref{generalPT}), we need to calculate $M(\eps | \Mono)$,
which requires in particular determining the least favorable
distribution $\nu_N\in\cF_{\eps,N,\mono}$ and proving that the limit
$N\to\infty$ of the minimax MSE exists. We present here the main
ideas, deferring details to Appendix \ref{Appendix:monotone}.

It is convenient to introduce the risk at $\mu\in\cMono_N$:
\begin{eqnarray}
R_N(\mu) \equiv \E\big\{\big\|\emono(\mu+ \bz)-\mu\big\|^2_2\big\}\, ,
\end{eqnarray}
where expectation is taken with respect to
$\bz\sim\normal(0,\id_{N\times N})$. 
It is further useful to introduce a specific notation for the risk at $0$, namely 
\begin{align}
r(N) \equiv R_N(\mu=0) =  \E\big\{\big\|\emono( \bz)\big\|^2_2\big\}\, .\label{eq:MonoRiskAtZero}
\end{align}
\begin{lemma}\label{lemma:MonotoneRiskAtInfty}
The risk of monotone regression satisfies the following properties
\begin{enumerate}
\item[$(a)$] The function $t\mapsto R_N(t\,\mu)$ is monotone
  increasing for $t\in\bR_+$.
\item[$(b)$] Let $I_+(\mu)\equiv \{i\in [N-1]:\; \mu_i<\mu_{i+1}\}$
be the set of increase points of $\mu$. Denoting them by 
$I_{+}(\mu)\equiv \{i_1,i_2,\dots, i_{K(\mu)}\}$, $i_k\le i_{k+1}$, let $J_k \equiv \{i_k+1,i_k+2,\dots, i_{k+1}\}$ for
$k\in \{0,\dots, K(\mu)\}$ (with, by convention, $i_0=0$, $i_{K(\mu)+1}=N$). 
Then, for any $\mu\in\cMono_N$,
\begin{align}
\lim_{t\to\infty} R_N(t\mu) = \sum_{k=0}^{K(\mu)} r(|J_{k}(\mu)|)\, . \label{eq:MonotoneRiskAtInfty}
\end{align}
\end{enumerate}
\end{lemma}
\begin{proof}[Proof of part $(a)$]
For a non-empty closed convex $\cS\subseteq \bR^N$, we let
$\Pr_\cS:\bR^N\to\bR^N$ denote the Euclidean projector to $\cS$, i.e. 
$\Pr_\cS(y)\equiv \argmin_{x\in S}\|x-y\|_2$. Further, for $v\in\bR^N$,
$\cS+v\equiv\{x+v\, :\; x\in \cS\}$. 

Note that it is sufficient to show that, letting
$D(\mu;z)\equiv\big\|\emono(\mu+ z)-\mu\big\|_2^2$,
$t\mapsto D(t\mu;z)$ is monotone increasing in $t\in\bR_+$.
By continuity of the projection operator, it follows that $\mu\mapsto
D(\mu;z)$ is continuous. Further, notice that $\cMono_N$ is a cone obtained
as the intersection of $N-1$ half-spaces: 
\begin{align*}
\cMono_N = \cap_{i=1}^{N-1}\cH_i\; ,\;\;\;\;\;\;\;
\cH_i\equiv\big\{x\in\bR^N:\; x_i\le x_{i+1}\big\}\, .
\end{align*}
Let $\cV_i= \{x\in\bR^N:\; x_i= x_{i+1}\}$ be the separating
hyperplane for $\cH_i$ and, for $B\subseteq[N-1]$,
define
\begin{align*}
\cV_B \equiv \cap_{i\in B}\cV_i = \big\{x\in\bR^N:\;\; x_i=x_{i+1} \;
\forall i\in B\big\}\, ,
\end{align*}
with, by convention $\cV_{\emptyset} \equiv\bR^N$. 
Since $\emono = \Pr_{\cMono_N}$, we have that $\emono$ is continuous
and piecewise
linear and equal one of the projectors $\Pr_{\cV_B}$, that we will
denote by $\Pr_{B}$ for $B\subseteq [N-1]$. It is therefore sufficient
to show that, defining for $B\subseteq [N-1]$, 
\begin{align*}
D_B(\mu;z)\equiv\big\|\Pr_B(\mu+ z)-\mu\big\|_2^2,
\end{align*}
the function $t\mapsto D_B(t\mu,z)$ is monotone increasing for $t\in
\bR_+$.

Let $B\equiv\cup_{k=1}^KB_k$ where each $B_k$ is a contiguous
segment (in the sense of point $(b)$), and $\oB_k \equiv \{i\in[N]:
i\in B_k\vee (i-1)\in B_k\}$. Further, for $x\in \bR^N$, let
$\ox_S\equiv |S|^{-1}\sum_{i\in S}x_i$. Then, for any $x\in\bR^N$,
\begin{align*}
\Pr_B(x)_i = \begin{cases}
\ox_{\oB_k}&\mbox{ if $i\in B_k$, $k\in\{1,\dots,K\}$,}\\
x_i&\mbox{ otherwise.}
\end{cases}
\end{align*}
Hence
\begin{align*}
D_B(t\mu;z) = \sum_{k=1}^K|\oB_k|\,
\left[\frac{t^2}{|\oB_k|}\sum_{i\in\oB_k}(\mu_i-\omu_{\oB_k})^2+\oz_{\oB_k}^2\right]
  +\sum_{i\in [N]\setminus \cup_{k=1}^K\oB_k } z_i^2\, ,
\end{align*}
which is clearly increasing in $t\in\bR_+$.
\end{proof}
The proof of part $(b)$ is deferred to Appendix
\ref{Appendix:monotone}.

The last Lemma shows that the least favorable signal $\mu$  is
constant on  $N(1-\eps)$ positions of the interval $\{1,2,\dots,N\}$ 
and has large (going to infinity) jumps at the remaining $N\eps$
increase points. The resulting risk only depends on the distribution
of the lengths of the intervals over which $\mu$ is constant.

 \begin{figure}
\begin{center}
\includegraphics[height=2.3in]{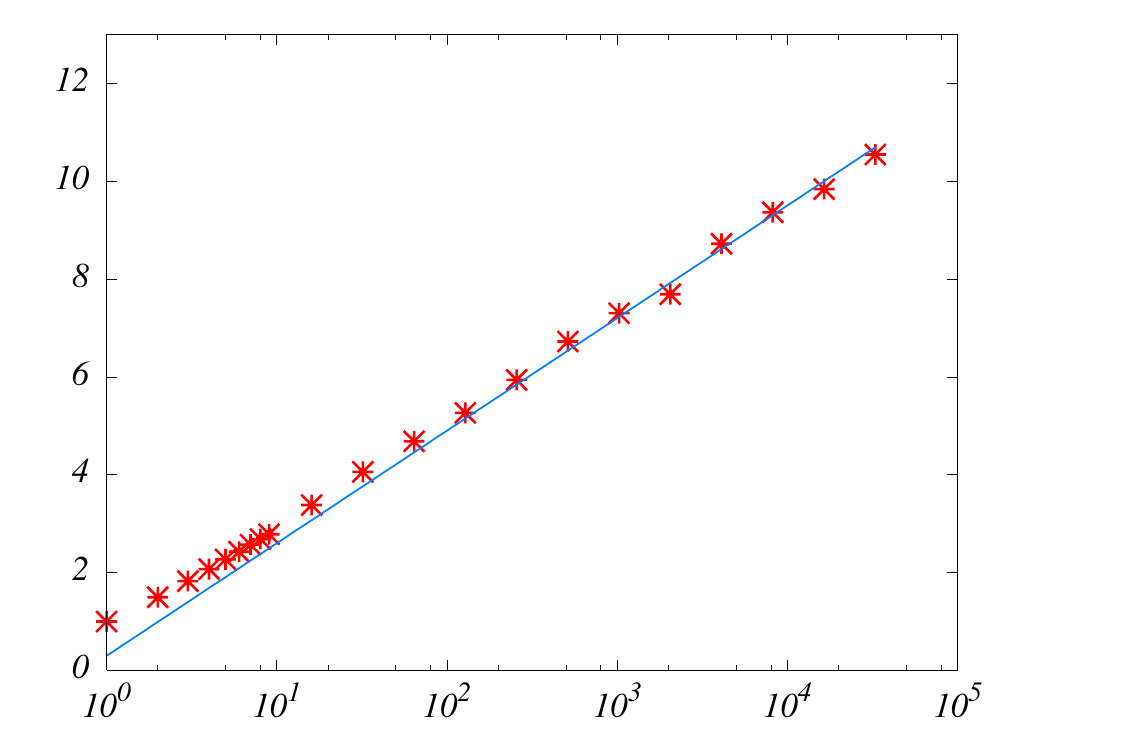}
\put(-140,-5){$\ell$}
\put(-270,100){$r(\ell)$}
\includegraphics[height=2.3in]{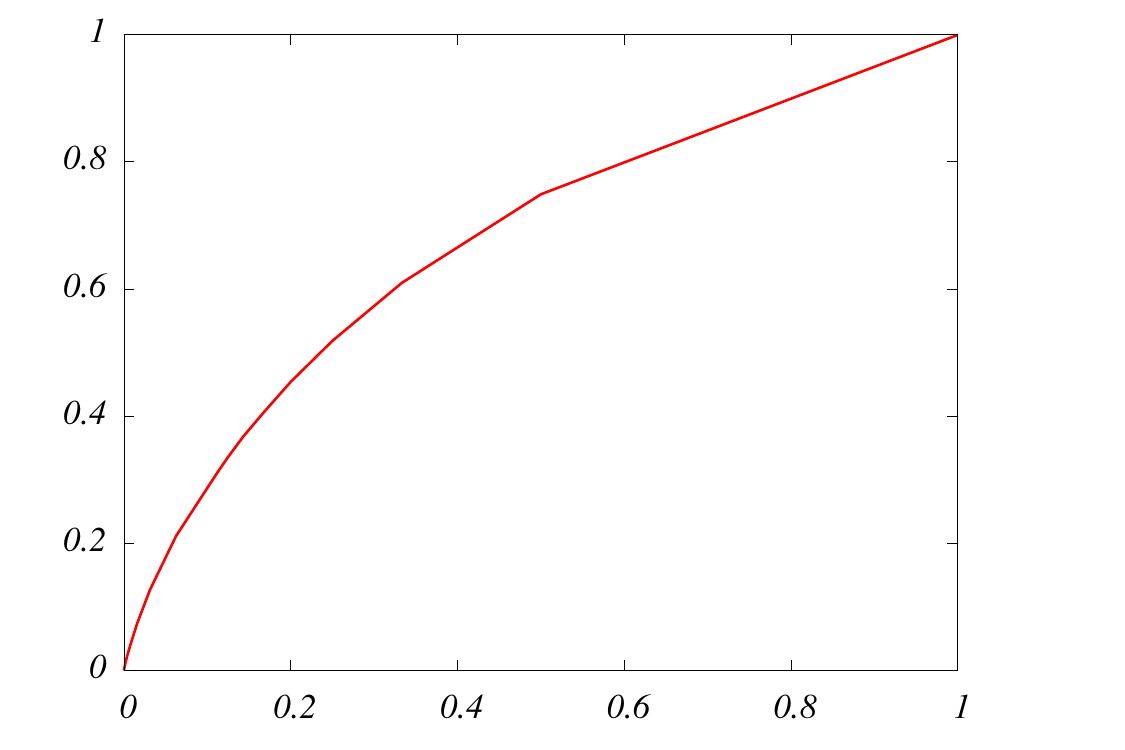}
\put(-140,-5){$\eps$}
\put(-160,80){$M(\eps|\Mono)$}
\caption{Left: risk at $0$  for monotone regression.
Here we plot the mean square error $r(\ell)\equiv R_{\ell}(\mu=0)$ as
a function of the signal dimension $\ell$. Symbols were obtained
through Monte Carlo simulations for $\ell \{1,2,3 \dots,10,16,32,
\dots, 32,768=2^{15}\}$. 
The continuous line is the fit $r(\ell) = \log(\ell)+0.3$. Right:
Minimax mean square error of monotone regression, as per Theorem \ref{thm:MMAX_Mono}.}
\label{fig:FragmentMSE}
\end{center}
\end{figure}

The next Lemma provides some useful insight on the behavior  of the
risk at $0$. This is crucial since it determines the minimax risk
though Eq.~(\ref{eq:MonotoneRiskAtInfty}).
\begin{lemma}\label{lemma:MonoRiskAtZero}
The monotone regression risk at zero, defined through
Eq.~(\ref{eq:MonoRiskAtZero}) satisfies $r(1)=1$
and, for any $N\ge 10$
\begin{align}
r(N)\le 20\, (\log N)^2\, .\label{eq:MonotoneLog}
\end{align}
Further $r(N)<2N[\log N+1]$ for all $N$.
\end{lemma}
The proof of this Lemma can be found in Appendix \ref{sec:ProofMonoRiskAtZero}.

For moderate values of $N$, $r(N)$ can be computed numerically through
Monte Carlo simulations. Figure \ref{fig:FragmentMSE} presents the
results of such a simulation. It appears that $r(N) = \Theta(\log N)$
as $N\to\infty$ suggesting that the last lemma is loose by a
logarithmic factor.

We can finally establish our main result on the minimax MSE of
monotone regression over the class $\cF_{N,\eps,\mono}$.
Remarkably, we are able to characterize the least favorable
distribution $\nu_N\in\cF_{N,\eps,\mono}$. 
\begin{theorem}\label{thm:MMAX_Mono}
The asymptotic minimax MSE of monotone regression
\begin{align*}
M(\eps|\Mono)=\lim_{N\to\infty}M(\cF_{N,\eps,\mono}|\emono)
\end{align*}
exists and is given by
\begin{eqnarray}
       M(\eps | \Mono) = \max \Big\{ \eps\sum_{\ell\ge 1} \pi_\ell r(\ell)\;
       :\;\; \sum_{\ell\ge 1}   \ell  \pi_\ell= 1/\eps \Big\}\, , \label{eq:MonoRegMMAXFinal}
\end{eqnarray}
where the maximization is over the probability distribution $\{\pi_{\ell}\}_{\ell\in\naturals}$.
 Equivalently, the curve $(1/\eps$, $M(\eps| \Mono)/\eps)$ is the least
 concave envelope of the point set  $\{(\ell , r(\ell))\}_{\ell
   \in\naturals}$. Further, there exists $\eps_0>0$ such that, for all
 $\eps\in [0,\eps_0]$, 
\begin{eqnarray*}
M(\eps|\Mono) \le 20\,\eps\, (\log 1/\eps)^2\, . \label{eq:MepsMono}
\end{eqnarray*}
Finally, for any $\xi>0$, $\eps\in (0,1)$, the following distribution
$\nu_N^{(\eps,\xi)}\in \cF_{N,\eps,\mono}$ has risk larger that
$M(\eps|\Mono)-\xi$ for all $N$ large enough. A signal $\bx\sim
\nu^{(\eps,\xi)}_N$ has $X_{i+1}-X_{i}= \Delta>0$ at all increase
points $i\in [N-1]$ for some $\Delta = \Delta(\xi)$ large enough, and the
lengths of intervals between increase points have distribution $\pi$ 
achieving the max in Eq.~(\ref{eq:MonoRegMMAXFinal}).
\end{theorem}
\begin{proof}
With a slight abuse of notation we define, for
$\nu_N\in\cF_{N,\eps,\mono}$, the expected risk of monotone regression
as $R_N(\nu_N) = \E_{\nu_N}\{\|\emono(\bx+\bz)-\bx\|_2^2\}$ where
$\bx\sim\nu_N$. Further, for $t\in\bR_+$, let $S_t\nu_N$ be the
distribution obtained by rescaling $\nu_N$: $S_t\nu_N((a,b]) =
\nu_N((a/t,b/t])$. Further let $\{\pi^\bx_\ell\}_{\ell}$ be the
empirical distribution of the lengths of the constant intervals $J_k$
(as per Lemma \ref{lemma:MonotoneRiskAtInfty}.$(b)$), and $K(\bx)$ be
their number. We then have, by
Lemma  \ref{lemma:MonotoneRiskAtInfty},
\begin{align*}
R_N(\nu_N)\le \lim_{t\to\infty}R_N(S_t\nu_N) =  
\E_{\nu_N}\Big\{K(\bx)\sum_{\ell=1}^{\infty}\pi^{\bx}_{\ell}\; r(\ell)\Big\}\, .
\end{align*}
Further, by definition
\begin{align*}
\E\Big\{\frac{K(\bx)}{N}\Big\} =
\E\Big\{\frac{1}{\sum_{\ell=1}^{\infty}\pi^{\bx}_{\ell}\;\ell}\Big\} =\eps\, .
\end{align*}
Hence $M(\eps|\Mono)$ is immediately upper bounded by the right hand
side of Eq.~(\ref{eq:MonoRegMMAXFinal}). The matching lower bound is
obtained by evaluating the above expressions for the distribution
$\nu_N^{(\eps,\xi)}$. 

Finally Eq.~(\ref{eq:MepsMono}) follows by using Lemma
\ref{lemma:MonoRiskAtZero} in Eq.~(\ref{eq:MonoRegMMAXFinal}).
\end{proof}
The resulting curve $M(\eps|\Mono)$ is presented in
Fig.~\ref{fig:FragmentMSE}.

\subsection{Empirical phase transition behavior}

\begin{figure}
\begin{center}
\includegraphics[height=2.3in]{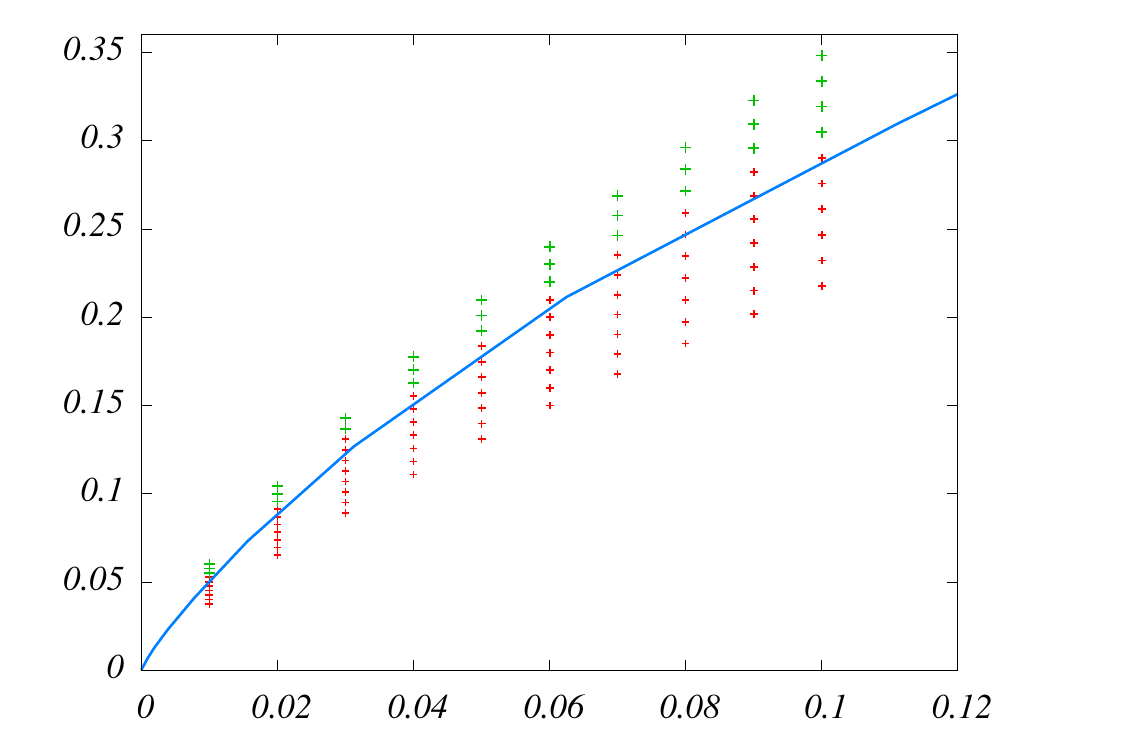}
\put(-140,-5){$\eps$}
\put(-255,100){$\delta$}
\includegraphics[height=2.3in]{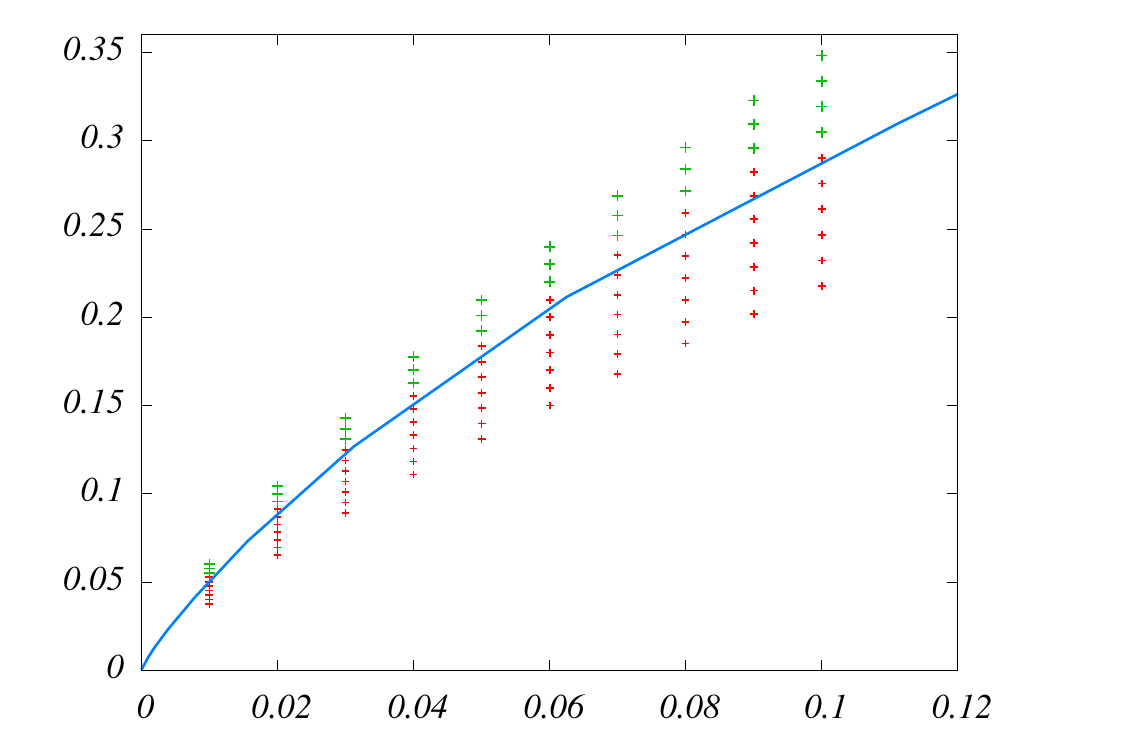}
\put(-140,-5){$\eps$}
\put(-255,100){$\delta$}
\caption{Results for monotone regression AMP.  Here $\delta = n/N$ undersampling
  fraction and $\eps$ is the sparsity measure (i.e. the fraction of
  increase points in the signal to be reconstructed). Signals were
  generated according to the least favorable distribution in Theorem \ref{thm:MMAX_Mono}.
Small red crosses:  less than $50\%$ fraction of correct recovery.
Large green crosses:  greater than $50\%$ fraction of correct
recovery.  Curve: minimax MSE  $\delta = M(\eps | \Mono)$. Left
frame: $N=1000$. Right frame: $N=2000$.}
\label{fig:monoRegAMPPT}
\end{center}
\end{figure}

We next consider the compressed sensing problem.
We programmed the  AMP iteration (\ref{AMPA})-(\ref{AMPC}),
with $\eta(\,\cdot\, ;\tau,\sigma) = \emono(\,\cdot\,)$ the monotone
regression denoiser. The denoiser itself was implemented using the
standard pool adjacent violators algorithm.
 
It is a simple exercise to obtain an explicit formula for the memory
term $\onsager_t$.
As in Lemma \ref{lemma:MonotoneRiskAtInfty}, let $K(\mu)$ denote the
number of increase points in the signal $\mu\in\bR^N$. We then have
\begin{align}
\onsager_t \equiv \left.\frac{1}{n} \, \div\, \emono(y)\right|_{y=y^{t-1}} =
\frac{1}{n}\Big\{
K(\emono(y^{t-1}))+1\Big\}\, .
\end{align}
We will refer to this specific version of AMP as to \emph{monoreg AMP}.

Our numerical simulations are summarized in Figure
\ref{fig:monoRegAMPPT}.
Each data point corresponds to the empirical success probability over
$100$ independent reconstruction experiments, using approximately
least favorable signals. More precisely, we used piecewise constant
signals, increasing, with equal jump sizes $\mu$ and 
constant intervals distributed according to the minimax law
$\{\pi_{\ell}\}$. The signals start with $x_1=\mu$, and we took
$\mu=10$ (the results were statistically independent of $\mu\gtrsim 5$). 

 In the present case we
evaluated success probability using the following (Hamming-like)
distance
\begin{align}
H_{\alpha}(x^t,x_0) \equiv \frac{1}{n}\Big|\big\{i\in [N]:\;\;
|x_i^t-x_{0,i}|\ge \alpha\big\}\Big|\, ,\label{eq:HammingDef}
\end{align}
and declared a success when $H(x^{t},x_0)\le \beta$. In Figure
\ref{fig:monoRegAMPPT} we used $t=300$ and $\alpha=\beta = 0.01$, but
very similar results are obtained with other values of the
parameters.  The rationale for using $H_{\alpha}(x^t,x_0)$ instead of
the normalized mean square error lies in the structure of the signals
$x_0$. Since the least favorable $x_0$ is monotone with large jumps,
its norm is very large,concentrated  at endpoints, and depends strongly
on $N$. This leads to subtle normalization issues across different
$N$.

The agrement between the empirical phase transition and the general
prediction $\delta = M(\eps|\Mono)$ in Fig.~\ref{fig:monoRegAMPPT} is
satisfactory
and improves with the signal's length.
%
%
\section{Total variation minimization}
\label{sec:TV}

In this section we consider vectors $x\in\bR^N$ that are mostly
constant, with a few change points. In order to model this problem, we
introduce the class of probability distributions 
\begin{align*}
\cF_{N,\eps,TV}\equiv
\Big\{\nu_N\in\cP(\bR^N)\, :\;\;\;
\E_{\nu_N}\big[\#\{t\in [N-1]:\; x_{t+1}\neq x_t\}\big]\le N\eps \Big\}\, .
\end{align*}
Again $\eps\in (0,1)$ is a  `simplicity' parameter. Note that this
class is quite similar to the class $\cF_{N,\eps,\mono}$ studied in
the previous section, the `only' difference being that change points
can be either points of increase or points of decrease. 

A convenient denoiser for this setting is the  \emph{total variation penalized least-squares} \cite{ROF92},
also called fused LASSO  \cite{Fused}, that we will denote by
$\etv(\,\cdot\,;\tau):\bR^N\to\bR^N$. This depends on $\tau\in\bR_+$
and, for $y\in\bR^N$ and noise variance $\sigma^2=$, it returns 
\begin{eqnarray}
    \etv (y;\tau) &\equiv & \argmin_{x\in\bR^N} \Big\{\frac{1}{2}\| y - x \|_2^2  +
    \tau\|x\|_{TV} \Big\}\, ,\label{eq:ETV}\\
\|x\|_{TV} &\equiv &\sum_{i=1}^{N-1} | x_{i+1} - x_{i} |\, .
\end{eqnarray}
An extensive literature is devoted to solving 
this denoising problem, see for example
\cite{TVSurvey}.
For general noise variance $\sigma^2$, the above expression is generalized through
the usual scaling relationship (\ref{eq:ScalingRelation}).

Much of the analysis in this section is analogous to the one of
monotone regression. We will therefore present several arguments in
synthetic form to limit redundancy.

\subsection{Minimax MSE}

In this section we outline the computation of the asymptotic minimax
MSE of the total variation denoiser over the class $\cF_{N,\eps,TV}$,
to be denoted by $M(\eps|TV)$.  

We start by defining a generalization of the problem (\ref{eq:ETV}).
For $s = (s_1,s_2)\in \{+1,-1\}^2$, we let
\begin{eqnarray*}
   \etv_s (y;\tau) &\equiv & \argmin_{x\in\bR^N} \Big\{\frac{1}{2}\| y - x \|_2^2  +
    \tau\|x\|_{TV} + \tau(s_1x_1+s_2x_N)\Big\}\, .
\end{eqnarray*}
For economy of notation, we will write $s = ++$, $+0$, $+-$, $\dots$ instead
of, respectively, $s=(+1,+1)$, $(+1,0)$, $(+1,-1)$, $\dots$. We further omit
the subscript for the standard case $s=00$.

\begin{figure}
\begin{center}
\includegraphics[height=2.3in]{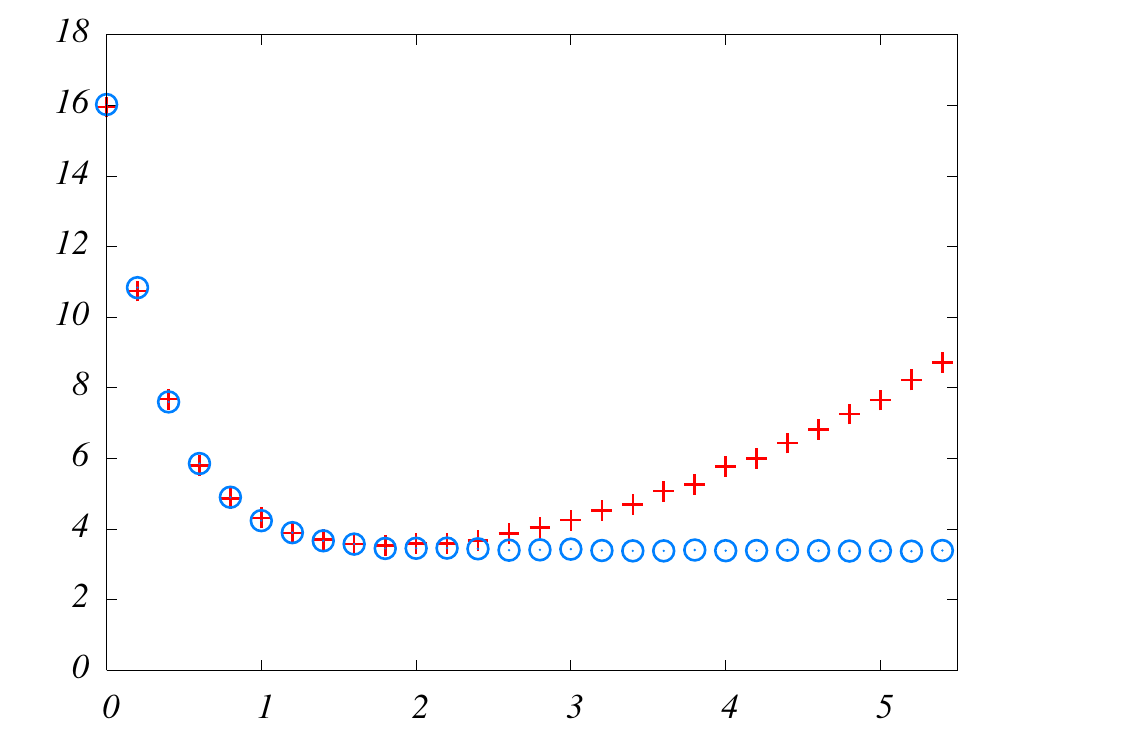}
\put(-140,-5){$\tau$}
\put(-120,85){{\small $r_{++}(N=16;\tau)$}}
\put(-130,25){{\small $r_{+-}(N=16;\tau)$}}
\includegraphics[height=2.3in]{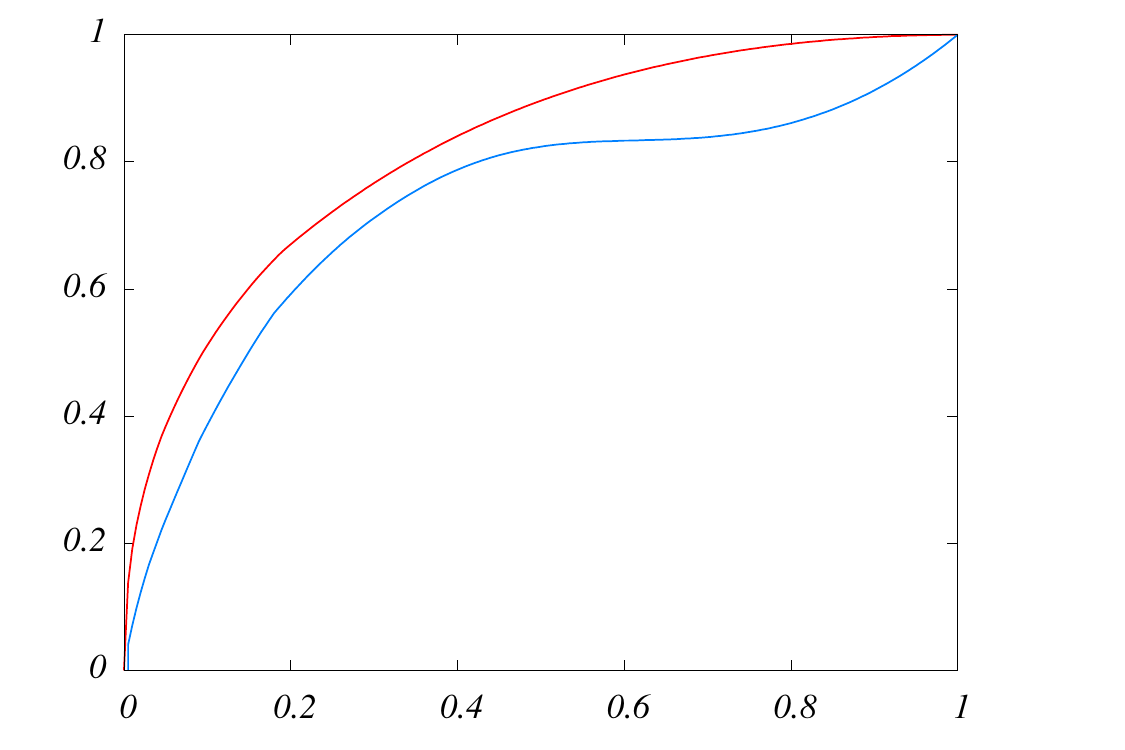}
\put(-140,-5){$\eps$}
\put(-200,135){{\small $M(\eps|TV)$}}
\put(-130,109){{\small $M_{\rm rand}(\eps|TV)$}}
\caption{Left: Risk at zero of total variation denoising, as a function of
  the regularization parameter $\tau\in \bR_+$. The risk was estimated
through Mont Carlo integration.
Right: Upper curve: minimax MSE of total variation denoiser over the class
$\cF_{M,\eps,TV}$ as per Theorem \ref{thm:MMAX_TV}. Lower curve:
minimax MSE over signals with random change points.}
\label{fig:TVRiskatZero}
\end{center}
\end{figure}
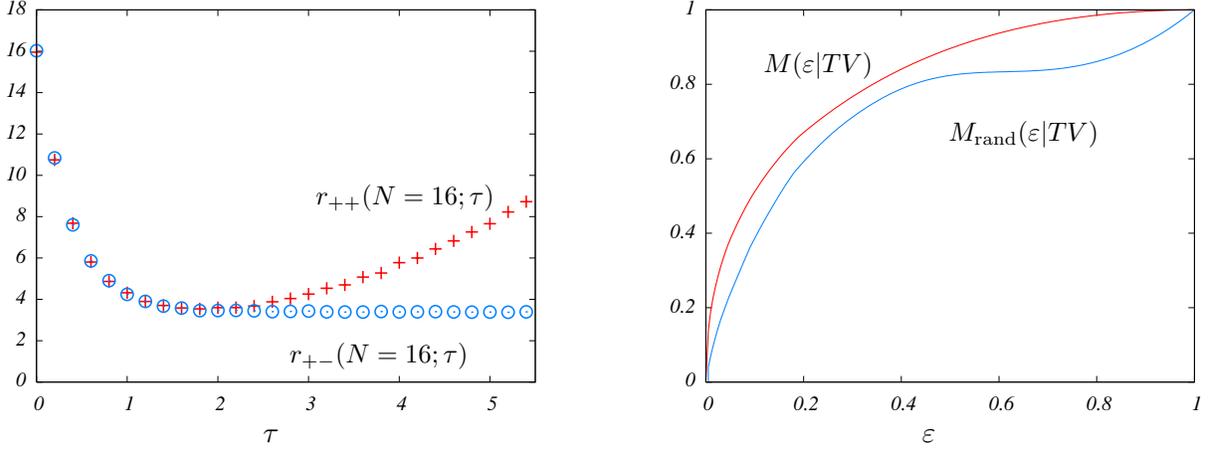
We define the risk at $\mu\in\bR^N$ as
\begin{eqnarray}
R_{N,s}(\mu;\tau) \equiv \E\big\{\big\|\etv_s(\mu+ \bz;\tau)-\mu\big\|^2_2\big\}\, ,
\end{eqnarray}
where $\bz\sim\normal(0,\id_{N\times N})$. We denote the risk at
$0$ by
\begin{align}
r_s(N;\tau) \equiv R_{N,s}(\mu=0;\tau) =  \E\big\{\big\|\etv_s( \bz;\tau)\big\|^2_2\big\}\, .\label{eq:TVRiskAtZero}
\end{align}
Notice that, by symmetry,  $r_s(N;\tau) = r_{-s}(N;\tau)$ and
$r_{s_1,s_2}(N;\tau) = r_{s_2,s_1}(N;\tau)$.
We then have the following analogous of Lemma \ref{lemma:MonotoneRiskAtInfty}.
\begin{lemma}\label{lemma:TVRiskAtInfty}
The risk of total variation regression satisfies the following properties
\begin{enumerate}
\item[$(a)$] The function $t\mapsto R_N(t\,\mu;\tau)$ is monotone
  increasing for $t\in\bR_+$.
\item[$(b)$] Let $I_{\neq}(\mu)\equiv \{i\in [N-1]:\; \mu_i\neq\mu_{i+1}\}$
be the set of change points of $\mu$. Denoting them by 
$I_{\neq}(\mu)\equiv \{i_1,i_2,\dots, i_{K(\mu)}\}$, $i_k\le i_{k+1}$, let $J_k \equiv \{i_k+1,i_k+2,\dots, i_{k+1}\}$ for
$k\in \{0,\dots, K\}$ (with, by convention, $i_0=0$, $i_{K+1}=N$).
Further, for $K\ge 1$, let $s(k) =[\sign(\mu_{i_k}-\mu_{i_k+1}),
\sign(\mu_{i_{k+1}+1}-\mu_{i_{k+1}}) ]$ for    $k\in \{1,\dots,
K-1\}$, $s(0) = [0,\sign(\mu_{i_{1}+1}-\mu_{i_{1}}) ]$, $s(K) =
[0,\sign(\mu_{i_k}-\mu_{i_{k}+1}) ]$.
For $K=0$ we let $s(0)=(0,0)$. 

Then, for any $\mu\in\bR^N$,
\begin{align}
\lim_{t\to\infty} R_N(t\mu;\tau) = \sum_{k=0}^{K(\mu)} r_{s(k)}(|J_{k}(\mu)|;\tau)\, . \label{eq:TVRiskAtInfty}
\end{align}
\end{enumerate}
\end{lemma}
\begin{proof}
The argument in part $(b)$ is essentially the same as in part $(b)$ of
Lemma \ref{lemma:MonotoneRiskAtInfty} and we will therefore omit it. 

For proving part $(a)$, we will prove that, letting
$D(\mu;z)\equiv\|\etv(\mu+z;\tau)-\mu\|_2^2$,
the function $t\mapsto D(t\mu;z)$ is increasing for $t\in\bR_+$.
First  notice that the stationarity condition for the
minimum in Eq.~(\ref{eq:ETV}) reads
\begin{align}
x_i-y_i &= \tau \, v_i-\tau\, v_{i-1}\, ,\label{eq:SubGradient}\\
v_i &= \begin{cases}
+1 & \mbox{ if $x_{i+1}>x_i$,}\\
-1 &  \mbox{ if $x_{i+1}<x_i$,}\\
v_i\in [-1,+1] & \mbox{ otherwise,}
\end{cases}
\end{align}
with the convention that $v_0=v_N=0$. Let $I_{\neq}=I_{\neq}(x)$ and
$J_k, s(k)$ be defined as in part $(b)$ of the statement. Then,
summing Eq.~(\ref{eq:SubGradient}) over $i\in J_k$, we get
\begin{align}
x_i = \frac{1}{|J_k|}\sum_{i\in J_k} y_i + \tau\os(k)\, ,\;\;
\mbox{ for all  } i \in J_k,\label{eq:PieceTV}
\end{align}
where $\os(k) = s_1(k)+s_2(k)$
Hence $\etv(\,\cdot\,;\tau)$  is piecewise affine with components
indexed by $\cJ=\{J_k,s(k)\}_{k\in [K]}$. Within each component,
we have $\etv(y;\tau) = F_{\cJ}(y)$ with $F_{\cJ}$ defined as per
Eq.~(\ref{eq:PieceTV}).

Since $y\mapsto\etv(y;\tau)$ is continuous  (and hence $t\mapsto
D(t\mu;z)$ is), it is sufficient to prov that, letting
\begin{align*}
D_{\cJ}(\mu;z) \equiv \big\|F_{\cJ}(\mu+z)-\mu\big\|_2^2\, ,
\end{align*}
the function $t\mapsto D_{\cJ}(\mu;\,z)$ is monotone increasing for
$t\in\bR_+$.   Using Eq.~(\ref{eq:PieceTV}) we obtain
\begin{align*}
D_{\cJ}(\mu;z) = \sum_{k=0}^K|J_k|\,
\left[\frac{1}{|J_k|}\sum_{i\in J_k}(\mu_i-\omu_{J_k})^2+(\oz_{J_k}+\os(k)\tau)^2\right]\, ,
\end{align*}
where $\ox_{J_k}$ denotes the average of vector $x$ over $J_k$. It
follows that $t\mapsto D_{\cJ}(\mu;\,z)$ is increasing as claimed.
\end{proof}
The risk at $0$, $r_{s}(N;\tau)$, can be computed numerically for
moderate values of $N$. 
Notice that the cases $r_{00}(N;\tau)$, $r_{0\pm}(N;\tau)$, and
$r_{\pm 0}(N;\tau)$ are only relevant for the boundary intervals
$J_0(\mu)$ and $J_{K(\mu)}(\mu)$ and turn out to be immaterial for
the asymptotic minimax risk. Thanks to symmetries, the only relevant
cases are   $r_{++}(N;\tau)$ and $r_{+-}(N;\tau)$. The results of a
numerical computation for these quantities is shown in Figure \ref{fig:TVRiskatZero}.
These calculations suggest $r_{++}(N;\tau)\ge r_{+-}(N;\tau)$, which
is indeed consistent with intuition as boundary conditions $++$ induce
a larger bias. Also, it is easy to prove that  $r_{+-}(N;\tau)\to 0$
as $\tau\to\infty$ (as $\tau\to\infty$, $\etv(y;\tau)$ converges to a constant vector).

Using the last Lemma, and proceeding as in the proof of Theorem
\ref{thm:MMAX_Mono}, it is immediate to obtain a
characterization of the minimax MSE of the total variation denoiser. 
For technical reasons, we need to introduce the class
$\cF_{N,\eps,TV}(L)$ of vectors in $\cF_{N,\eps,TV}$ with distance at
most $L$ between changepoints.
\begin{theorem}\label{thm:MMAX_TV}
The asymptotic minimax MSE of total variation denoiser
\begin{align*}
M_L(\eps|TV)\equiv \lim_{N\to\infty}M(\cF_{N,\eps,TV}(L)|TV)
\end{align*}
exists and is given by
\begin{eqnarray}
       M_L(\eps | TV) = \inf_{\tau\in\bR_+}\max_{\pi} \Big\{ \eps\sum_{\ell, s\in\{++,+-\}} \pi_{\ell,s} r_s(\ell;\tau)\;
       :\;\; \sum_{\ell,s\in\{++,+-\}}   \ell  \pi_{\ell,s}= 1/\eps \Big\}\, ,\label{eq:TV_MMAXFinal}
\end{eqnarray}
where the maximization is over the probability distribution
$\{\pi_{\ell,s}\}_{1\le \ell\le L,s\in\{++,+-\}}$.

For any $\xi>0$, $\eps\in (0,1)$, the following distribution
$\nu_N^{(\eps,\xi)}\in \cF_{N,\eps,TV}(L)$ has risk larger that
$M_L(\eps|TV)-\xi$ for all $N$ large enough. A signal $\bx\sim
\nu^{(\eps,\xi)}_N$ has $|X_{i+1}-X_{i}|= \Delta>0$ at all change
points $i\in [N-1]$ for some $\Delta = \Delta(\xi)$ large
enough. Further, for an interval $J$ between change points, let the
\emph{type} of $J$ be $(\ell,s)$ with $\ell\in\naturals$ its length
and $s\in\{++,+-\}$ depending whether the adjacent change points are
both increase or decrease points ($+-$) or not $(++)$. Then the
empirical distribution of types under $\nu_N^{\eps,\xi}$ is given by $\pi$ 
solving the saddle point problem in Eq.~(\ref{eq:MonoRegMMAXFinal}).
\end{theorem}
We omit this proof since it is an immediate generalization of the one
of Theorem \ref{thm:MMAX_Mono}. Notice that $M_L(\eps)$ is monotone
increasing in $L$ and hence admit a limit as $L\to\infty$. We expect
that
\begin{align*}
M(\eps|TV) = \lim_{L\to\infty}M_{L}(\eps|TV)\, .
\end{align*}
This limit can be evaluated numerically and in Figure
\ref{fig:TVRiskatZero} we plot the resulting minimax risk. 

Notice that, by properly modifying Eq.~(\ref{eq:TV_MMAXFinal}), one
obtains the minimax risk over subsets of $\cF_{N,\eps}$ with constrained
change point distributions.  For instance, we can consider the case in
which the lengths between change points are distributed as for
uniformly random change points, and increase/decrease points are
alternating. We then get
\begin{align*}
\pi^{\eps}_{\ell,++} = \eps (1-\eps)^{\ell-1}\, ,\;\;\;\;\;\;
\pi^{\eps}_{\ell,+-} = 0 \, ,\;\;\;\;\;\;\ell\ge 1\, .
\end{align*}
We consequently define the random changepoint minimax risk as
\begin{eqnarray}
       M_{\rm rand}(\eps | TV) = \inf_{\tau\in\bR_+} \Big\{ \eps\sum_{\ell} \pi_{\ell,s}^{\eps} r_s(\ell;\tau)\; \Big\}\, . \label{eq:MTVRand}
\end{eqnarray}
This curve is plotted in Figure \ref{fig:TVRiskatZero} for comparison.

\subsection{Empirical phase transition behavior}

\begin{figure}
\begin{center}
\includegraphics[height=3in]{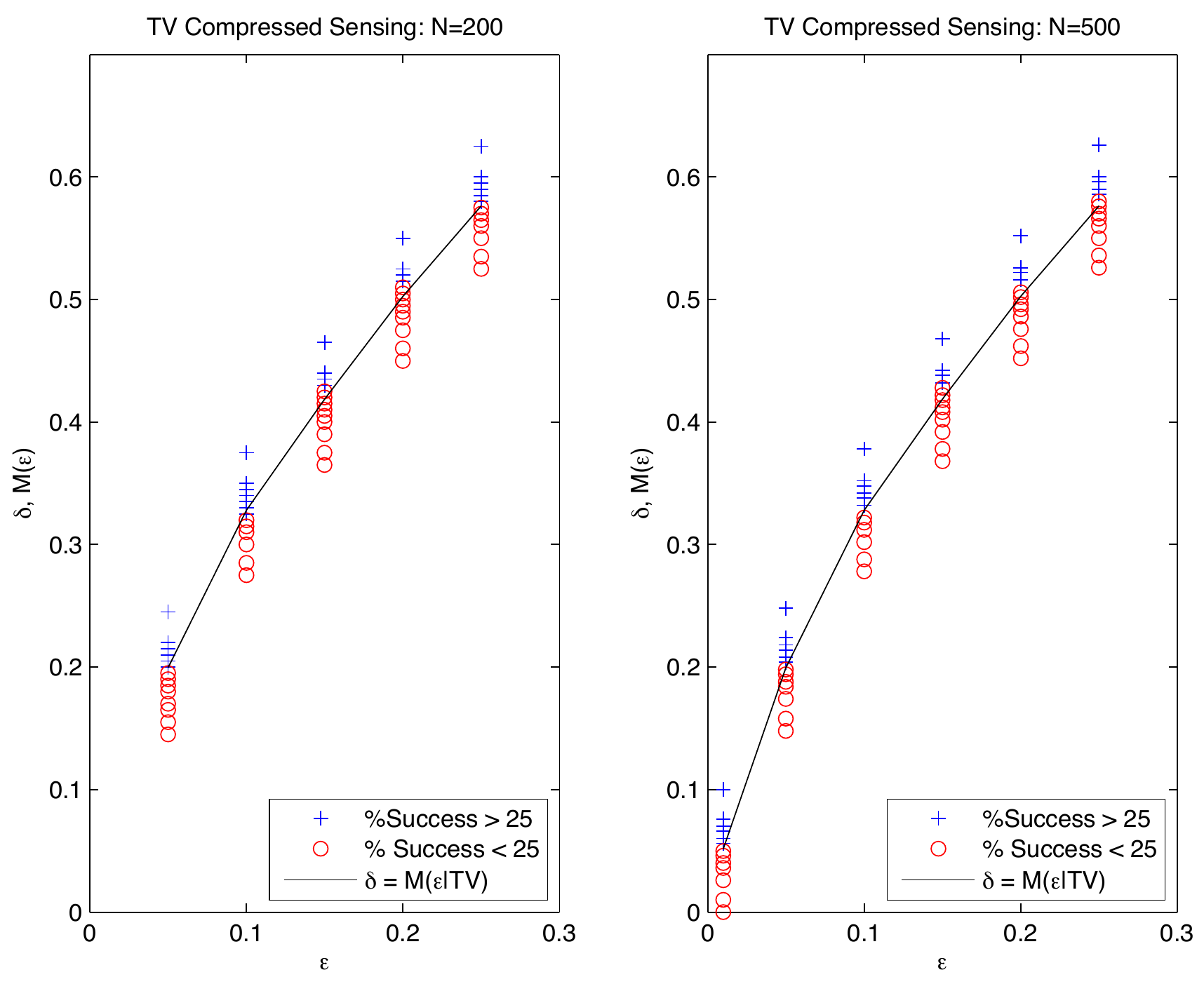}
\caption{Results for TV AMP under the random changepoint signal prior.   
Here $\delta = n/N$ is the undersampling fraction
$\eps$ is the generalized sparsity measure (i.e. the fraction of
change points in the signal to be reconstructed). 
Circles:  less than 25\%  fraction of correct recovery.
Crosses: greater than 25\% fraction of correct recovery.  Curve:
minimax MSE curve for random changepoints $\delta = M_{\rm random}(\eps | TV)$.
Left frame:  $N=200$. Right frame: $N=500$.}
\label{fig:TVAMPPT}
\end{center}
\end{figure}
We implemented the   the AMP iteration (\ref{AMPA})-(\ref{AMPC}),
using the total variation denoiser $\eta(\,\cdot\, ;\tau,\sigma) =
\etv(\,\cdot\,;\tau,\sigma)$. 
For the latter, we used the software package {\sf tvdip} (in the
Matlab implementation)
or the projected Newton method \cite{TVSurvey} (in the Java implementation).

For $x\in\bR^N$ be $K_0(x)$ denote the number of constant segments in
$x$ or, equivalently, the number of change points in $x$, plus one.
We then have the following expression for the memory term in Eq.~(\ref{AMPA})-(\ref{AMPC}):
\begin{align}
\onsager_t \equiv \left.\frac{1}{n} \, \div\,\etv(y;\tau,\sigma_{t-1})\right|_{y=y^{t-1}} =
\frac{1}{n}\, K_0(\etv(y^{t-1};\tau,\sigma^{t-1}))\, .
\end{align}
We will refer to this specific version of AMP as to \emph{TV-AMP}.

\begin{figure}
\begin{center}$
\begin{array}{ccc}
\includegraphics[height=2.1in]{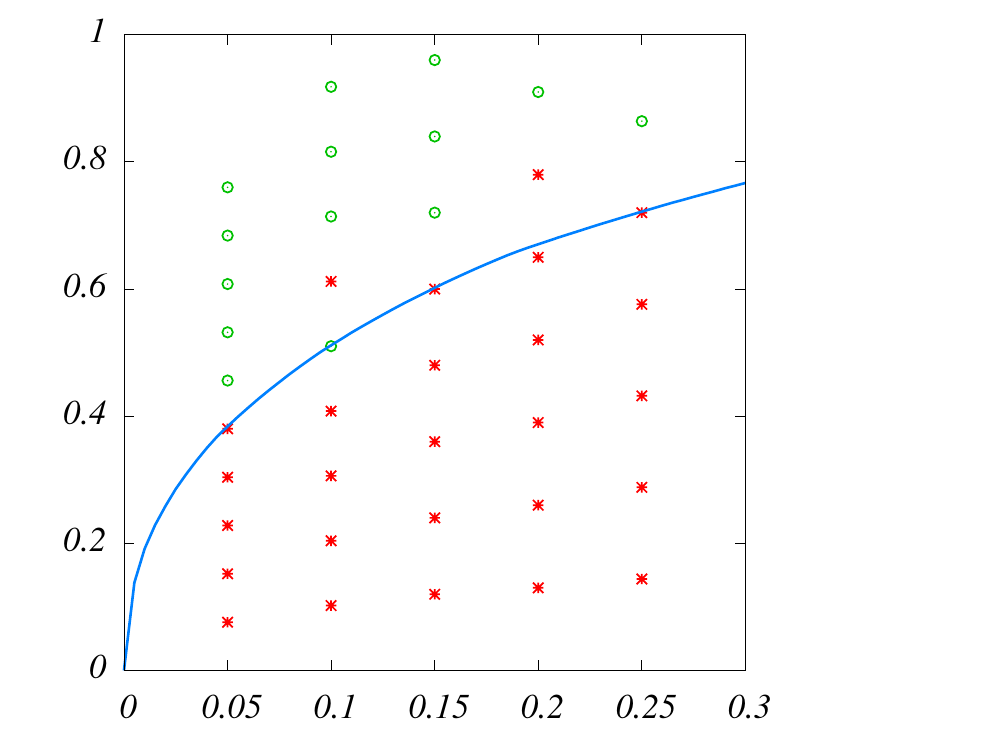}
&
\hspace{-2cm}\includegraphics[height=2.1in]{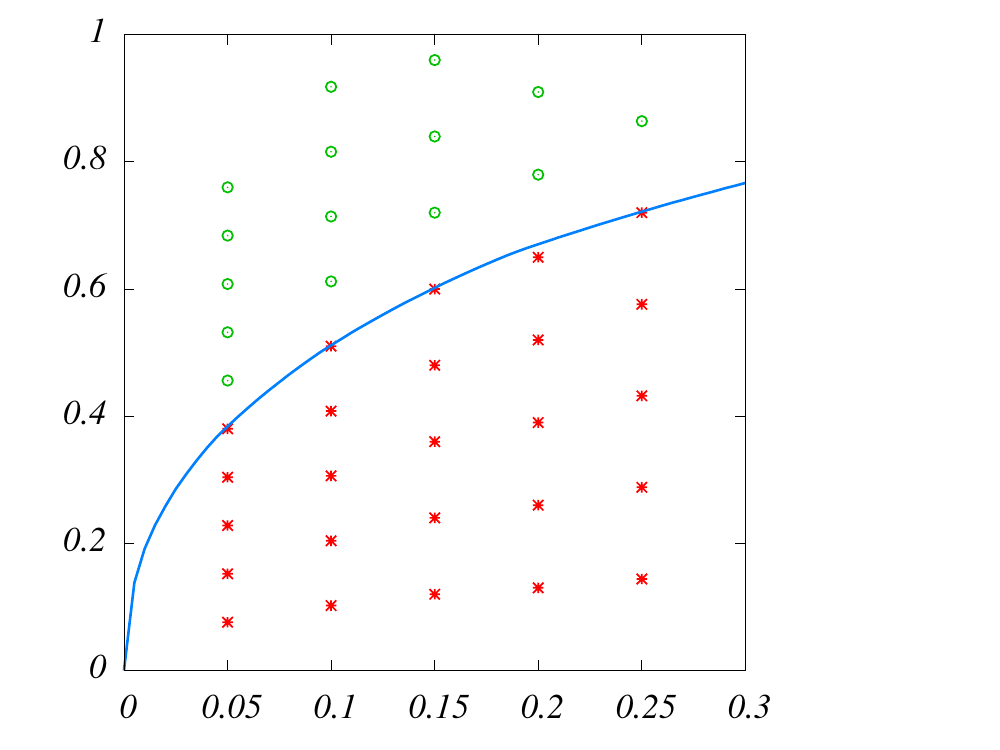}
&
\hspace{-2cm}\includegraphics[height=2.1in]{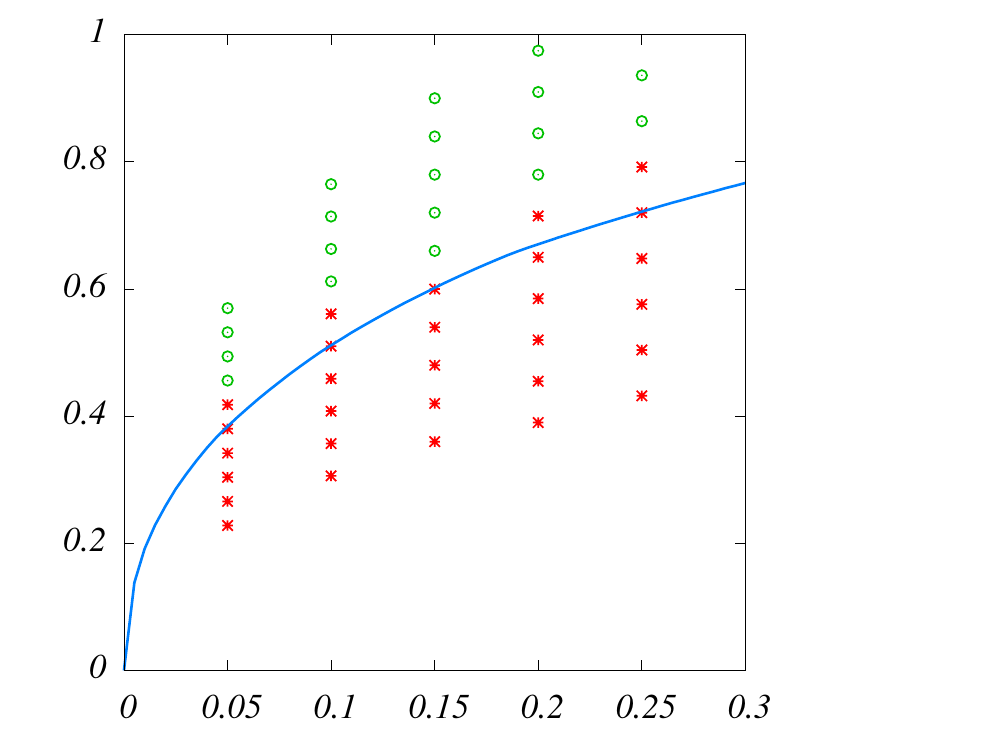}
\end{array}$
\put(-125,-75){$\eps$}
\put(-280,-75){$\eps$}
\put(-435,-75){$\eps$}
\put(-510,10){$\delta$}
\end{center}
\caption{Results for TV AMP under the least favorable signal prior. Here $\delta = n/N$ is the undersampling fraction
$\eps$ is the generalized sparsity measure (i.e. the fraction of
change points in the signal to be reconstructed). 
Red stars:  less than 50\%  fraction of correct recovery.
Green circles: greater than 50\% fraction of correct recovery.  Curve:
minimax MSE curve $\delta = M(\eps | TV)$.
The three frames correspond to (from left to right) $N=100$, $250$, $500$.}
\label{fig:TVAMPPT2}
\end{figure}
We carried out two types of experiments. 
In the first class of experiments we considered signals $x$ with 
distances between change points distributed as 
\begin{align*}
\pi^{\eps}_{\ell,++} = \pi^{\eps}_{\ell,+-} =\frac{1}{2}\, \eps(1-\eps)^{\ell-1}\,
.
\end{align*}
This is the same distribution as if each position is independently an increase point
or a decrease point, each with probability $\eps/2$.
The predicted phase transition curve is given by
Eq.~(\ref{eq:MTVRand}) and the minimax value of $\tau$ is used in the
AMP implementation. The simulations results are presented in Figure
\ref{fig:TVAMPPT} for $N=200$, $500$ and show good agreement between
predictions and observations. In this case we used the
Hamming metric (\ref{eq:HammingDef}), because the norm of the typical
signal  $\|x_0\|_2^2$ scales super-linearly in $N$. 

In the second set of experiments we used the distribution
$\{\pi_{\ell,s}\}_{s\in\{++,+-\},\ell\in \naturals}$ achieving the
$\sup$ in Eq.~(\ref{eq:TV_MMAXFinal}) for $L$ large. More precisely,
we fixed $L=25,50,100$ and solved numerically the optimization problem
(\ref{eq:TV_MMAXFinal}). The solution appears to be independent of $L$
and put small weight on large lengths $\ell$. The least favorable
distribution for $L = 50$ was used to generate signals. 
The success probability is compared
with the predicted phase transition curve $M(\eps|TV)$ in Figure
\ref{fig:TVAMPPT2} for $N=100$, $250$, $500$.
In this case we used the MSE metric and declared success if
$\|\hx^5-x_0\|_2^2/\|x_0\|^2_2\le 0.001$. 
%
%
\section{Characterization of the phase transition using state evolution}
\label{sec:DerivePT}

In this section we prove the basic relation (\ref{generalPT}) between
minimax mean square error in denoising and the phase-transition
boundary in the sparsity-undersampling plane. Our proof assumes that 
the state evolution formalism developed in
\cite{DMM09,DMM_ITW_I,DMM-NSPT-11} holds, in the precise terms stated below.
This formalism was established rigorously for separable denoisers (under
additional regularity assumptions) in \cite{BM-MPCS-2011}.

A crucial observation  for state evolution is that 
the mean squared error of the AMP reconstruction $x^t$ at iteration
$t$ is practically non-random for large system sizes $N$, and has a
well-defined limit as $N\to\infty$. In particular, the limit
\[
           m_t = \lim_{N \goto \infty} \frac{1}{N}\,\| x_0 - x^t \|_2^2 ,
\]
exists almost surely (here we assume $n\to\infty$, while $n/N\to\delta$).
Moreover, the evolution of $m_t$ with increasing $t$ 
is dictated by a formula $m_{t+1} = \Psi(m_t)$  which is explicitly
computable, and defined below.
We will use the term \emph{state evolution} to refer both to the
mapping $m\mapsto \Psi(m)$ and to the sequence $\{m_t\}_{t\ge 0}$
with appropriate initial condition. 
State evolution allows to determine whether AMP recovers the signal
$x_0$ correctly, by simply checking whether $m_t  \goto
0$ as $t \goto \infty$ (in which case the MSE vanishes asymptotically)
or not. The latter problem does in turn reduce to a 
problem in real analysis.

The papers \cite{DMM09,DMM_ITW_I,DMM-NSPT-11} 
developed the state evolution framework for 
separable denoisers
and verified its predictions numerically for three specific
examples (namely for the shrinkers \Soft, \SoftPos, and \Cap). However, the heuristic
argument presented in those papers was much more general. 
Indeed, \cite{BM-MPCS-2011} proved that state evolution holds, in a
precise asymptotic sense, for Gaussian measurement matrices $A$ with
iid entries and generic separable denoisers, under mild regularity assumptions.
A generalization to non-gaussian entries was subsequently proved in \cite{BLM-Universality}.

Here we generalize this approach to non-separable denoisers $\eta(\,\cdot\,;\tau,\sigma): \bR^N
\mapsto \bR^N$.
This framework covers all the shrinkers discussed in Sections \ref{sec-Scalar} to \ref{sec:TV},
and yields a formal proof of the main formula (\ref{generalPT}),
under the assumption that indeed state evolution is correct in this
broader context.  We will throughout assume the scaling relation (\ref{eq:ScalingRelation}).

In the next sections we will first introduce some basic notations and
facts 
about state evolution. Then we will prove the phase transition expression
 (\ref{generalPT}) by establishing first a lower bound and then a
 matching upper bound, both given by the minimax MSE.

\subsection{State Evolution}
\label{subsec:SE}

The next definition provides the suitable generalization of the state
evolution mapping to the present setting.
\begin{definition}
For given $\delta, \tau\ge 0$, and $\nu = \{\nu_N\}_{N\in\naturals}$ a sequence of probability
distributions over $\reals^N$, 
define the \emph{state evolution mapping} $\Psi(\,\cdot\,;\delta,\tau,\nu): \bR \mapsto \bR$ by
\begin{align}
\Psi(m; \delta, \tau, \nu) \equiv \lim_{N\to\infty}\frac{1}{N}\E_{\nu_N} \Big\{
\Big\| \bx- \eta\Big(\bx+\sqrt{\frac{m}{\delta}} \,\bz ;  \tau,
\sqrt{\frac{m}{\delta}}\Big) \Big\|_2^2\Big\} \, ,\label{eq:SEMapDef}
\end{align}
whenever the limit on the right-hand side exists.
Here, as before, $\bx$ and $\bz$ are independent  vectors, 
$\bx \sim \nu_N$, $\bz\sim \normal(0,\id_{N\times N})$.
 In other words, $\Psi(m; \delta, \tau, \nu)$ is the 
per-coordinate MSE of denoiser $\eta$ at noise level  $\sigma^2 \equiv m /\delta$.

Given implicit fixed parameters $( \delta, \tau, \nu)$, \emph{state evolution}
is the one-dimensional dynamical system: $m_{t+1} = \Psi(m_t)$ starting from
$m_0 = \lim_{N\to\infty}n^{-1} \E_{\nu_N} \|\bx\|_2^2$.
\end{definition}

The fixed points of the mapping $m\mapsto\Psi(m)$ play of course a
crucial role in the analysis of state evolution.
\begin{definition}
The \emph{highest fixed point} of the mapping $\Psi(\, \cdot\,)
=\Psi(\,\cdot\,; \delta, \tau, \nu)$  
is defined as $\HFP(\Psi) \equiv \sup\{ m : \Psi(m) \geq m \}$.
\end{definition}
The importance and applicability this notion is underscored by the
next two observations.
Here and below we say that a function $f:\reals_+\to\reals$ is
\emph{starshaped} if $x\mapsto f(x)/x$ is decreasing.
\begin{lemma}\label{lemma:HFP}
 Suppose that $m_0>0$ and any one of these three conditions holds:
\begin{enumerate}
\item[(a)]  $\Psi(m)$ is an increasing function of $m$, and the initial
  condition of state evolution satisfies $m_0  \ge  \HFP(\Psi)$; 
\item[(b)]  $\Psi(m)$ is an increasing  starshaped function of $m$.
\item[(c)] $\HFP(\Psi) = 0$.
\end{enumerate}
Then  state evolution converges to the
highest fixed point:
\[
     \lim_{t \goto \infty} m_t = \HFP(\Psi).
\]
Further, if $\HFP(\Psi)>0$ and $\Psi(m)$ is a  starshaped function of
$m$, then   $\lim\inf_{t \goto \infty} m_t >0$.
\end{lemma}
The proof is a standard calculus exercise; we
omit it.

\begin{lemma}\label{lemma:Star}
The function $m\mapsto \Psi(m)$ is  starshaped for all of
the following choices of the denoiser $\eta$:
\bitem
  \item Soft thresholding.
  \item Positive soft thresholding.
  \item Block soft thresholding.
  \item James-Stein shrinkage.
  \item Monotone regression denoiser.
  \item Total variation denoiser.
 \eitem
\end{lemma}
The proof of this Lemma is deferred to Appendix \ref{Appendix:RiskProperties}.

\subsection{State Evolution Phase Transition}

Consider a collection $\cF_{N,\eps}$ of probability distributions over
$\reals^N$, indexed by $\eps\in [0,1]$ as per Section
\ref{sec:SignalModels} (these do not need to be simple sparse signals).
For a sequence of probability distributions $\nu =
\{\nu_N\}_{N\in\naturals}$ we write $\nu\in\cF_{\eps}$ if 
 $\nu_N\in \cF_{N,\eps}$ for all $N$ \emph{and} the limit on the
 Eq.~(\ref{eq:SEMapDef}) exists for each $m\in\bR_+$.
Letting  $\HFP(\delta,\tau,\nu) = \HFP(\Psi(\;\cdot
\; ;\delta,\tau,\nu))$,  we consider the minimax value:
\[
     \HFP^*(\eps,\delta) = \inf_{\tau\in\reals_+}  \sup_{\nu  \in \cF_{\eps}}  \HFP( \delta,\tau,\nu).
\]
\begin{definition}\label{def:deltaSE}
For $\eps\in [0,1]$,  define  the state evolution phase transition as
\[
\delta_{SE}(\eps|\eta)\equiv  \inf \big\{ \delta\ge 0 : \; \HFP^*(\eps,\delta) = 0  \big\} .
\]
\end{definition}
Note that $\HFP^*(\eps,\delta)$ is monotone decreasing as a function
of $\delta$, by  definition of the state evolution mapping $\Psi$,
cf. Eq.~(\ref{eq:SEMapDef}).
It follows that $\HFP^*(\eps,\delta) = 0 $  for $\delta >
\delta_{SE}(\eps)$ and  $\HFP^*(\eps,\delta) > 0 $  for $\delta <\delta_{SE}(\eps)$.
Further, by nestedness, it is monotone increasing as a function of  $\eps$,
which implies that $\eps\mapsto \delta_{SE}(\eps)$ is monotone increasing.
The rationale for this definition is that, for
$\delta>\delta_{SE}(\eps)$ and under any of the assumptions of  Lemma
\ref{lemma:HFP}, state evolution predicts that AMP will 
correctly recover the signal $x_0$.

\begin{theorem}\label{thm:CSeqDenoising}
Let $\cF_{N,\eps}$ be a nested, scale-invariant collection of probability distributions,
and assume that the shrinker $\eta(\,\cdot\,;\tau,\sigma)$ obeys the
scaling relation  (\ref{eq:ScalingRelation}).
Define the minimax MSE $M(\eps | \eta)$ as per
Eq.~(\ref{eq:generalMAsymptotic}). Then
\[
        \delta_{SE}(\eps|\eta)   = M(\eps|\eta)\, .
\]
\end{theorem}

In order to prove this result, we shall first establish a more
general fact. 
Given a sequence of distributions $\nu=\{\nu_N\}_{N\in\naturals}$ and, for any $\tau\in\Theta$,
we let
\begin{eqnarray}
\delta_{SE}(\tau,\nu|\eta) \equiv \inf\big\{ \delta\ge 0\;\; \mbox{ such that
}\;\;\HFP(\delta,\tau,\nu) = 0\big\}\, .
\end{eqnarray}
The rationale for this definition is clear. Under the assumption
that state evolution holds, for a  signal $x_0$ sampled from distribution
$\nu_N$, AMP (with tuning parameter $\tau$) is guaranteed to reconstruct $x_0$
if and only if  $\delta>\delta_{SE}(\tau,\nu|\eta)$. 

Analogously, we let
\begin{eqnarray}
   M(\tau,\nu|\eta) =  \sup_{\sigma>0}\lim_{N\to\infty}\frac{1}{N\sigma^2}\E_{\nu_N} \Big\{\big\|\bx
   - \eta(\bx+ \sigma \bz; \tau,  \sigma)\big\|_2^2\Big\} \, .
\end{eqnarray}
 the normalized  MSE for denoising with worst case case
 signal-to-noise ratio.
With these definitions we have the following.
\begin{lemma}\label{lemma:IdentityPerNu}
For any sequence of probability measures $\nu=\{\nu_N\}_{N\ge 1}$, and any
$\tau\in\Theta$, we have
\begin{eqnarray}
\delta_{SE}(\tau,\nu|\eta) = M(\tau,\nu|\eta)\, .
\end{eqnarray}
 \end{lemma}
\begin{proof}
By definition
\begin{eqnarray*}
\delta_{SE}(\tau,\nu|\eta) &=& \inf\Big\{\delta>0\;\;\mbox{ such that}\;\;
\Psi(m;\delta,\tau,\nu)<m \mbox{ for all } m>0\Big\}\\
&= &\inf\Big\{\delta>0\;\;\mbox{ such that}\;\;
\sup_{m>0}\frac{1}{m}\Psi(m;\delta,\tau,\nu)<1\Big\}\\
&= &\inf\Big\{\delta>0\;\;\mbox{ such that}\;\;
\sup_{m>0}\lim_{N\to\infty}\frac{1}{mN}\E\{\|\bx-\eta(\bx+\sqrt{m/\delta}
\bz;\tau,\sqrt{m/\delta})\|^2\}<1\Big\}\\
& = &\inf\Big\{\delta>0\;\;\mbox{ such that}\;\;
\sup_{\sigma>0}\lim_{N\to\infty}\frac{1}{\delta \sigma^2 N}\E\{\|\bx-\eta(\bx+\sigma
\bz;\tau,\sigma)\|^2\}<1\Big\}\\
& = &\inf\Big\{\delta>0\;\;\mbox{ such that}\;\;\frac{1}{\delta
}M(\tau,\nu|\eta)<1\Big\} = M(\tau,\nu|\eta)\, .
\end{eqnarray*}
\end{proof}

We are now in position to prove Theorem \ref{thm:CSeqDenoising}.
\begin{proof}[Proof of Theorem \ref{thm:CSeqDenoising}]
Throughout the proof we drop the argument $\eta$, since this is kept constant.

Define $\delta_*(\eps)\equiv
\inf_{\tau\in\Theta}\sup_{\nu\in\cF_{\eps}}\delta(\tau,\nu)$.
We claim that $\delta_*(\eps) = \delta_{SE}(\eps)$. Indeed choose
$\delta\in [0,\delta_{SE}(\eps))$.  Then by definition there exists $m>0$ such
that, for all $\tau\in\Theta$,  there exists $\nu\in\cF_{\eps}$ with
$\HFP(\delta,\tau,\nu)>m$.
Hence,  for all $\tau\in\Theta$,  there exists $\nu\in\cF_{\eps}$ with
$\delta(\tau,\nu)\ge\delta$. This implies that, for all
$\tau\in\Theta$, $\sup_{\nu\in\cF_{\eps}}\delta(\tau,\nu)>\delta$,
i.e.  $\delta_*(\ve)\ge\delta$.
Proceeding in the same way, it is immediate to prove that, for any
$\delta\in[0,\delta_*(\eps))$, we have $\delta_{SE}(\ve)>\delta$.
Hence, we conclude that $\delta_{SE}(\eps) = \delta_*(\eps)$ as claimed.

To conclude the proof, we note that, by Lemma
\ref{lemma:IdentityPerNu},
we have 
\begin{eqnarray*}
\delta_*(\eps) &=& \inf_{\tau\in\Theta}
\sup_{\nu\in\cF_{\eps,B}}M(\tau,\nu) \\
&=&\inf_{\tau\in\Theta}
\sup_{\nu\in\cF_{\eps,B}}  \sup_{\sigma>0}
\lim_{N\to\infty}\frac{1}{N\sigma^2}\E_{\nu_N} \Big\{\| \bx
   - \eta(\bx+ \sigma\, \bz; \tau,  \sigma)\|_2^2\Big\}\\
&=&\inf_{\tau\in\Theta}
\sup_{\nu\in\cF_{\eps}}\frac{1}{N}\E_\nu \Big\{\| \bx
   - \eta(\bx+ \bz; \tau)\|_2^2\Big\} \, ,
\end{eqnarray*}
where the last equality follows from the scale invariant property of
$\cF_{N,\eps}$.
The last quantity is nothing but $M(\eps)$.
\end{proof}

\subsection{Non-convergence of state evolution}

By Definition \ref{def:deltaSE} and Lemma \ref{lemma:HFP}.(c) it
follows that, for all $\delta<\delta_{SE}(\eps|\eta)$,
all probability distributions $\nu\in \cF_{\eps}$, and all initial
conditions $m_0$, state evolution converges to the zero error fixed
point, namely $\lim_{t\to\infty}m_t=0$. Viceversa, for  $\delta>\delta_{SE}(\eps|\eta)$,
there exists $\nu\in \cF_{\eps}$, and an initial
condition $m_0$, such that $\lim_{t\to\infty}m_t=m_\infty>0$. The reader
might wonder whether this conclusion (non-convergence to $0$) 
also holds if we use the initial condition that is relevant for AMP,
i.e. if,  for  any $\delta>\delta_{SE}(\eps|\eta)$, there exists $\nu\in \cF_{\eps}$
such that, taking  $m_0=\lim_{N\to\infty}N^{-1}\E_{\nu_N}\{\|\bx\|_2^2\}$, we have
$\lim_{t\to\infty}m_t=m_\infty>0$.

Lemmas \ref{lemma:HFP}.(b) and \ref{lemma:Star} immediately imply that
the answer is positive for soft thresholding, positive soft
thresholding, block soft thresholding, James-Stein shrinkage, monotone
regression and
total variation denoising.
It turns out that the answer is positive also for firm
thresholding and the global minimax denoiser. In Appendix \ref{app:Sub}
we describe the argument for these cases. 

\section{Phase transitions for other algorithms}
Formula (\ref{generalPT}) connects an algorithmic property -- phase transitions of
AMP recovery algorithms -- with a property from statistical decision theory -- minimax mean squared error in denoising.

We want to explore a further connection, relating the behavior of convex optimization with
that of certain AMP algorithms. As proved in
\cite{BayatiMontanariLASSO}, in the large system limit
AMP with soft thresholding denoiser effectively computes the solution to
\[
       (P_{1,\lambda}) \qquad    \mbox{minimize}\;\;\;\;\frac{1}{2} \|y -
       Ax\|_2^2 + \lambda \| x \|_1 ,\qquad x\in\bR^N\, .
\]
for an appropriately calibrated $\lambda = \lambda(\tau)$. 

More generally, consider a generalized reconstruction method of the
form 
\begin{eqnarray}
  (P_{J}) \qquad   \mbox{minimize}\;\;\;\; 
\frac{1}{2} \|y - Ax\|_2^2+ \lambda J(x), \qquad x \in\bR^N\, ,
\end{eqnarray}
where  $J:\reals^N\to\reals$ is a convex penalization. To this
reconstruction problem, we can associate an
AMP algorithm, by using the denoiser $
\eta^J(\,\cdot\,;\tau)$ in Eq.~(\ref{AMPA}), (\ref{AMPB}), (\ref{AMPC}), whereby
\begin{eqnarray}
\eta^J(y;\tau) \equiv \arg\min_{x\in\bR^N}\Big\{\frac{1}{2}\|y-x\|^2_2 +
\tau\, J(x)\Big\}\, .
\end{eqnarray}
(we also let $\eta^J(\,\cdot\,;\tau,\sigma) =
\eta^J(\,\cdot\,;\tau\cdot\sigma)$). 
In other words $\eta^J$ is the proximal operator of the penalization $J(\,\cdot\,)$.
We will refer to this algorithm as
to AMP-$J$.

We then have the following general correspondence, which 
follows immediately by writing the stationarity condition of problem
$P_J = P_J(\lambda)$ and the fixed points of AMP-$J$ (see \cite{MontanariChapter}).
\begin{proposition}
Any fixed point $x^{\infty}$ of AMP-$J$ with fixed point parameters
$\tau_{\infty}$, $\onsager_{\infty}$, $\sigma_{\infty}$ corresponds to a stationary point
of $P_J(\lambda)$ with $\lambda$ given by
\begin{eqnarray}
\lambda = \tau_{\infty}\sigma_{\infty} (1-\onsager_{\infty})\, ,\qquad
\onsager_\infty =   \left.\frac{1}{n} \,\div\, \eta(y; \tau^{\infty}) 
\right|_{y = y^{\infty}}
\end{eqnarray}
In particular, if the regularizer $J$ is convex, then fixed points
correspond to minimizers.
\end{proposition}
\begin{example}
The fixed points of AMP with positive soft thresholding denoiser are
solutions of
\[
       (P_{1,\lambda}^+) \qquad   \mbox{minimize}\;\;\;\; 
\frac{1}{2}  \|y - Ax\|_2^2+ \lambda \<1,x\> , \qquad x \in\bR_+^N,
\]
where $\<\,\cdot\,,\,\cdot\,\>$ denotes the standard scalar product
over $\bR^N$, and $1$ is the all-one vector. 
\end{example}
\begin{example}
The fixed points of AMP with capping denoiser effectively are
solutions of
\[
       (P_{\Box}) \qquad   \mbox{minimize}\;\;\;\;   \frac{1}{2}\|y -
       Ax\|_2^2,\qquad x \in [0,1]^N. 
\]
\end{example}

For  noiseless compressed sensing reconstruction, the correspondence
involves the limit $\lambda\to 0$ of the above problem,
that is
\begin{eqnarray*}
\mbox{minimize}&& J(x)\,\;\;\;\;\;\; x\in\bR^N\, ,\\
\mbox{subject to}&& y=Ax\, .
\end{eqnarray*}
Phase transitions for such convex programs were characterized in
\cite{Do,DoTa08,DoTa10b} for the three examples mentioned above, using
methods from combinatorial geometry.
Thus, whenever AMP converges with high probability, formula (\ref{generalPT})
connects fundamental problems of minimax decision theory to 
fundamental problems in high dimensional combinatorial geometry.

We wil next  verify numerically this correspondence beyond the three
classical examples (again focusing
on noiseless measurements).
In the block-structured case, we can compare
block-soft AMP to the following convex optimization problem:
\begin{eqnarray*}
\peeTwoOne \quad  \mbox{minimize}&&\| x \|_{2,1}\equiv \sum_{m=1}^M \| block_m(x) \|_2 \,
,\\
  \mbox{ subject to }&& y = A x\,  .
\end{eqnarray*}
The norm $ \| x \|_{2,1} $ is called the mixed $\ell_{2,1}$ norm.
Figure \ref{fig:P21PT} presents empirical reconstruction results
obtained by solving $(P_{2,1})$. These experiments verify that the phase transition occurs
around the location predicted by (\ref{generalPT}), just as with block soft thresholding AMP.

Analogously,  we can compare
monoreg AMP to the following convex optimization problem:
\begin{eqnarray*}
(P_{{\rm mono}}) \quad  \mbox{minimize} && \<1, \Delta x \>\, \\
  \mbox{ subject to }&& y = A {x} , \;\; \;\;\Delta x \geq 0,
\end{eqnarray*}
where $\Delta x = (\Delta x_1,\dots ,\Delta x_N)$,
$\Delta x_i = (x_{i+1} - x_i )$.
Figure \ref{fig:PmonoPT} verifies that the phase transition occurs
around the location predicted by (\ref{generalPT}), just as with monoreg AMP.

\begin{figure}
\begin{center}
\includegraphics[height=3in]{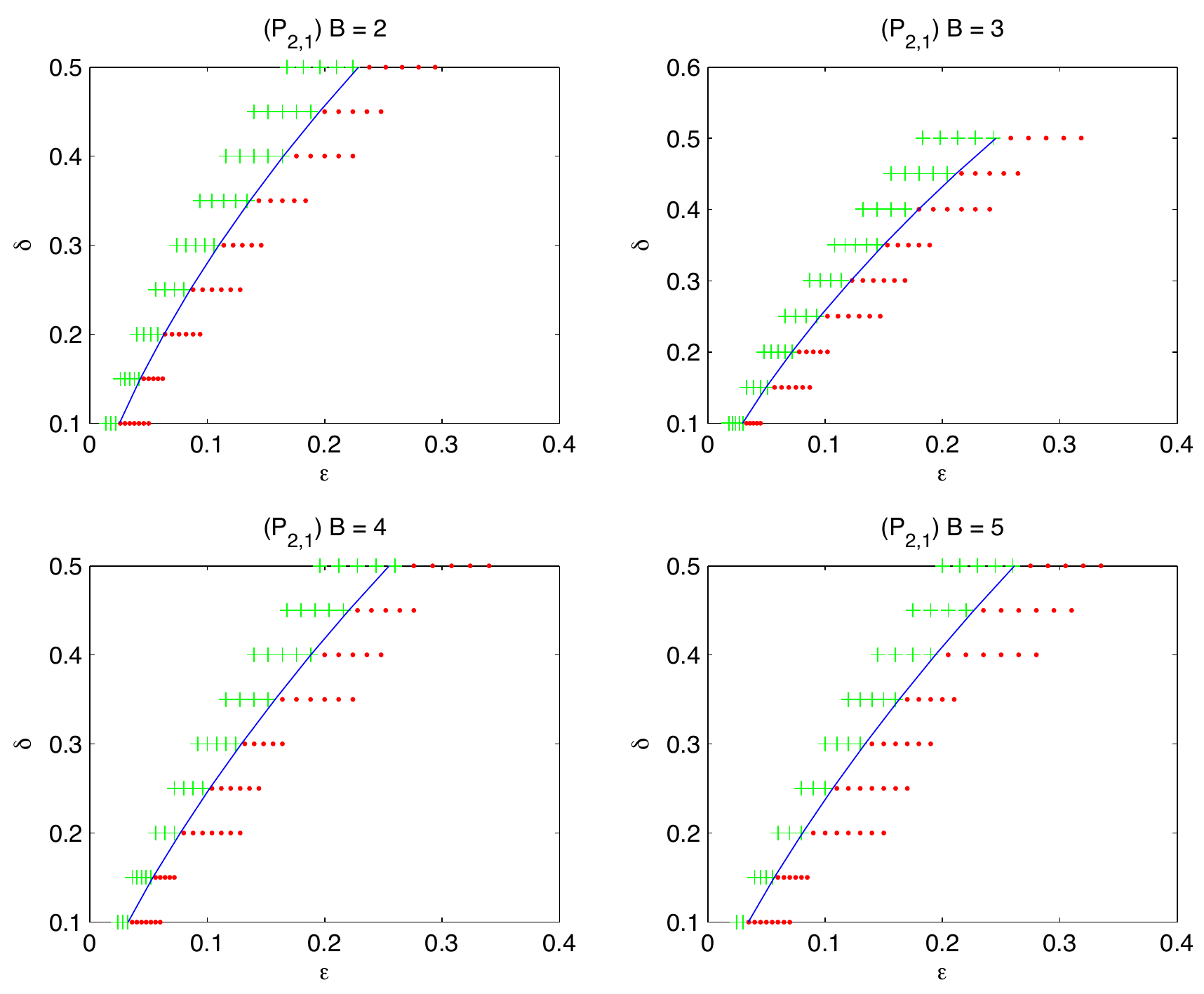}
\caption{Phase transition results for reconstructing block-sparse
  signals using the convex program $\peeTwoOne$. Here  $\delta = n/N$ is the
  undersampling fraction and 
$\eps = k/N$ is the sparsity measure. 
Here the signal dimension is $N=1000$,
$\delta = n/N$ 
is the undersampling fraction, and $\eps$ is the sparsity parameter
(fraction of non-zero entries). 
Red:  probability of correct
recovery smaller than 50\%.
Green:  probability of correct recovery larger than 50\%.  The problem
$(P_{2,1})$ was solved using Sedumi. 
The curve separating success from failure is the minimax MSE curve $\delta = M(\eps | \Blocksoft)$.}
\label{fig:P21PT}
\end{center}
\end{figure}

\begin{figure}
\begin{center}
\includegraphics[height=3in]{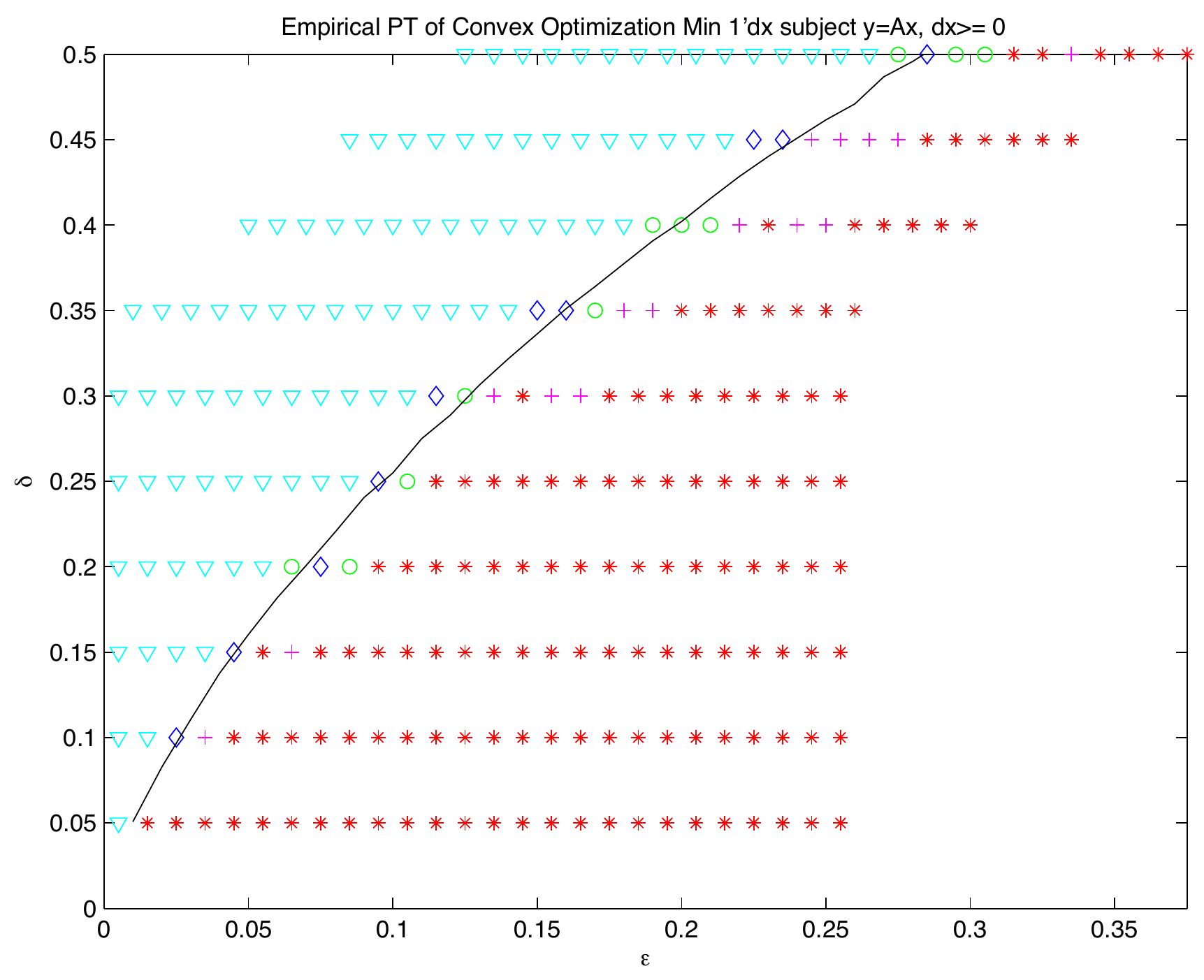}
\caption{Phase transition results for reconstructing monotone signals
  using the convex program $(P_{\rm mono})$.  
Here the signal dimension is $N=200$,
$\delta = n/N$ 
is the undersampling fraction, and $\eps$ is the sparsity parameter
(fraction of change points). 
Red: fraction of
correct recovery smalled than 50\%.
Aqua: fraction of correct recovery larger than 50\%.  Curve: minimax MSE curve $\delta = M(\eps | \Mono)$.
Compare with figure \ref{fig:monoRegAMPPT}.}
\label{fig:PmonoPT}
\end{center}
\end{figure}

Finally, we can compare
TV AMP to the following convex optimization problem:
\begin{eqnarray*}
(P_{TV}) \quad  \mbox{minimize}&& \|x \|_{TV}\, , \\
  \mbox{ subject to}&& y = A {x} ,
\end{eqnarray*}
where again $\|x\|_{TV}\equiv \sum_{i=1}^{N-1}|\Delta x_i|$.
Figure \ref{fig:PtvPT} verifies that the phase transition occurs
at the location predicted by (\ref{generalPT}), just as with TV AMP.

\begin{figure}
\begin{center}
\includegraphics[height=3in]{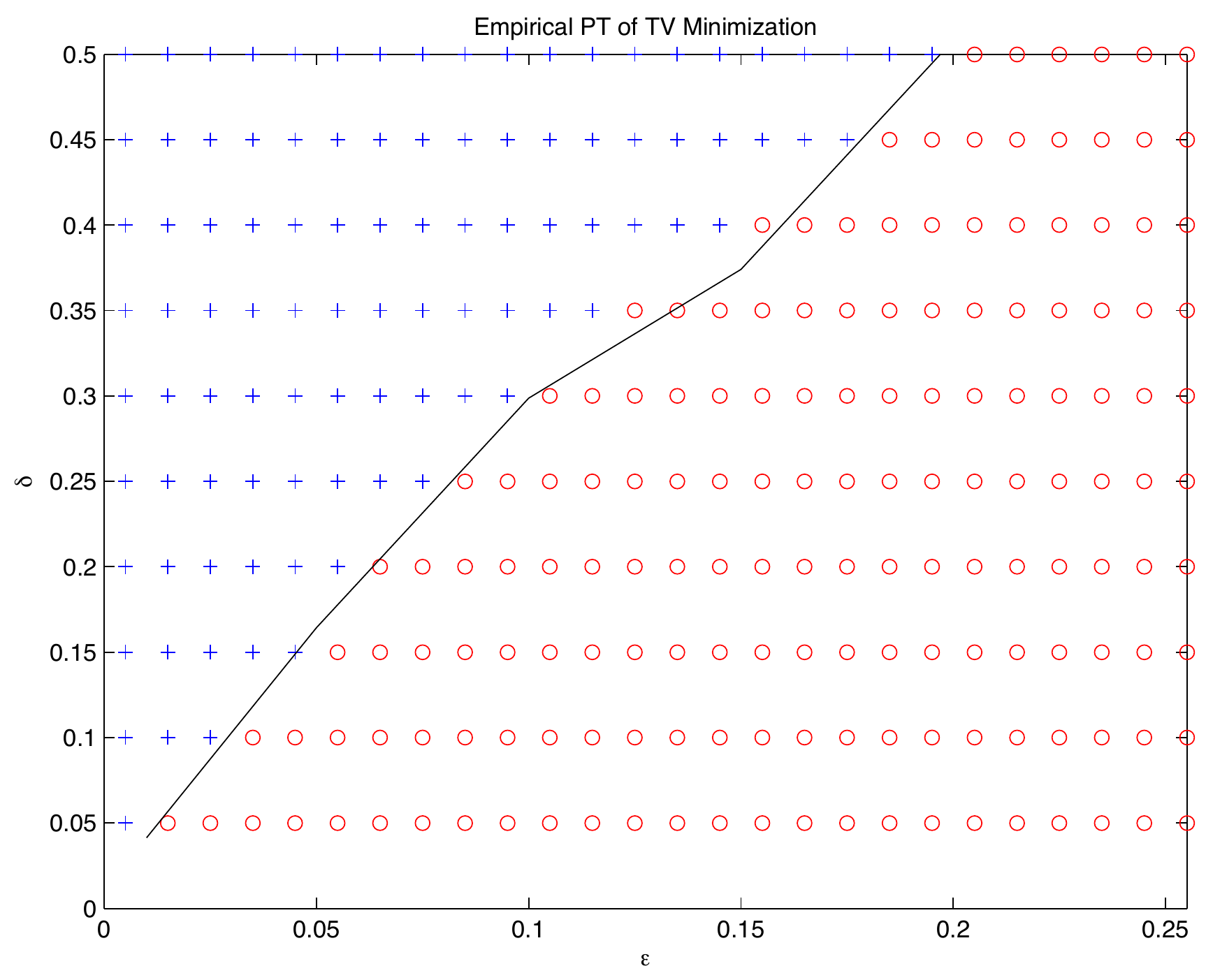}
\caption{Phase transition results for reconstructing piecewise constant signals
  using the convex program $(P_{TV})$. Here, $N=200$,  $\delta = n/N$ is the
  undersampling fraction and
$\eps$ is the fraction of change points.
Here random changepoints were used. 
 Red: fraction of
correct recovery smalled than 50\%.
Aqua: fraction of correct recovery larger than 50\%.  Curve: minimax
MSE curve $\delta = M_{\rm random}(\eps | TV)$.
Compare with  Figure \ref{fig:TVAMPPT}.}
\label{fig:PtvPT}
\end{center}
\end{figure}

\subsection{Interpretation of separable denoisers in terms of  penalized least-squares}

Scalar soft thresholding $\eta^{soft}$ can be interpreted as the
solution of a $\ell_1$-penalized least-squares 
problem
\[
    (P_{J}) \quad  \mbox{minimize}\;\;\;\; \frac{1}{2}(y-x)^2  + 
    J(x) \, ,
\]
where $J(x) =   \lambda |x|$ (here we redefined  $J$ to absorb the factor $\tau$). 

It turns out that a similar interpretation can be given for any scalar 
denoiser (and hence for any separable denoiser as well).
In particular, the minimax and firm thresholding rules
$\eta^{all}$ and $\eta^{firm}(\,\cdot\,;\tau_1,\tau_2)$
are optimizers of  penalization schemes of the same form as above,
with $J(x)$ \emph{non-convex}.

We can construct the penalty $J(\,\cdot\, )$ corresponding to a denoiser $\eta(\,\cdot\,)$
 by observing that $x + J'(x) = y$ at the
solution $x = \eta(y)$. Defining the residual $\Delta(y) = y - \eta(y)$,
and noting that $\eta(y) + \Delta(y) = y$, we obtain that $\Delta(y)$ and
$J'(x)$ are related through the change of variables
\[
     J'(\eta(y)) =  \Delta(y).
\]
In other words, 
\begin{eqnarray*}
J(x) = \int_0^{x} \Delta(\eta^{-1}(u)) \, \de u.
\end{eqnarray*}

Figure \ref{fig:scalarPenalties} shows the implied penalties for minimax firm and
global minimax shrinkage, respectively.  
In both cases, the optimal penalizations are nonconvex, in accord with the commonly-held
belief that nonconvex penalization is ``superior'' to $\ell_1$
penalization \cite{DBLP:conf/icassp/ChartrandY08,DBLP:journals/corr/abs-1004-0402}.
However, the optimal penalization is seemingly very close to the $\ell_1$ penalization,
so that the prejudice in favor of nonconvex penalization
must be re-examined.\footnote{The minimax nonlinearity is somewhat complicated to
implement -- see Appendix \ref{Appendix:minimax}.
Figure \ref{fig:scalarPenalties} suggests to us that among piecewise linear rules,
close to the best phase transition behavior is obtained by a rule which
obeys $\eta(y) < y - \rm{const}$ for large $y$, a possibility that is not explored in the firm family, or in any other conventional proposal.}

\begin{figure}
\begin{center}
\includegraphics[height=4in]{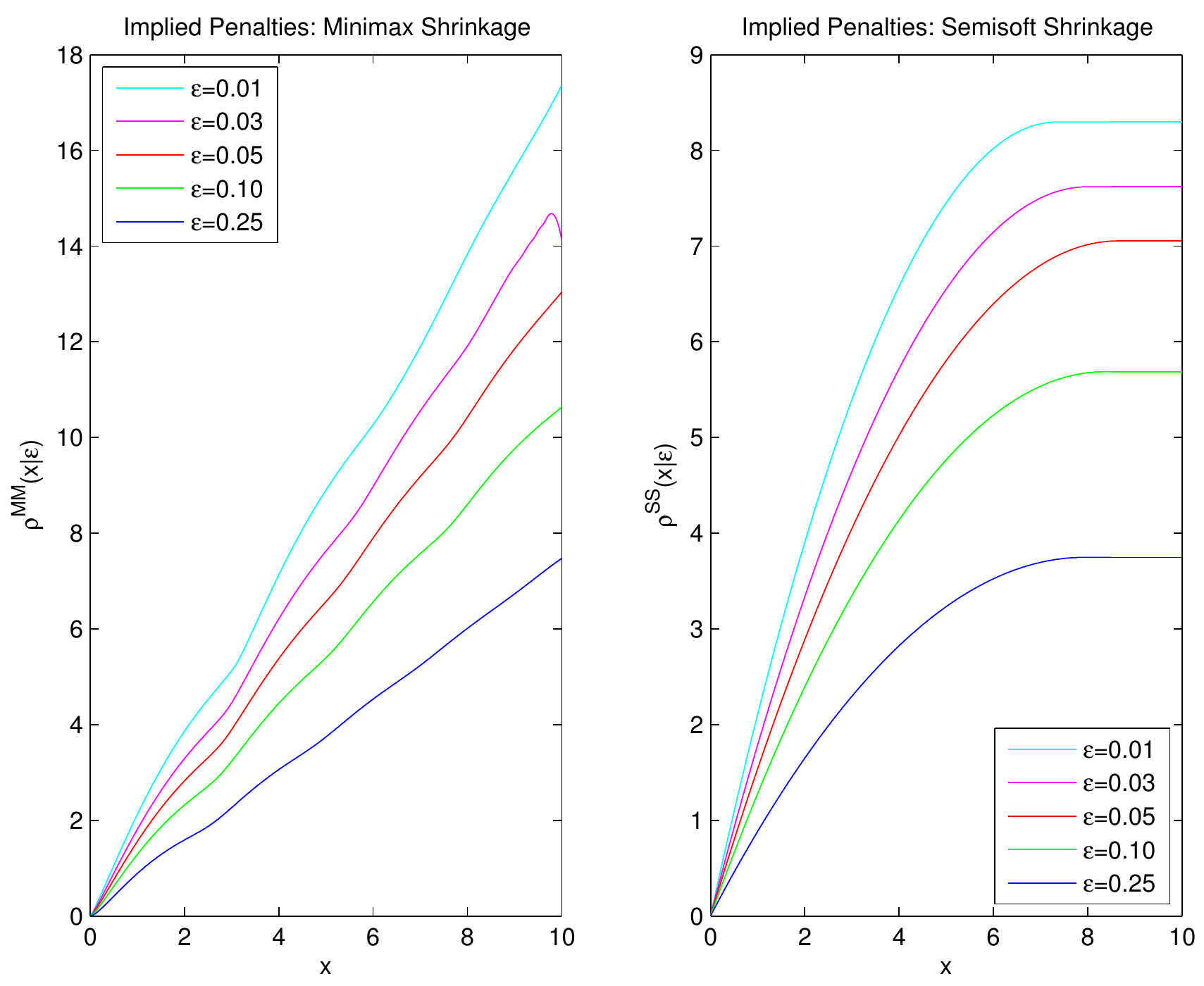}
\caption{Implied penalties  $J(x)$
for two denoisers at $\eps = 0.01$, $0.025$, $0.05$, $0.10$, $0.25$
(from top to bottom). Left:
Globally minimax shrinkage. Right: Minimax firm shrinkage. We plot
only results for positive arguments $x > 0$; results for negative arguments are obtained by symmetry $J(-x) = J(x)$.}
\label{fig:scalarPenalties}
\end{center}
\end{figure}

A similar analysis can be carried out for block separable denoisers
that are covariant under rotation, i.e. if
$\eta(\,\cdot\,;\tau):\bR^B\to\bR^B$ satisfies $\eta(Rx;\tau) =
R\eta(x;\tau)$ for any rotation $R$.  We already mentioned that the
block soft denoiser can be written as
\begin{align*}
\esoft(y;\tau)  =\argmin_{x\in\bR^B}\Big\{\frac{1}{2} \|y-x\|^2_2 +
\tau\|x\|_2\Big\}\, .
\end{align*}
An implied penalty can also be derived for block James-Stein
shrinkage $\eJS$. 
Due to the covariance under rotation, the corresponding penalty only
depends on the modulus of $x\in\reals^B$.
Figure \ref{fig:blockPenalties} shows the implied penalty $J(s|B)$ as a
function of $s=\|x\|_2$.
Namely, the penalty $J(\,\cdot\,|B):\bR_+\to\bR$ is such that, for $y\in\bR^B$,
\[
\eJS(y)  =\argmin   \min_{x\in\bR^B}\Big\{\frac{1}{2} \|y-x\|^2_2 +
J(\|x\|_2|B)\Big\}\, 
\]
coincides with the positive-part James-Stein estimator.
The implicit penalization is again nonconvex.

\begin{figure}
\begin{center}
\includegraphics[height=4in]{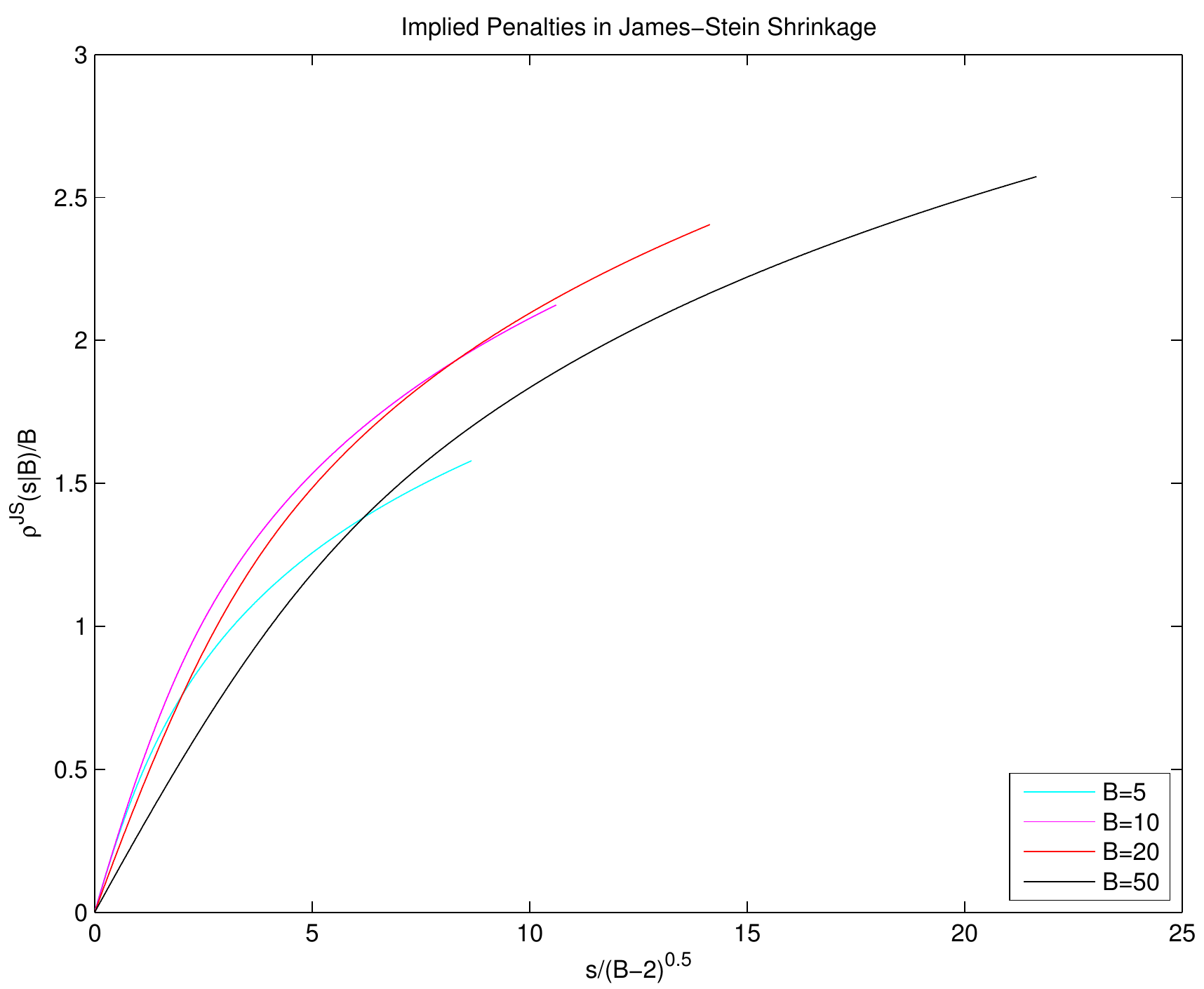}
\caption{Implied penalties $J(s|B)$ for James-Stein positive-part
  estimator for block sizes $B \in \{ 5,10,20,50 \}$. Note that both axes are scaled
with $B$ dependent, so that all plots fit on a common scale. More
precisely, the $y$-axis presents $J(s|B)/B$ and  the $x$-axis presents $s/\sqrt{B-2}$.}
\label{fig:blockPenalties}
\end{center}
\end{figure}

\appendix
\newpage

\section{Classical cases}
\label{Appendix:classical}

The general formula (\ref{generalPT}) was already validated in 
\cite{DMM09} for the following three classical cases:
\begin{enumerate}
\item[$(i)$] Simple sparse signals from the class $\cF_{N,\eps}$
(cf. Eq.~(\ref{eq:SimpleSparse}))
and soft thresholding denoiser.
\item[$(ii)$] Non-negative sparse signals with soft positive thresholding
denoiser, cf. Example \ref{example:Softpos}.
\item[$(iii)$] Box constrained signals with  capping denoiser, cf. Example~\ref{example:Softpos}.
\end{enumerate}
Analytical expressions for the phase transition curves
were computed using state evolution in the Online Supplement to
\cite{DMM09}. We review  the results here
since they provide a useful stepping stone for
understanding more complicate cases (cf., e.g.. Section
\ref{sec:DerivePT}).

Notice that the last of the examples above (box-constrained signals) is not scale invariant,
according to the general definition of Section \ref{sec:SignalModels}.
For non scale-invariant classes, the definitions (\ref{generalM}) and
(\ref{eq:generalMAsymptotic}) are generalized by taking the supremum
over the noise covariance as well. Namely, for a generic class
$\cF_{N,\eps}$, we let 
\begin{eqnarray}
     M(\cF_{N,\eps} | \eta)&=&\inf_{\tau\in\Theta} 
\sup_{\sigma\in \bR_+}\sup_{\nu_N \in \cF_{N,\eps}} 
\frac{1}{N\sigma^2} 
\E_{\nu_N} \Big\{\big\| \bx - \eta(\bx+\sigma\,\bz; \tau\sigma)
\big\|_2^2\Big\} ,\\
M(\eps | \eta) &=&\lim_{N\to\infty}M(\cF_{N,\eps} | \eta)\, .
\end{eqnarray}
It is easy to check that this definition coincides with the earlier one
for scale-invariant classes. We will write $M(\eps|\Cap)$ instead of 
$M(\eps|\ecap)$ for box constrained signals with capping denoiser,
i.e. case $(iii)$ above. 
The results for  case $(iii)$ is further elucidated by comparing it with the following
scale invariant problem:
\begin{enumerate}
\item[$(iv)$]  We consider sparse non-negative signals
  $x_0 \geq 0$, modeled
through the class  $\cF_{N,\eps,+}$, and the simple positive part
denoiser $\eplus(y) \equiv \max(y,0)$.  We denote corresponding minimax
risk by $M(\eps | \Pos)$.
\end{enumerate}

The minimax risk for examples $(i)$, $(ii)$, $(iv)$ can be computed
explicitly. Indeed in both cases we have
\begin{align*}
     M(\eps| \eta) &=  \inf_{\tau\in\bR_+} \sup_{\nu \in
       \cF_{1,\eps}/\nu\in \cF_{1,\eps,+}}   \E\big\{(\eta(X+Z;\tau) - X)^2\big\}  \\
         &=  \inf_{\tau\in\bR_+} \sup_{\nu = (1-\eps) \delta_0 + \eps \delta_\mu}  \E\big\{(\eta(X+Z;\tau) - X)^2 \big\}
\, .
\end{align*}
Here  $X \sim \nu$ and $Z\sim\normal(0,1)$ are independent random
variables. 
The reduction second equality  follows
from th remark that the extremal distributions in $\cF_{1,\eps}$ and $\cF_{1,\eps,+}$ 
are mixtures of  two point distributions. This remark reduces the
calculation of the minimax risk to a simple calculation \cite{DMM09}, whose results
are summarized below.
\begin{proposition}\label{propo:Classical}
The minimax risk for the problems stated above is 
\begin{align}
M(\eps|\Soft) &=
\frac{2\phi(\tau)}{\tau+2\phi(\tau)-2\tau\Phi(-\tau)}\, ,\;\;\;\;\;
\eps =
\frac{2\phi(\tau)-2\tau\Phi(-\tau)}{\tau+2\phi(\tau)-2\tau\Phi(-\tau)}\,
,\\
M(\eps|\SoftPos) & = \frac{\phi(\tau)}{\tau+\phi(\tau)-\tau\Phi(-\tau)}\, ,\;\;\;\;\;
\eps =
\frac{\phi(\tau)-\tau\Phi(-\tau)}{\tau+\phi(\tau)-\tau\Phi(-\tau)}\,
,\\
M(\eps|\Cap) & = \frac{1}{2}(1+\eps)\, ,\\
M(\eps|\Pos) & = \frac{1}{2}(1+\eps)\, .
\end{align}
(The first two are parametric expressions in $\tau$, which is the
optimal threshold at the given sparsity level.)

Also, the AMP phase transition for the noiseless reconstruction
problem is given in all of these cases by $\delta=M(\eps|\eta)$. In
other words AMP succeeds with high probability of
$\delta>M(\eps|\eta)$ and fails with high probability if $\delta<M(\eps|\eta)$.
\end{proposition}
As mentioned above, the calculation of the minimax risk is a calculus
exercise, and follows the same lines as in \cite{DMM09}. This coincide
with the AMP threshold by the general analysis of Section
\ref{sec:DerivePT}
for cases $(i)$, $(ii)$, $(iv)$. For the non-scale invariant case, we
refer, once more, to \cite{DMM09}.

Notice that problem $(iii)$ and $(iv)$ have the same minimax
risk. This identity mirrors a result in  \cite{DoTa10b} that
characterizes  the phase transition threshold for reconstructing
$x_0\in \cS_N\subseteq\bR^N$ from noiseless linear measurements
$y=Ax_0$, with $\cS_N=[0,1]^N$ or $\cS_N=\bR_+^N$. If a simple
feasibility linear program is used (namely, find any $x\in\cS_N$ with
$y=Ax$), then the undersampling threshold for both problems is given
by $\delta =(1+\eps)/2$.

\section{Calculation of minimax MSE}
\label{Appendix:minimax}

In this Appendix we describe the calculation of the global minimax
risk over the class $\cF\equiv \cF_{1,\eps}$, as defined per
Eq.~(\ref{eq:AllMMAX}). In particular, we will explain how the values
in Table \ref{table:MMSEscalar} and Figure \ref{fig:MMSEscalar} for
$M(\eps|\Minimax)$ have
been computed.

Throughout  this section we let $\MSE(\nu,\eta) \equiv \E\{(\eta(X+Z)-X)^2\}$
where $X \sim \nu$ and $Z\sim\normal(0,1)$ are independent random
variables.
From standard minimax theory \cite{JohnstoneBook}, and using an
identity attributed to Brown, we  have, for  $\cF$ convex and
weakly compact
\begin{eqnarray}
 \inf_\eta \sup_{\nu \in \cF} \MSE(\nu,\eta) =   \sup_{\nu \in \cF}
 \inf_{\eta}\MSE(\nu,\eta) = \sup_{\nu \in \cF}\MSE(\nu,\eta^{\nu}) =
 1 -\inf_{\nu \in \cF} I ( \gamma \star \nu) .
\label{eq:BrownIdentity}
\end{eqnarray}
Here $\eta^{\nu}$ is he posterior mean estimator for prior $\nu$, 
$\gamma$ is the Gaussian measure $\gamma(\de x) = \phi(x)\de x$,
$\phi(x) = e^{-x^2/2}/\sqrt{2\pi}$, $\gamma \star \nu$ denotes the
convolution of measures, and $I$ denotes the Fisher information. For a probability measure $\nu_f(\de x) =
f(x)\de x$, with density $f$ with respect to the Lebesgue measure
this is defined as
\[
     I(\nu_f) = \int \frac{ (f'(x))^2}{f(x)}\, \de x\, .
\]

Further, if $\cF$ is convex and weakly compact, the set of probability
measures $\{\gamma\star\nu :\; \nu\in\cF\}$ is also convex and weakly
compact. If follows from \cite[Theorem 4]{HuberMinimax} that the $\inf$ in
Eq.~(\ref{eq:BrownIdentity}) is achieved at a unique point $\nu =
\nu^*$. Hence
\begin{eqnarray}
M(\eps|\Minimax) = 1-\min_{\nu \in \cF} I ( \gamma \star \nu) =
1-I(\gamma\star\nu^*)\equiv 1-I^*\, ,
\end{eqnarray}
where we defined $I^*=I^*(\eps) = I(\gamma\star\nu^*)$.
The unique minimizer $\nu^*$ is known as the\emph{least favorable}
distribution. The minimax optimal denoiser (achieving the
$\inf$ over $\eta$ in Eq.~(\ref{eq:BrownIdentity})) is the posterior
expectation with respect to the prior $\nu^*$.

Bickel and Collins  \cite[Theorem 1]{BiCo83},
prove that, under suitable assumption on the class $\cF$, the least favorable  distribution is a mixture
of point masses
\[
    \nu^* = \sum_{i=-\infty}^\infty  \alpha_i \delta_{ \mu_i},
\]
where  $\sum_i \alpha_i = 1$, $\alpha_i > 0$ and the sequence
$\{\mu_i\}_{i\in\bZ}$ has no accumulation points except, possibly, at $\pm\infty$.
As mentioned above, the minimax denoiser is the posterior expectation
associate to
the prior $\nu^*$. By Tweedie's formula, this takes the form
\[
   \eta^*(y) = y - \psi^*(y),
\]
where  $\psi^*$ is the so-called score function
\[
    \psi^*(y) = -\frac{\de\phantom{y}}{\de y} \log{ f^*(y) },
\]
and $f^*$ is the density of  $\nu^* \star \gamma$.

Focusing now specifically on the class  $\cF = \cF_{1,\eps}$.
This case is covered by the general theory of \cite{BiCo83}, and
corresponds to their example $(ii)$. 
Without loss of generality we can assume that the $\mu_i$ are monotone
increasing, with $\mu_{-i}=\mu_i$,  and that $\mu_0 = 0$, with $\alpha_0 = (1-\eps)$.
A conjecture of Mallows \cite{Mallows78} states that in fact we may take
\[
    \mu_i = c \cdot i \qquad  \alpha_i = \frac{1}{2} (1-\eps) (1-\lambda) \lambda^{i-1},
    \qquad  i\in\bZ\setminus\{ 0\}\, .
\]
In other words,  the conjectured least favorable distribution has the form of a two-sided geometric distribution
on a scaled copy of $\bZ$. 
While the conjecture has not been proved, Mallows \cite{Mallows78}
provided an argument (based on an analogous problem in robust
estimation \cite{HuberMinimax}) that suggesting that indeed it
captures the correct tail behavior.

For estimating the minimax risk numerically, we chose a large parameter $K$ and assume
a generalized ``Mallows form'' for $|i| > K$. More precisely,
  we assume an equispaced grid, and 
geometrically decaying weights. This is a little more general than what Mallows proposed, having 3 total
degrees of freedom (spacing, total weight and rate of decay), rather than two.
For $-K\le i\le K$, we  allow the parameters $\alpha_i$ and $\mu_i$ to
vary freely. 
In this way we obtained a parametric  family $\nu_\theta$ of probability distributions,
with parameter $\theta = ((\alpha_i)_{i\in[K]} , (\mu_i)_{i\in[K]}, c_0,c_1,\lambda)$, with
\[
    \nu_\theta = (1-\eps) \delta_0 + \frac{\eps}{2} \cdot \sum_{k=1}^K \alpha_i
    \cdot ( \delta_{\mu_i} + \delta_{-\mu_i})+ 
\frac{c_0}{2}\sum_{k > K} \lambda^k (\delta_{c_1 k} + \delta_{-c_1 k}).
\]
Define
\begin{equation}
     I^+ \equiv \min \{ I(\gamma \star \nu_{\theta}) : \theta \in \bR^{2K+3} \}.\label{eq:ParametricMinimax}
\end{equation}
The quantity $I^+=I^+(\eps)$ can be estimated numerically, and provides an upper
bound on $I^*$. 
We used  up to $K=50$ and checked that the resulting $I^+$ is insensitive to
this choice. Notice that choice of the Mallows form for $i>K$ is
immaterial for two reasons: $(i)$ As a consequence \cite{BiCo83} and
of the weak continuity of Fisher information, $I^+$ should be
insensitive to the tail behavior of the distribution $F$; $(ii)$ We
are only using it to derive an upper bound.

In order to get a lower bound on $I^*$,
we use Huber's minimax theorem \cite{HuberMinimax,HuberBook}, which
implies that, for any $\psi:\reals\to\reals$ differentiable in measure, 
\begin{equation} 
\label{saddlepoint}
     I(\gamma \star \nu^*) \geq  \sup_g J(\psi,g)  \equiv \sup_g
\frac{(\int
       \psi' g_\psi)^2} {\int \psi^2 g_\psi }\, ,
\end{equation}
where in the supremum $g$ ranges over all densities of probability
measures  $\gamma \star \nu$ with $\nu \in \cF_{1,\eps}$.

Let $\nu^+$ denote the probability measure corresponding to the optimum of the parametric optimization 
(\ref{eq:ParametricMinimax}).
Denote by $ \psi^+$ denote the corresponding score function 
\[
    \psi^+(y) = -\frac{\de\phantom{y}}{\de y} \log{ f^+(y) },
\]
where $f^+ = \nu^+ \star \gamma$. 
This corresponds to a denoiser $\eta^+(y) = y-\psi^+(y)$.
Let $g^+$ denote a maximizing 
density $g$ for $J(\psi^+,g)$.
By Huber's theory, this can be chosen to be two-point mixture
$(1-\eps) \delta_0+ \eps \delta_{\mu^+}$
where $\mu^+$ is chosen to
achieve the worst case value on the right side of (\ref{saddlepoint}) and set $I_- =I_-(\eps)= J(\psi^+,g^+)$.

We have the bounds $I_- \leq I^* \leq I^+$. Numerically, we compute integrals and
extrema over fine grids with at least $100$ samples per unit of range, getting not $I^+$ and $I_-$ but
instead numerical approximations $\widetilde{I}^+$ and $\widetilde{I}_-$. 
Table \ref{table:minInfo} presents some information about numerical approximation 
results, which may help the reader assess its accuracy for small
values of $K$.
Some minimizing distributions obtained in this way are shown in Figure \ref{fig:LFDist};
the mass points $(\mu_i)$ are displayed in Figure
\ref{fig:LFMassPoints}.

\begin{table}
\begin{center}
\begin{tabular}{|l|l|l|l|}
\hline
 $\eps$ & $K$ & $1-I^+(\eps)$ & $\max_{\nu}\MSE^*(\nu,\eta^+)$ \\
 \hline
 0.01 & 2 & 0.0533 &0.0533 \\
 0.05 & 2 & 0.1841 & 0.1841\\
 0.10 & 2 &0.3026 &0.3026 \\
 0.15 & 2 & 0.3983&0.3984 \\
 0.20 & 4 & 0.4802&0.4803\\
 0.25 & 4 & 0.5516&0.5516 \\
 \hline
\end{tabular}
\caption{\emph{Numerical values of  the lower bound on minimax $\MSE$
    ($1-I^+(\eps)$) and of upper bound on minimax} $\MSE$   ($1-I_-(\eps)$).
The upper bound is the numerically computed worst-case
MSE of the corresponding denoiser $\eta^+$. In all of these cases the bounds agree except possibly in the fourth decimal place. }
\label{table:minInfo}
\end{center}
\end{table}

 \begin{figure}
\begin{center}
\includegraphics[height=3in]{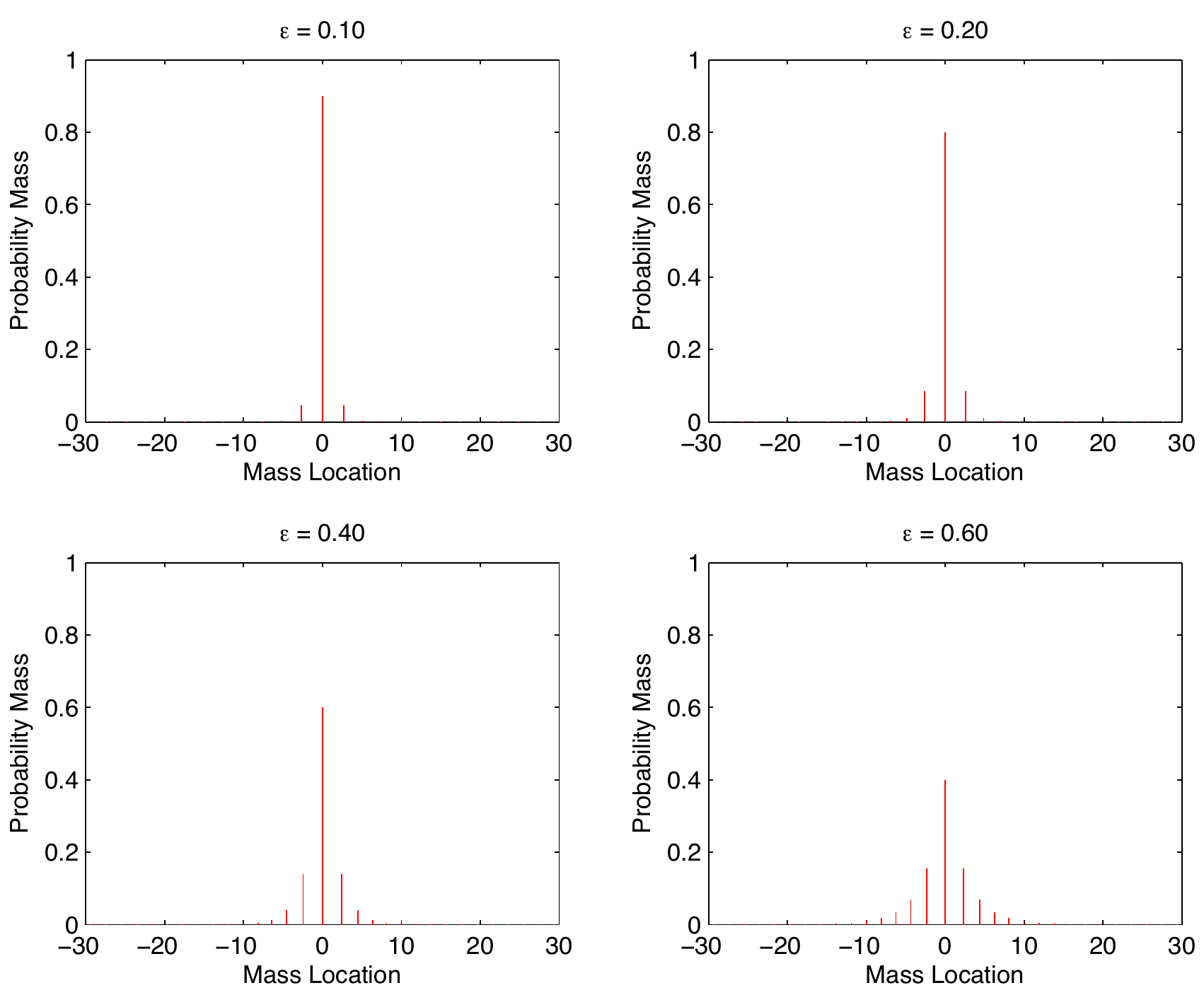}
\caption{Least-favorable distributions over the class $\cF_{1,\eps}$
  obtained by numerically minimizing Fisher
information.}
\label{fig:LFDist}
\end{center}
\end{figure}

 \begin{figure}
\begin{center}
\includegraphics[height=3in]{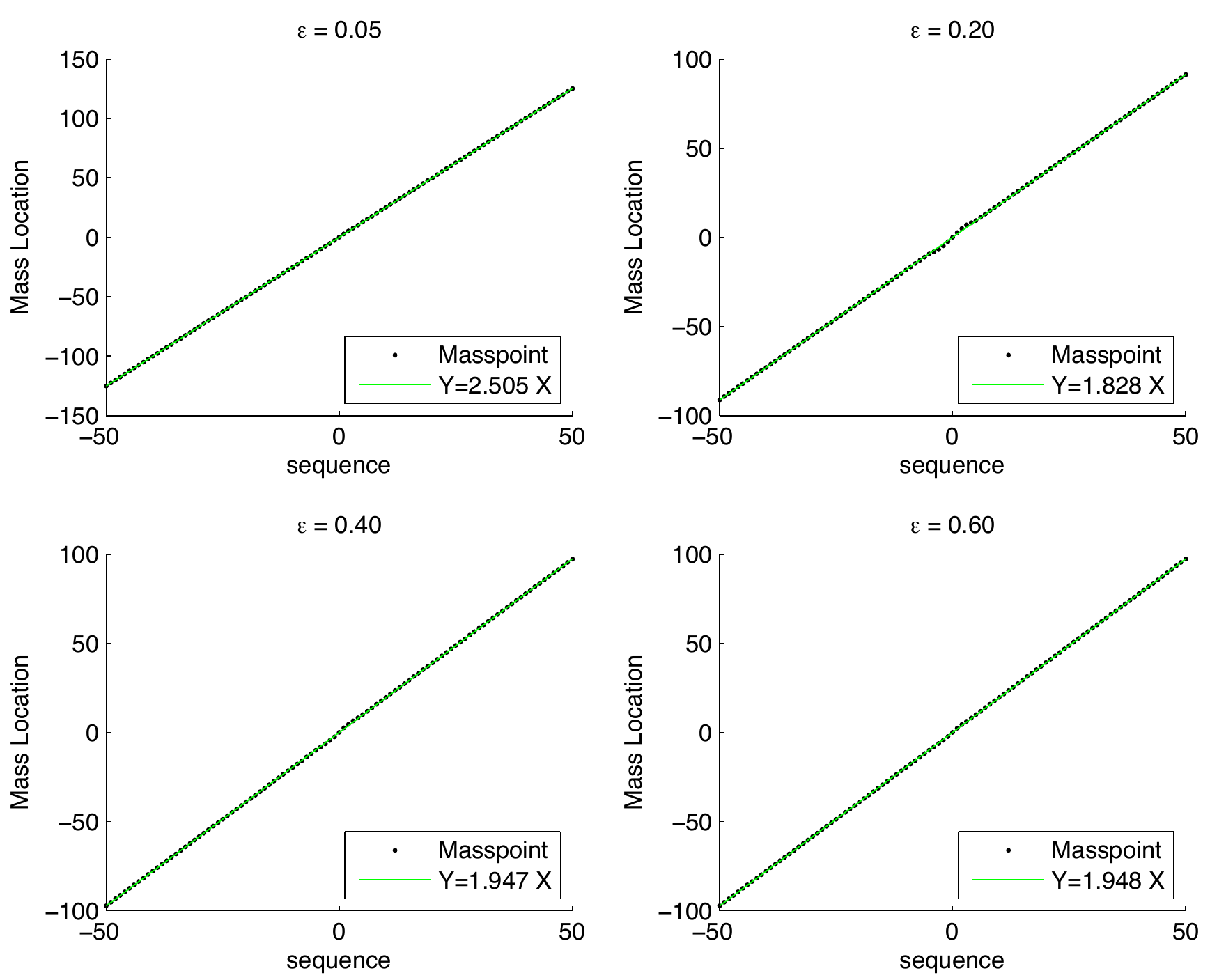}
\caption{Locations of mass points in approximate least-favorable distribution
$(\mu_i)$ versus index $i$, at various $\eps$,
 Note the approximate linearity; least-squares linear fit has slope parameter given in legend.
 The `wiggles' away from the linear behavior occur near the origin.}
\label{fig:LFMassPoints}
\end{center}
\end{figure}

Our numerical
results, showing that  $\widetilde{I}_- \approx \widetilde{I}_+$
allows us to infer that the Mallows form is
approximately correct. The denoiser that we actually apply in our estimation
and compressed sensing experiments is:
\[
    \eta^+(y) = y - \psi^+(y) \, .
\]

\section{Convergence properties of AMP}
\label{Appendix:Convergence}

\begin{figure}
\begin{center}
\includegraphics[height=2.3in]{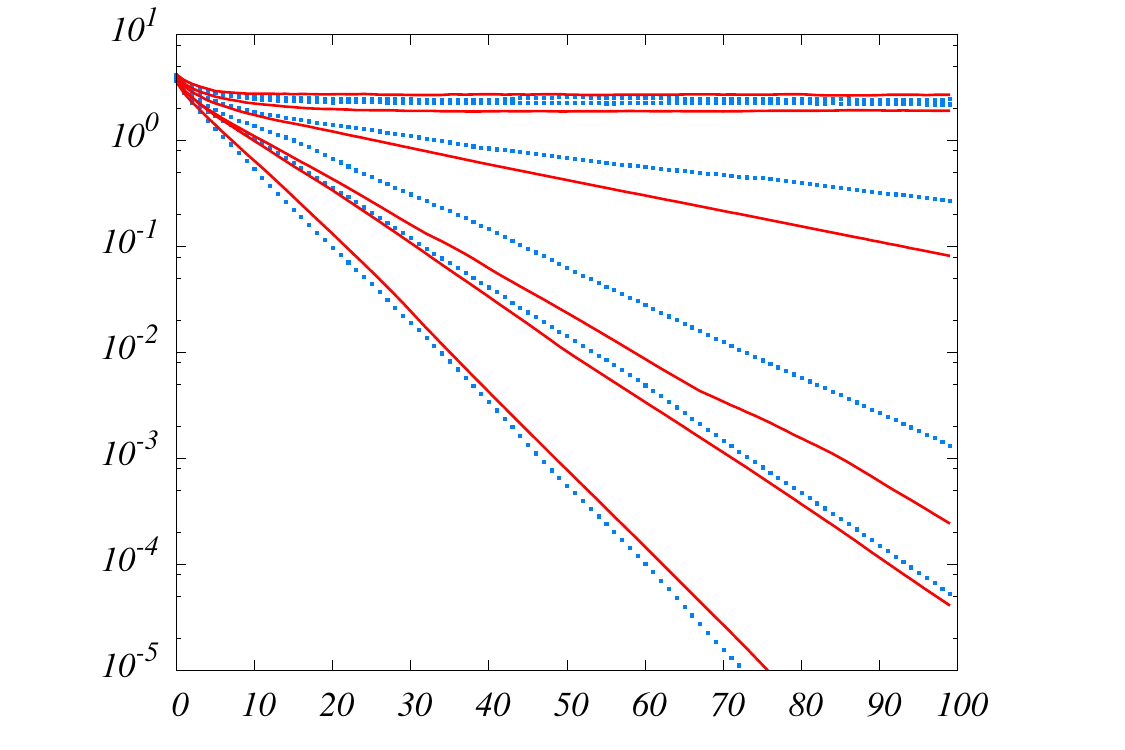}
\put(-140,-5){$t$}
\put(-210,85){{\small $\MSE$}}
\includegraphics[height=2.3in]{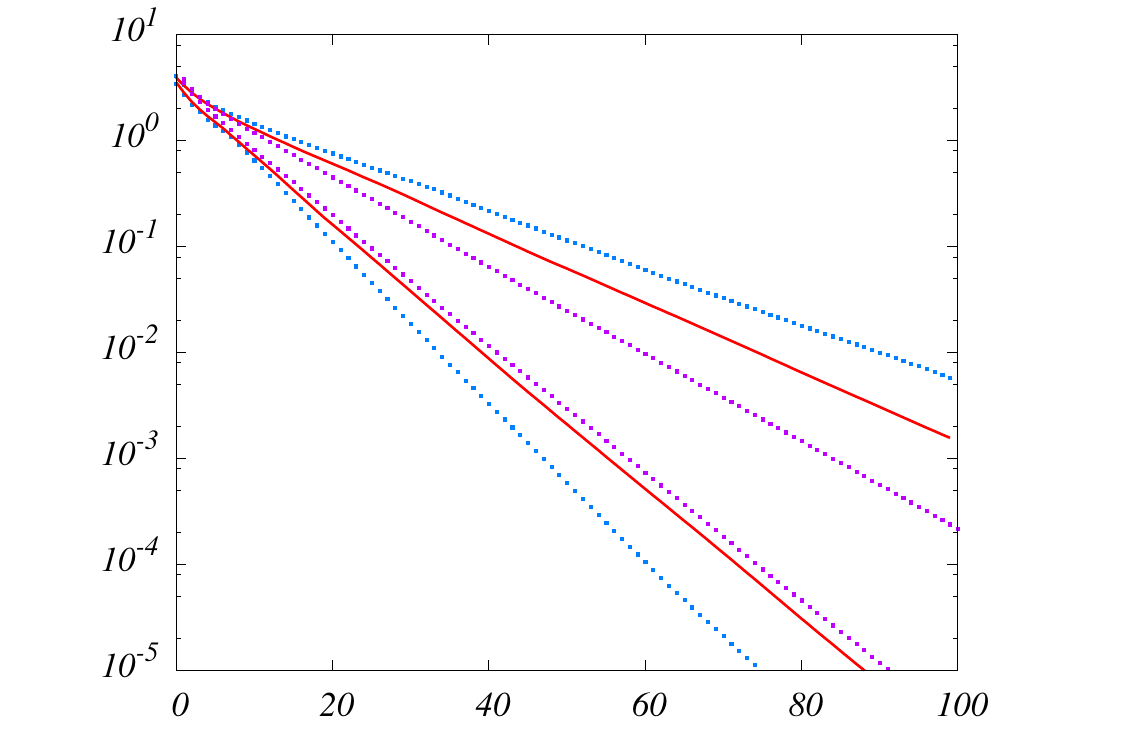}
\put(-140,-5){$t$}
\put(-210,85){{\small $\MSE$}}
\caption{Convergence of AMP. Here we consider soft thresholding AMP
  for simple sparse signals with $\eps=0.05$ (predicted threshold
  $M(\eps|\eta)\approx 0.2039$). Left frame: Median $\MSE$ after $t$
  iterations over $50$ independent realizations of the sensing
  matrix. Blue dotted lines correspond to $n=2000$ and red continuous
  lines to $n=4000$. From top to bottom: $\delta = 0.19$, $0.20$,
  $0.21$, $0.22$, $0.23$, $0.24$. Right frame: $25\%$ (lower tree
  curves) and $75\%$ (upper tree curves)
  percentile curves for the case $\delta=0.23$. The external (dotted),
middle (continuous) and internal (dotted) pair refer, respectively, to
$n=2000$, $4000$, and $8000$.}
\label{fig:Convergence}
\end{center}
\end{figure}
Throughout the paper we checked convergence of AMP by imposing
a threshold on the reconstruction accuracy and the number of
iterations. For instance, in the case of separable denoisers,
cf. Section \ref{sec-PTExperiments-scalar}, we declared the
reconstruction successful if
\begin{align*}
     \frac{ \| \hx^t - x_0\|_2^2}{ \|x_0\|_2^2} < \gamma 
\end{align*}
for a certain choice of $t$ and $\gamma>0$. In particular, the results
presented correspond to $\gamma = 0.01$ and $t=300$. 

It is natural to ask how to choose $\gamma$ and $t$, and whether different choices of $\gamma$ and $t$
would lead to significantly different estimates of the phase
transition boundary. It turns out that the empirical phase transition
is fairly insensitive to these choices for the cases considered here, as soon as $t\gtrsim 100$ is sufficiently
large and $\gamma\lesssim 0.05$. 
This insensitivity is related to the
convergence properties of AMP. Indeed both theory and empirical
evidence \cite{DMM09,DMM-NSPT-11} indicate exponential  convergence.
Namely,  for all $\delta>M(\eps|\eta)$, there exist dimension
independent constants $C=C(\delta)$, $b=b(\delta)>0$ such that, with high probability,
\begin{eqnarray*}
\| \hx^t - x_0\|_2^2\le C \|x_0\|_2^2\, e^{-bt}\, .
\end{eqnarray*}
On the other hand, for $\delta<M(\eps|\eta)$, $\| \hx^t - x_0\|_2^2\ge
c(\delta)n$, with high probability.

Figure \ref{fig:Convergence} presents data that confirm this behavior
(further numerical evidence can be found in \cite{DMM09}). The data
refer to simple sparse signals with $\eps = 0.05$, and soft
thresholding denoising. The curves correspond to several values of $\delta$ close to the
predicted phase transition location $\delta=M(\eps|\eta) \approx
0.0239$. Notice the clear exponential decay of the error for
$\delta>M(\eps|\eta)$ and a large constant mean square error for
$\delta<M(\eps|\eta)$.

If the phase transition has to be determined with relative accuracy
$\Delta$, this suggests the rule of thumb
$\exp\{-b(\delta_*+\Delta)t\}\le \gamma$ and $c(\delta_*-\Delta)\ge
\gamma$. We verified that these conditions are satisfied by our
choices of $t$ and $\gamma$.

\section{Calculation of minimax MSE for block soft thresholding}
\label{Appendix:block}

\subsection{Proof of Lemma \ref{lemma:MSEBlockSoft}}

The argument is analogous to the one for the scalar case (corresponding to $B=1$) treated
in Appendix \ref{Appendix:classical}. For $\mu\in\bR^B$ and
$\tau\in\bR_+$, define the risk at $\mu$ as
\begin{align}
R(\mu;\tau)\equiv \E\big\{\|\esoft(\mu+\bx;\tau)-\mu\|_2^2\big\}\, .\label{eq:BlockRisk}
\end{align}
where $\bx \sim \nu$ and $\bz\sim\normal(0,\id_{B\times B})$.
Since the two point mixtures are the extremal distributions in
$\cF_{\eps,B}$, we have
\begin{align*}
     M_B(\eps| \Blocksoft)&=  \inf_{\tau\in\bR_+} \sup_{\nu \in \cF_{B,\eps}}
     \frac{1}{B} \, \E \big\{\| \eta(\bx+\bz;\tau) - \bx \|^2 \big\} \\
         &=  \inf_{\tau\in\bR_+} \sup_{\begin{array}{c}\nu = (1-\eps) \delta_0 + \eps
           \delta_\mu\\ \mu\in\bR^B\end{array}}   \frac{1}{B} \E\big\{ \| \eta(\bx+\bz;\tau) - \bz \|^2 \big\}\\
         & = \frac{1}{B}\inf_{\tau\in\bR_+} \sup_{\mu\in\reals^B} \Big\{ (1-\eps)
         R(0;\tau)+ \eps\, R(\mu;\tau)\Big\}\, .
\end{align*}
By the definition of chi-square distribution, we have
\begin{align*}
  R(0;\tau) = \E \big\{\| \eta(Z;\tau) \| ^2\big\}
  =    \E \big\{( \sqrt{X_{B}} - \tau)_+^2 \big\}\, .
\end{align*}
It follows from the invariance of the distribution of $\bz$ under
rotations that $R(\mu;\tau)$ only depends on $\mu$ through its norm
$\|\mu\|$. Further, as proved in Appendix \ref{lemma:PropertiesOfBlock} 
$R(\mu;\tau)$ is increasing in $\|\mu\|$, and
\begin{align*}
R(\infty;\tau) \equiv \lim_{\|\mu\|\to\infty} R(\mu;\tau) = B+\tau^2\,
.
\end{align*}
We therefore obtain
\begin{align}
M_B(\eps| \Blocksoft) = \frac{1}{B}\inf_{\tau\in\bR_+} \Big\{ (1-\eps) \E \big\{( \sqrt{X_{B}} - \tau)_+^2 \big\}
         + \eps\, (B+\tau^2)\Big\}\, .\label{eq:BlockProofLast}
\end{align}
At this point the problem is reduced to a calculus exercise.

\subsection{Proof of Lemma  \ref{lem:blockMSElimit}}
\label{app:LargeBlocks}

In this appendix we consider the asymptotics for large block size $B\to\infty$.
It is easy to show that  the minimax threshold level $\tau$ must be of
order $\sqrt{B}$. 
By a compactness argument, we can assume $\tau=c \,\sqrt{B}$ for some
$c$ to be determined.
Define the risk as in Eq.~(\ref{eq:BlockRisk}) (note that this depends
implicitely on $B$) and the normalized risk as $\tR_B(\mu;\tau) =
R(\mu;\tau)/B$.
We claim that 
\begin{align}
\lim_{B\to\infty} \tR_B(0;c\sqrt{B}) &= (1-c)_+^2\, ,\label{eq:ClaimBlock1}\\
\lim_{B\to\infty} \tR_B(\infty;c\sqrt{B}) &= 1+c^2\, . \label{eq:ClaimBlock2}
\end{align}
Assuming these claim to hold, we have, by Eq.~(\ref{eq:BlockProofLast}),
\begin{align*}
\lim_{B \goto \infty} M_B(\eps;\Blocksoft)
    =  (1-\eps) \, (1-c)_+^2 + \eps  \, ( 1 + c^2)  . 
\end{align*}
Calculus shows that the minimum is achieved at  $c^* = 1-\eps$, whence
\[
   \lim_{B \goto \infty} M_B(\eps|\Blocksoft) =  (1-\eps) \eps^2   +
   \eps \, (1 + (1-\eps)^2) \, ,
 \]
which coincides with the statement of Lemma \ref{lem:blockMSElimit}.

In order to complete the proof, we have to prove claims
(\ref{eq:ClaimBlock1}) and (\ref{eq:ClaimBlock2}).
The second one is immediate because of Lemma
\ref{lemma:PropertiesOfBlock}  that implies indeed
$\tR_B(\infty;c\sqrt{B}) = 1+c^2$.

The limit (\ref{eq:ClaimBlock1})
follows instead from the central limit theorem.
Indeed, let $X_{B}$ denote a central chi-squared with $B$ degrees of freedom.
Its square root is the norm of a standard Gaussian random vector in dimension $B$,
and concentrates around $\sqrt{B}$. 
Indeed by the central limit theorem $\sqrt{2} (\sqrt{X_{B}}
-\sqrt{B})   \Rightarrow_D  \normal(0,1)$ as 
$B \goto \infty$.
Therefore, we have
\begin{align*}
\lim_{B\to\infty} \tR_B(0;c\sqrt{B})  &=\lim_{B\to\infty} \frac{1}{B}\E \big\{( \sqrt{X_{B}} - c\sqrt{B})_+^2 \big\}\\
                         & = \lim_{B\to\infty} \frac{1}{B}\E\big\{(
                         \sqrt{B} +\frac{1}{\sqrt{2B}}  Z -
                          c\sqrt{B})_+^2\big\}\\
  & = \lim_{B\to\infty} \E\Big\{\Big(
                         1-c +\frac{1}{\sqrt{2B}}  Z\Big)_+^2\Big\}
\end{align*}
and the latter converges to $(1-c)_+^2$ by dominated convergence.

%
%
\section{Non-convergence of state evolution}
\label{app:Sub}
 
In this appendix we show non-convergence of state evolution  for firm-thresholding AMP and globally
minimax AMP below $\delta_{SE}(\eps)$. More precisely, we show that,
for $\delta<\delta_{SE}(\eps)$ there exists a probability distribution
$\nu$ such that the state evolution sequence $\{m_t\}_{t\ge 0}$, $m_0=\E_{\nu}\{X^2\}$ does not converge
to $0$.
 
We begin by developing  a lower bound that holds for all the denoisers
$\eta$ studied in this paper. For notational simplicity, we consider
in fact separable denoisers, but the result is easily see to hold in
general, provided that the signal class is scale-invariant.
 
Let  $\nu_{\eps,\mu}$ denote the mixture $(1-\eps)\delta_0 + \eps \delta_\mu$
where $\mu\in \bR_+\cup\{\infty\}$ can be either finite or infinite.
Define the risk at $\mu$ as
\begin{align*}
R(\mu;\tau) = \E_{\nu_{\eps,\mu}}\Big\{\big[\eta(X+Z;\tau)-X\big]^2\Big\}\, .
\end{align*}
where $X\sim\nu_{\eps,\mu}$ is independent of $Z\sim\normal(0,1)$.
 \begin{definition}
We say that the risk function $R$ is \emph{super-quadratic}  on $[0,\mu_*)$ if, for any
$\mu\in [0,\mu_*)$,
\begin{equation} \label{eq:superquadratic}
     R(\mu_*)  \geq  \left( \frac{\mu}{\mu_*} \right)^2 \cdot R(\mu_*;\tau).
\end{equation}
\end{definition}
 
The next result shows that superquadratic behavior of the risk
function implies non convergence of state evolution
for signal distribution $\nu_{\ve,\mu}$.
\begin{lemma} {\bf (State Evolution Non-Convergence)}
 \label{lem:SEnonConverge}
Fix any $\tau\in \Theta$, and assume there exists $\mu_0>0$ such that:
$(i)$ The risk function $R(\mu; \tau)$ is superquadratic
on $[0,\mu_*)$,  $(ii)$ $R(\mu_*) \ge \delta$, and $(iii)$ $\delta \geq
\eps$. 

Let $\Psi(m) = \Psi(m;\delta,\tau,\nu_{\eps,\mu})$. Then 
there is  $m_{\rm fp} > 0$ such that
\[
       m\ge \mfp\Rightarrow \Psi(m) \geq \mfp \, .
\]
In particular, if $\mfp < m_0 = \E_{\nu_{\eps,\mu}} \|X\|^2=\eps\mu^2$ then, for all $t\ge 0$,
\begin{equation} \label{nonconvergence}
      m_t \geq \mfp > 0 \, . 
\end{equation}
\end{lemma}
\begin{proof}
By the scaling relation  (\ref{eq:ScalingRelation}).
\begin{eqnarray*}
\Psi(m) \equiv \frac{1}{\delta}\E_{\nu_{\eps,\mu}} \Big\{
\big\| X - \eta(X+\sqrt{m/\delta} Z;  \tau, \sqrt{m/\delta})
\big\|_2^2\Big\} 
=\frac{m}{\delta}\, R\big(\frac{\mu}{\sqrt{m}};\tau\big).
\end{eqnarray*}
Define $\mfp = (\mu/\mu_*)^2$ and assume $m\ge \mfp$.
This implies $\mu/\sqrt{m}\le \mu_*$, and, since $R$ is superquadratic by
assumption,
\begin{eqnarray*}
\Psi(m)
\ge\frac{m}{\delta}\Big(\frac{\mu}{\sqrt{m}}\,\frac{1}{\mu_*}\Big)^2\,R(\mu_*;\tau)
= \Big(\frac{\mu}{\mu_*}\Big)^2=\mfp\, ,
\end{eqnarray*}
which concludes the proof.
\end{proof}

For firm and minimax denoisers $\eta\in\{\efirm, \eall\}$, we took a fine grid of $\eps$ and 
at each fixed $\eps$ evaluated $R(\mu;\tau)$ on a fine grid of $\mu$, checking the inequality
(\ref{eq:superquadratic}). In the case of firm thresholding we used
the minimax threshold values $\tau = \tau^*(\eps)$.
We further used the least favorable $\mu$, $\mu = \mu^*(\eps)$.
 Sample results are
presented in Figures \ref{fig:VerifySuperQuadFirm}
and \ref{fig:VerifySuperQuadMinimax}.

These computations show that the risk function $R(\mu;\tau^*(\eps))$ is
superquadratic on $(0,\mu^*(\eps))$.

\begin{figure}
\begin{center}
 \includegraphics[height=3in]{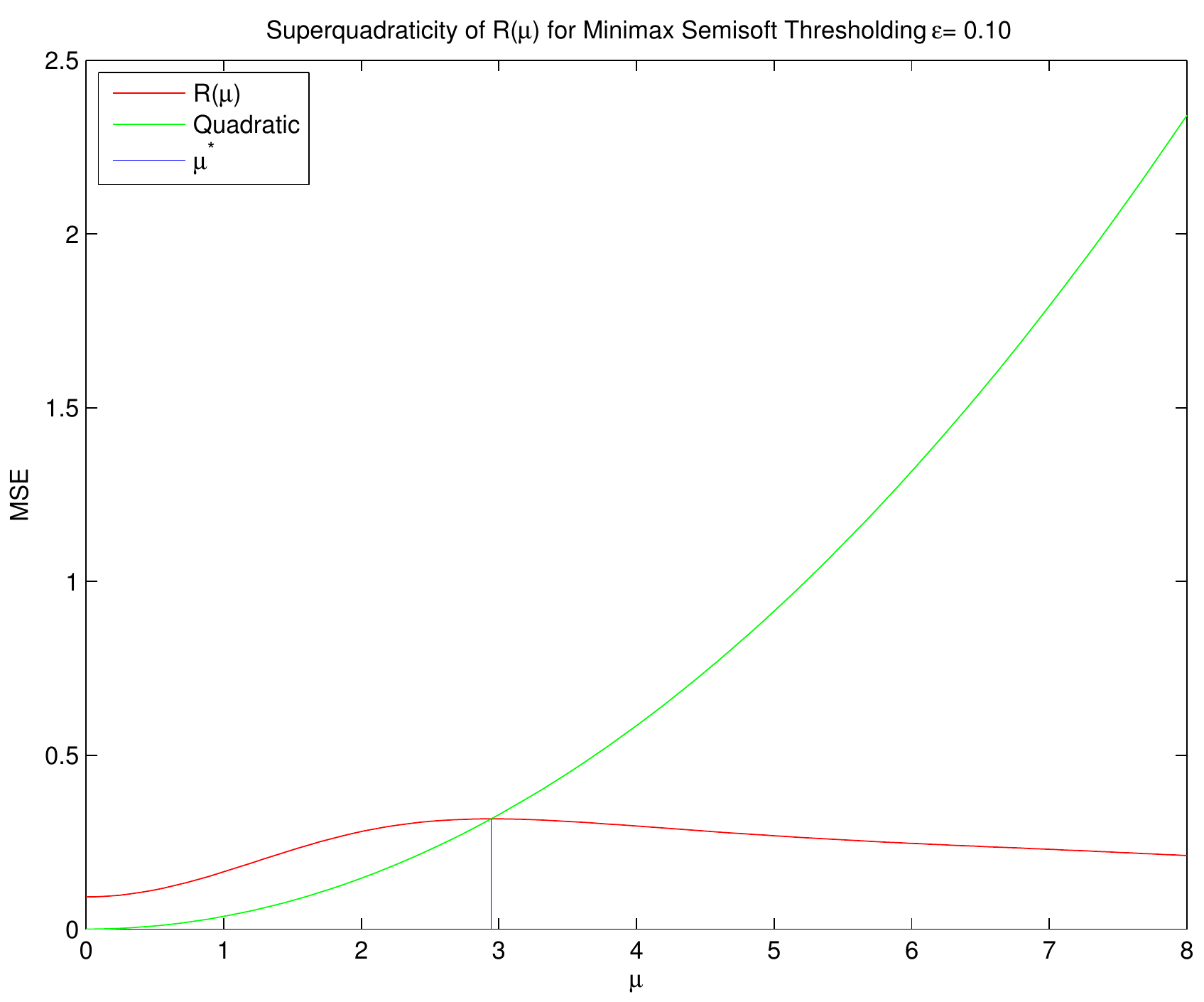}
\caption{Numerical verification of superquadraticity for firm shrinkage.
The green parabola depicts the quadratic function $R(\mu^*;\tau^*)  \cdot (\mu/\mu^*)^2$.
The vertical line depicts the position of $\mu^*$. The red curve depicts $R(\mu;\tau^*)$.
The fact that the latter is above the parabola throughout the interval
$[0,\mu^*)$ verifies the superquadratic condition (\ref{eq:superquadratic}).}
\label{fig:VerifySuperQuadFirm}
\end{center}
\end{figure}

\begin{figure}
\begin{center}
 \includegraphics[height=3in]{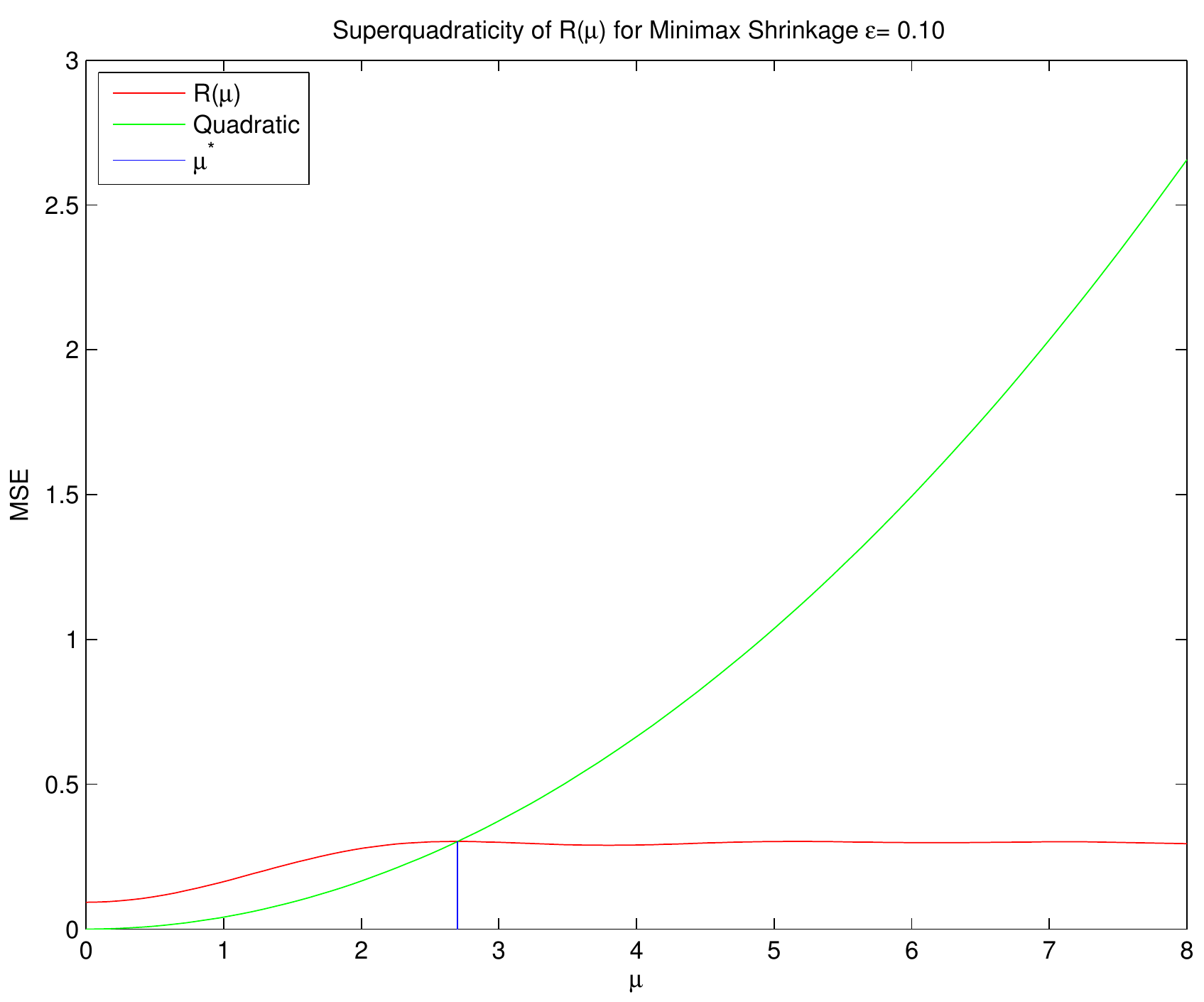}
\caption{Numerical verification of superquadraticity for minimax
  shrinkage.
The green parabola depicts the quadratic function $R(\mu^*;\tau^*)  \cdot (\mu/\mu^*)^2$.
The vertical line depicts the position of $\mu^*$. The red curve depicts $R(\mu;\tau^*)$.
The fact that the latter is above the parabola throughout the interval
$[0,\mu^*)$ verifies the superquadratic condition
(\ref{eq:superquadratic}).
}
\label{fig:VerifySuperQuadMinimax}
\end{center}
\end{figure}

%
%
\section{Monotone regression}
\label{Appendix:monotone}

\subsection{Proof of Lemma \ref{lemma:MonotoneRiskAtInfty}, part $(b)$}

The risk at $\mu$ can be equivalently be written as $R_N(t\mu) =
\E\{\|v(t\mu,\bz)\|_2^2\}$, $\bz\sim\normal(0,\id_{N\times N})$,
where $v = v(t\mu,\bz)$ solves the optimization problem
\begin{eqnarray*}
\mbox{ minimize } &&\| v-z \|_2^2\, ,\\
  \mbox{ subject to }  && \Delta v_i \ge -t\Delta\mu_i, \mbox{ for all } i \in\{1,\dots,N-1\}\, ,
\end{eqnarray*}
where we recall that $\Delta v_i = v_{i+1}-v_i$ is the discrete derivative.
(Of course this problem does not provide an algorithm since it
requires to know $\Delta\mu_i$, but here we are interested in it only for
analysis purposes.) 

As $t\to \infty$, all the constraints $\Delta v_i \geq
-t\Delta\mu_i$ for which $\Delta\mu_i>0$ become irrelevant.
We are naturally led to defining the following \emph{localized} monotone regression problem
\begin{eqnarray*}
                (Q_{lmono})   \quad \mbox{ minimize } &&\|  v - z
                \|_2^2\, ,\\
  \mbox{ subject to }  && \Delta v_i \geq 0, \mbox{ for all } i \in
  I_0\equiv [N-1]\setminus I_+\, .
\end{eqnarray*}
(Here and below omit the dependence of $I_0$, $I_+$ on $\mu$.)
 Let $\eta^{lmono}(z ; I_0)$ denote the solution of $(Q_{lmono})$ with
data $z, I_0$. The above discussion implies that, for
$z=\bZ\sim\normal(0,\id_{N\times N})$,
 we have the following limit in probability:
 \[
               \lim_{t\to\infty } \big(\eta^{mono}(t\mu +  \bz) - \mu\big) = \eta^{lmono}(\bz ; I_0).
 \]
As a consequence (and using the fact that the higher moments of 
 $(\eta^{mono}(t\mu +  \bz) - \mu)$ are bounded uniformly in $t$)
defining $R^{loc} (I_0) = \E\{ \|  \eta^{lmono}(\bz ; I_0) \|_2^2\}$, we
 have
 \[
                \lim_{t \goto \infty} R_N( t\mu )  = R^{loc} ( I_0 ) .
 \]
 In  words, the risk `at infinity' of monotone regression is simply
 given by the local risk.

In order to conclude the proof, it is sufficient to show that $R^{loc}
( I_0 )$ is given by the right-hand side of Eq.~(\ref{eq:MonotoneRiskAtInfty}).
It is  easy to check that  the problem $(Q_{lmono})$ separates into
 independent optimization problem for each $J_k$. Namely, for $i\in
 J_k$, $v_i$ can be found by solving the following smaller problem
\begin{eqnarray*}
      (Q_{lmono,J_k})  \quad \mbox{ minimize}&& \|v_{J_k} -
      z_{J_k}\|_2^2, ,\\
      \mbox{ subject to }&& \Delta v_i   \geq 0, \qquad \mbox{ if both
      } i, i+1\in J_k ,
\end{eqnarray*}
 where $v_{J_k} = (v_{i_k+1},\dots, v_{i_{k+1}})$ and, if the segment $J_k$ is a singleton, the constraint
 disappears. Let $v_{J_k}(z_{J_k})$ be the solution of this
 problem. Then, 
\begin{align*}
R^{loc}( I_0 )=\sum_{k=0}^{K}\E\{\|v_{J_k}(\bz_{J_k})\|^2\}\, .
\end{align*}
On the other hand $(Q_{lmono,J_k})$ is simply the monotone regression
problem on the segment $J_k$ with data $y_{J_k} = 0+z_{J_k}$, and
hence $\E\{\|v_{J_k}(\bz_{J_k})\|^2\} = r(|J_k|)$, which implies
immediately the desired claim.

\subsection{Proof of Lemma \ref{lemma:MonoRiskAtZero}}
\label{sec:ProofMonoRiskAtZero}

Throughout the proof, we denote by $x = x(\bz) = \emono(z=\bz)$ the solution of the monotone
regression problem
\begin{eqnarray*}
\mbox{ minimize } &&\| x-z \|_2^2\, ,\\
  \mbox{ subject to }  && x_{i} \le x_{i+1}\, , \;\;\;\mbox{ for all } i \in\{1,\dots,N-1\}\, ,
\end{eqnarray*}
with data $z=\bz\sim\normal(0,\id_{N\times N})$. We then have $r(N) =
\E\{\|x(\bz)\|_2^2\}$.

Clearly $r(1)=1$ since in this case there is no monotonicity
constraint and the solution of the regression problem is simply
$x=z$. In order to prove Eq.~(\ref{eq:MonotoneLog}), let $k\in I_+(x)$ 
(i.e. an increase point: $x_k<x_{k+1}$) and define $r^{(k)}\in \bR^N$
by letting $r^{(k)}_i = 1_{\{i>k\}}$, and $l^{(k)}\in \bR^N$ by letting $l^{(k)}_i = 1_{\{i\le k\}}$. Then, $x+\xi\,r^{(k)}
\in\cMono_N$, $x+\xi\, l^{(k)}\in\cMono_N$ for all $\xi$ small enough. Hence, we must have $\|x+\xi\, r^{(k)}-z\|_2^2\ge
\|x-z\|_2^2$, and $\|x+\xi\, l^{(k)}-z\|_2^2\ge \|x-z\|_2^2$ for all
$\xi$ small enough. Expanding to linear order in $\xi$, we conclude
that, for all $k\in I_+(x)$:
\begin{align}
\sum_{i=1}^kx_i = \sum_{i=1}^kz_i\, ,\;\;\;\;\;\;
\sum_{i=k+1}^Nx_i = \sum_{i=k+1}^Nz_i\, . \label{eq:MinCondition}
\end{align}
Further, if $r^{(0)}$ is the all $1$ vector, $x+\xi\, r^{(0)}\in
\cMono_N$ for all $\xi\in\bR$. Minimizing with respect to $\xi$, we get
\begin{align}
\sum_{i=1}^Nx_i = \sum_{i=1}^Nz_i\, .\label{eq:MinCondition2}
\end{align}

Define the events (for $k\in [N-1]$) 
\begin{align*}
\cE_k \equiv \Big\{ k\in I_+(x)\, ;\; x_k\le -\sqrt{\frac{6\log N}{k}}
\; \mbox{ or }\; x_{k+1} \ge \sqrt{\frac{6\log N}{N-k}}\, \Big\}\, .
\end{align*}
By virtue of Eq.~(\ref{eq:MinCondition}), and using the fact that $x$
is monotone, we have
\begin{align}
\prob\{\cE_k\}&\le \prob\Big\{ \sum_{i=1}^kZ_i \le -\sqrt{6k\log
  N}\Big\} + \prob\Big\{ \sum_{i=k+1}^NZ_i \ge \sqrt{6(N-k)\log
  N}\Big\}\nonumber\\
&\le \frac{1}{2}\exp\Big\{-\frac{6k\log N}{2k}\Big\} + \frac{1}{2}\,
\exp\Big\{-\frac{6(N-k)\log N}{2(N-k)}\Big\} = \frac{1}{N^3}\, ,
\end{align}
where we used the fact that, for $Z\sim\normal(0,1)$ and $a\ge 0$,
$\prob\{Z\ge a\}\le (1/2)\, exp(-a^2/2)$.

Next define
\begin{align*}
\cE_{\emptyset} \equiv \Big\{  x_1\ge \sqrt{\frac{6\log N}{N}} \mbox{
  or } x_N\le -\sqrt{\frac{6\log N}{N}} \, \Big\}\, .
\end{align*}
By Eq.~(\ref{eq:MinCondition2}) and monotonicity we have
\begin{align}
\prob\{\cE_{\emptyset}\} \le \prob\Big\{\Big|\sum_{i=1}^NZ_i\Big|\ge
\sqrt{6N\log N}\Big\}\le \frac{1}{N^3}\, .
\end{align}

Finally, consider the event 
\begin{align*}
\cE \equiv \Big\{ -\sqrt{\frac{6\log N}{i}}\le x_i\le \sqrt{\frac{6\log
    N}{N-i+1}} \;\;\mbox{ for all } i\in [N]\Big\}\, .
\end{align*}
It is then easy to check that, letting $\cE^c$ denote the complement
of $\cE$,
\begin{align*}
\cE^c \subseteq \cE_{\emptyset}\cup\cE_1\cup\cdots\cup\cE_{N-1}\, .
\end{align*}
Indeed define, for
$i\in [N]$, $k(i)\equiv \max\{k :\; k<i, \; k\in I_+(x)\}$ with, by
convention, $k(i) = 0$ if the set $\{k<i, \; k\in I_+(x)\}$ is empty.
Then, by definition of $I_+(x)$, 
\begin{align*}
x_i = x_{k(i)+1}\le \sqrt{\frac{6\log N}{N-k(i)}}\le \sqrt{\frac{6\log
    N}{N-i+1}}\, ,
\end{align*}
where the first inequality follows by definition of the events
$\cE_{\emptyset}$, $\cE_1$, \dots $\cE_{N-1}$ and the second since
$i\ge k(i)+1$. The inequality $x_i\ge -\sqrt{(6\log N)/i}$ follows
essentially by the same argument. 
By union bound we therefore obtain 
\begin{align}
\prob\{\cE\}\ge 1 -\frac{1}{N^2}\, .\label{eq:EventMono}
\end{align}

We can therefore bound the risk as follows 
\begin{align*}
r(N) = \E\{\|x(\bz)\|_2^2\} = \E\{\|x(\bz)\|_2^2;\cE\}
+\E\{\|x(\bz)\|_2^2;\cE^c\} \, .
\end{align*}
By definition of $\cE$, we have 
\begin{align}
\E\{\|x(\bz)\|_2^2;\cE\} &\le \sum_{i=1}^{N}\max\Big\{\frac{6\log
  N}{i};\frac{6\log N}{N-i+1}\Big\}\nonumber\\
&\le 12\log N \sum_{i=1}^{\lceil N/2\rceil}\frac{1}{i} \le 12\log
N\big(\log  N+1\big) \, . \label{eq:MonoProofAlmost}
\end{align} 
On the other hand it is easy to see that, defining $Z_{\max} =
\max_{i\in [N]} |Z_i|$, we necessarily have $|x_i(\bZ)|\le Z_{\max}$
for all $i\in [N]$, whence
\begin{align}
\E\{\|x(\bz)\|_2^2;\cE^c\} \le N\E\{Z_{\max}^2;\cE^c\}\le
N\,\E\{Z_{\max}^4\}^{1/2}\prob\{\cE^c\}^{1/2}
\le \E\{Z_{\max}^4\}^{1/2}\, .\label{eq:BoundZmax}
\end{align}
Here the second inequality follows from Cauchy-Schwarz and the last
one from Eq.~(\ref{eq:EventMono}). By union bound, $\prob\{Z_{\max}\ge
z\}\le N\,\prob\{|Z_1|\ge z\}\le N\,\exp\{-z^2/2\}$. Hence
\begin{align*}
\E\{Z_{\max}^4\} & = 4\,\int_0^{\infty}z^3\, \prob\{Z_{\max}\ge
z\}\,\de z\\
&\le   4\,\int_0^{\sqrt{2\log N}}z^3\,
\de z + 4N\, \int^{\infty}_{\sqrt{2\log n}}z^3\, e^{-z^2/2}\,\de z\\
&\le (2\log N)^2 + 4\int_0^{\infty}(\sqrt{2\log N}+u)^3e^{-\sqrt{2\log
    N}\; u}\de u\\
&\le  (2\log N)^2 + 4(2\log N) + 20 \le (4\log N)^2\, ,
\end{align*}
where the last inequality holds for $N\ge 10$. 
Using this bound in conjunction with Eq.~(\ref{eq:MonoProofAlmost}),
we finally get
\begin{align*}
r(N)\le 12(\log N)^2+ 16\log N \le 20 (\log N)^2\, .
\end{align*}

Finally notice that for arbitrary $N$ the simple bound $r(N)\le
N\E\{Z_{\max}^2\}\le 2N[\log N +1]$ can be proved by controlling
$Z_{\max}$ along the same lines as above.

\section{Further computational details}
\label{Appendix:Computation}

We record here a few details that have been omitted from the main text.
\begin{description}
\item[Estimate of the effective noise level.] 
The AMP iteration, cf. Eqs.~(\ref{AMPB}), (\ref{AMPB}),  (\ref{AMPC}),
requires to estimate the variance $\sigma_t^2$ of the effective
observation at time $t$. According to the the general theory of state evolution
\cite{DMM09,BM-MPCS-2011}, the empirical distribution of the
coordinates of $z^t$ is asymptotically Gaussian with mean $0$ and
variance $\sigma_t^2$. 
This motivates the following estimator, first proposed in \cite{DMM09}:
\[
     \hat{\sigma}^t =  \frac{1}{\Phi^{-1}(0.75)}{\rm median}\{\, | z^t|\,
     \} \,  ,
\]
where $\Phi(z)$ is the normal distribution function.
This is known as the $25$\% pseudo-variance in robust estimator and
has the advantage of being insensitive to a small fraction of large outliers.
\item[Computation environment.] Computations were done in partly using
  Matlab, and partly through a Java program .
In the spirit of reproducible research, a suite of Java classes that
allow to repeat our simulations is available through an open code repository
\cite{AssemblaCode}.
\item[Numerical computation of minimax risk.] 
The plots of minimax risk were obtained by evaluating numerically the
expression in the main text. For separable and block-separable
denoisers (with the exception of the global minimax denoiser $\eall$),
the integrals can be expressed in terms of the Gaussian distribution
function or incomplete beta functions. 
For the global minimax denoiser, integration was performed numerically
using the standard MAtlab routines. 
 
Evaluation of the minimax risk required searching the least favorable
distribution among two point mixtures of the form  $(1-\eps) \delta_0
+ \delta_\mu$, in the separable case and
$(-\eps)\delta_0+\eps\delta_{\mu\,e_1}$ in the block separable one. 
Optimization over $\mu\in\bR_+$ was performed by brute force search
over a grid, with recursive refinement of the grid.

 For the non-separable cases, the procedure for approximating the minimax MSE was explained
Section \ref{sec:Mono} and \ref{sec:TV}.
\end{description}
 
 \section{Finite-$N$ scaling and error analysis}
\label{Appendix:NScaling}

The empirical phase transitions observed in this paper admit further analysis,
to verify whether the following expected behavior take place, namely: 
$(a)$ the offsets tend towards zero with increasing $N$;
$(b)$ the steepness of the phase transition increases with increasing $N$.

\subsection{Offsets decay toward zero}

As described in Section \ref{sec-PTExperiments-scalar}, at each fixed value $\eps$ of the sparsity parameter,
we gathered data at several different values of $\delta$, and obtained 
the empirical phase transition parameter $\hPT(N,\eps,\eta)$,
recorded as \emph{offset} from prediction, so that  $\hPT(N,\eps,\eta) = 0$ means that
the  50\% success location  fitted to the $\eps$-fixed, $\delta$-varying dataset
is exactly  at the predicted location $M(\eps | \eta)$.  Our analysis gave not only
the empirical phase transition location, but also its formal standard error
$SE(\hPT)$. (Here we make explicit the dependence of $\hPT$ on the
specific denoiser.)

We fit a linear model to the dataset including all the phase
transition results for soft and firm thresholding. We considered several exponents $\gamma$
that might be describing the decay of the offset with increasing $N$:
\beq \label{eq:OffsetModel}
      \hPT(N,\eps,\eta) = c(\eps,\eta) N^{-\gamma} + SE(\hPT) \cdot\normal(0,1),
\eeq
Table \ref{table:OffsetModel} shows that  $\gamma = 1/3$ provides an
adequate description of the offsets, with an $R^2$ exceeding $0.995$.
A plot of raw $\hPT$'s versus the predictions of model (\ref{eq:OffsetModel})
is given in Figure \ref{fig:OffsetModel}.

\begin{table}
\begin{center}
\begin{tabular}{|l|l|}
\hline
$\gamma$ & $R^2$ \\
\hline
$1/3$ & 0.99521 \\
$1/2$ &0.98724 \\
$2/3$ & 0.9666811 \\
$3/4$ & 0.950506\\
$1$    & 0.893454\\
\hline
\end{tabular}
\caption{Powers $\gamma$ and resulting $R^2$ for fitting model
  (\ref{eq:OffsetModel}) to the empirical offsets of various phase transitions.}
\label{table:OffsetModel}
\end{center}
\end{table}

\begin{figure}
\begin{center}
\includegraphics[height=3in]{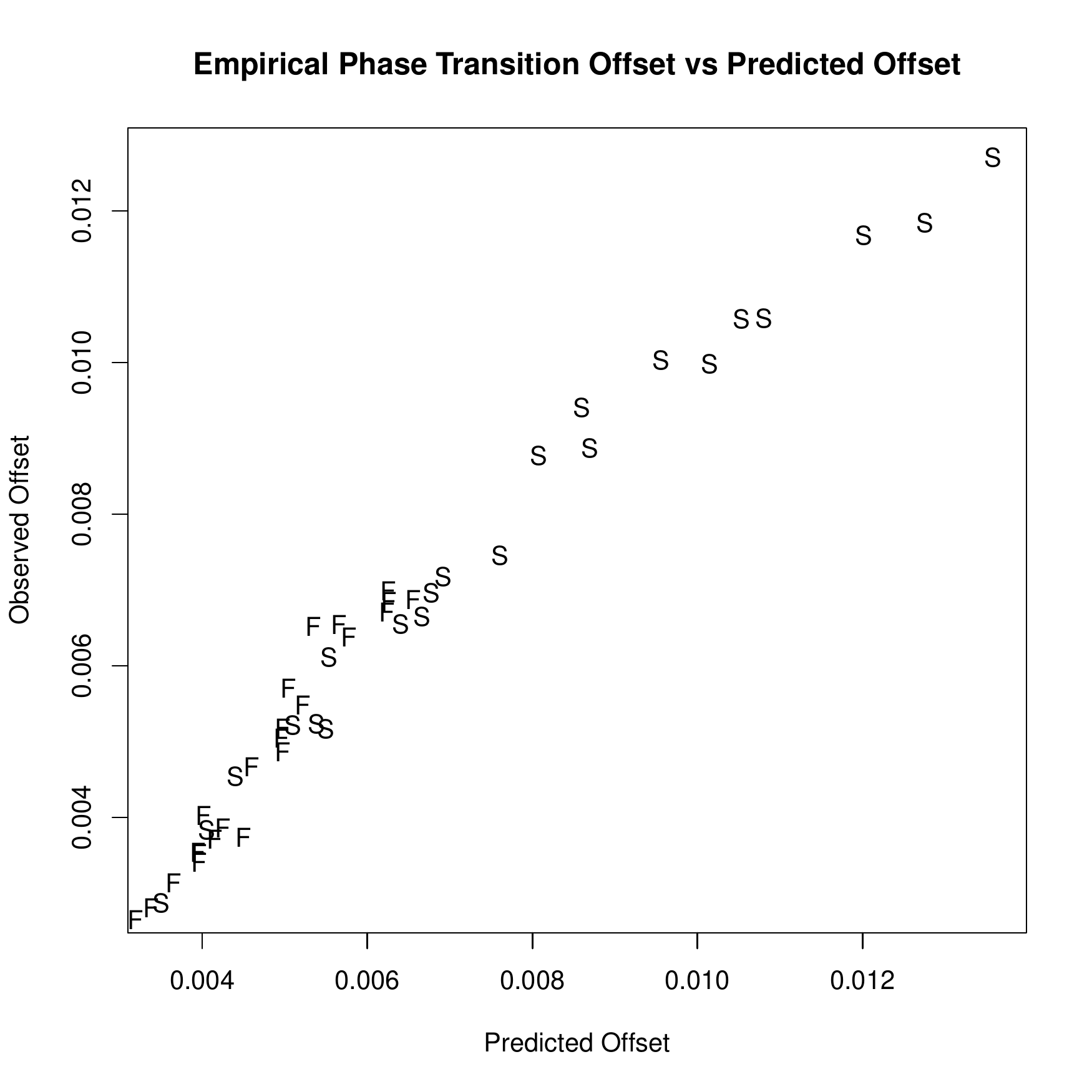}
\caption{Fitted offsets versus predicted offsets in model
  (\ref{eq:OffsetModel}) for various values of $\eps$ and soft and
  firm thresholding (respectively S and F).}
\label{fig:OffsetModel}
\end{center}
\end{figure}

\subsection{Transitions sharpen}

In addition to an empirical phase transition parameter $\hPT(N,\eps,\eta)$ 
we also fitted an empirical steepness parameter $\hbeta(N,\eps,\eta)$, according to
the logistic model:
\[
\log\frac{\hp_i}{1-\hp_i} = \beta \big(\delta_i -
      \hPT\big) , 
\]
where $\hp_i$ is the empirical success probability for the $i$-th experiment. 

We expect $\hbeta(N,\eps,\eta)$ to grow with increasing $N$, corresponding to increasingly abrupt transitions
from complete failure at $\delta \ll \hPT(N,\eps,\eta) $ to complete success at $\delta \gg \hPT(N,\eps,\eta)$.
In order to test this behavior, we fitted a linear model to the
values of $\hbeta$ computed for multiple values of $N$, $\eps$, and
denoisers $\eta$. We considered a range of powers $\tgamma$
that might be describing the growth of the steepness with increasing $N$:
\beq \label{eq:SlopeModel}
      \hbeta(N,\eps,\eta) = c(\eps,\eta) \, N^{\tgamma} + {\rm Error}.
\eeq

Table \ref{table:SlopeModel} shows that  $\gamma = 1/2$ provides an
adequate description of the steepnesses, with an $R^2$ exceeding $0.999$.
A plot of raw $\hbeta$'s versus the predictions of model (\ref{eq:SlopeModel})
is given in Figure \ref{fig:SlopeModel}.

\begin{table}
\begin{center}
\begin{tabular}{|l|l|}
\hline
$\gamma$ & $R^2$ \\
\hline
$1/3$ & 0.99286 \\
$1/2$ &0.99927\\
$2/3$ & 0.991064 \\
$3/4$ & 0.982350\\
$1$    & 0.947641\\
\hline
\end{tabular}
\caption{Powers $\gamma$ and resulting $R^2$ for fitting model
  (\ref{eq:SlopeModel}) to the empirical slope parameters.}
\label{table:SlopeModel}
\end{center}
\end{table}

\begin{figure}
\begin{center}
\includegraphics[height=3in]{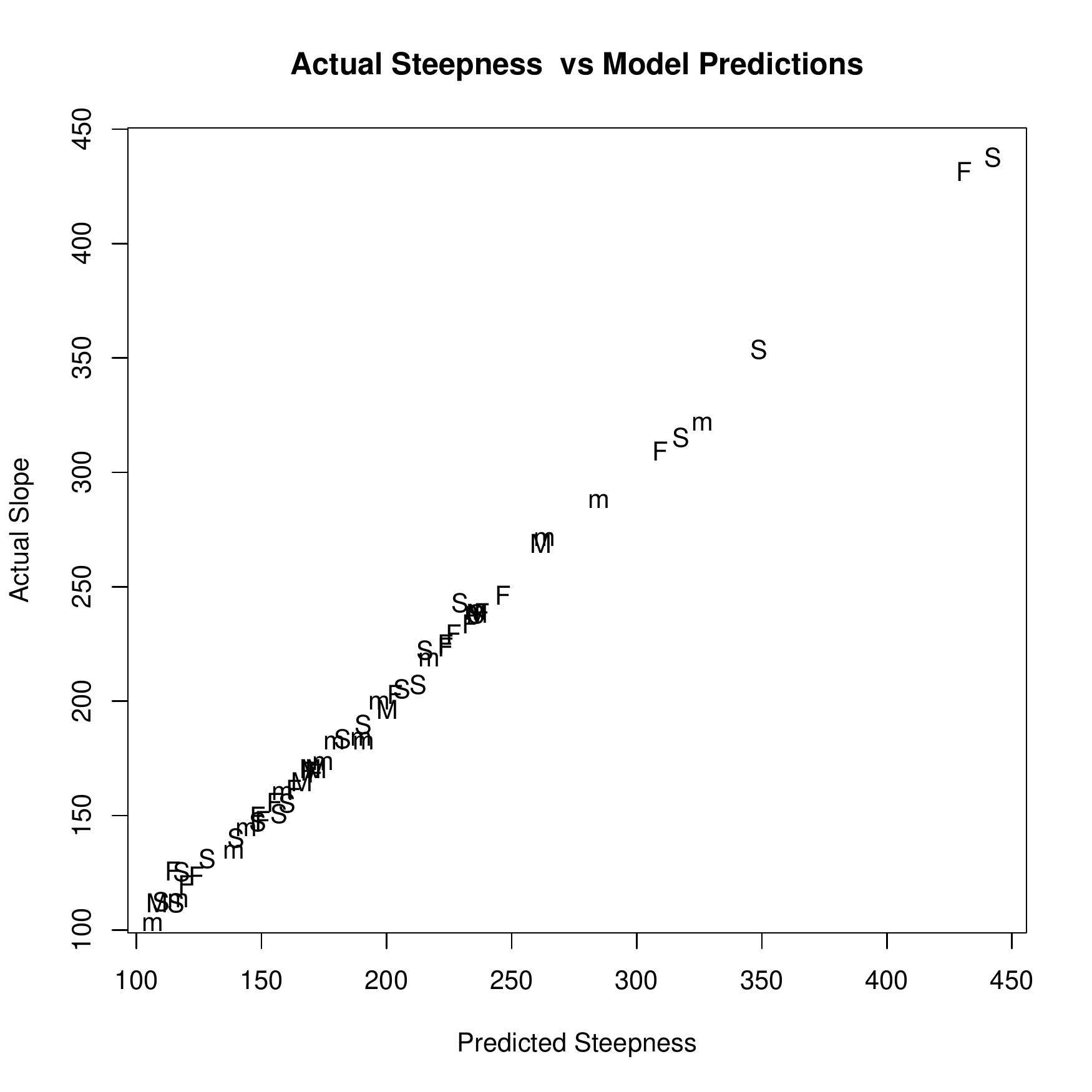}
\caption{Fitted steepness versus predicted steepness in model
  (\ref{eq:SlopeModel}) 
for soft (S), firm (F) and minimax (m) denoisers.}
\label{fig:SlopeModel}
\end{center}
\end{figure}

\section{Further properties of the risk function}
\label{Appendix:RiskProperties}

In this appendix we prove several useful properties of the risk
function of the block soft and
James-Stein denoisers. Throughout this section, we define the risk at 
$\mu$ as
\begin{eqnarray}
R(\mu;\eta) = \E\big\{\|\eta(\mu+\bz)-\mu\|_2^2\big\}\, .
\end{eqnarray}
The argument $\eta$ will be droppend or replaced by the threshold
level $\tau$ whenever clear from the context. 
Since we only consider denoisers that are equivariant under rotation,
$R(\mu;\eta)$ depends on the vector $\mu$ only through its norm
$\|\mu\|_2$. With a slight abuse of notation, we will use $\mu$ to
denote the norm as well. In other words, the reader can assume $\mu =
\mu\, e_1$.
%
%
\subsection{Block soft thresholding}
\label{sec:PropertiesOfBlock}

In this section we consider the block soft denoiser
$\esoft(\,\cdot\,;\tau):\bR^B\to\bR^B$.
We will write $R(\mu;\tau)$ for $R(\mu;\esoft(\,\cdot\,;\tau))$.
\begin{lemma}\label{lemma:PropertiesOfBlock}
For block soft thresholding the risk function $\mu \mapsto R(\mu\tau)$
has these properties:
\begin{align}
&  \mu  \rightarrow R(\mu;\tau) \text{ is monotone increasing},  \label{eq:mono} \\
& \frac{\partial\phantom{\mu}}{\partial \mu^2} R(\mu;\mu)  \leq 1, \label{eq:bound1} \\
& R(\mu;\tau)  \leq \min \{ R(0;\tau) + \mu; B + \tau^2 \}, \label{eq:oracle1}\\
& \lim_{\mu\to\infty}R(\mu;\tau)   = B + \tau^2 \, .  \label{eq:oracle2}
\end{align}
\end{lemma}

\begin{proof}
It will be convenient within the proof to set $d\equiv B$ and  $\xi \equiv \mu^2$.
Let $S^2 = \| \mu +\bZ\|_2^2$, so that $S^2$ is distributed
as a non-central chi-square $\chi^2_d(\xi)$ with, in particular,  $\E S^2 = d+\xi $.
Since the block soft thresholding rule is weakly differentiable, Stein's
unbiased estimate of risk yields $R(\mu;\tau) = \E U(S)$, with
\begin{displaymath}
  U(S) =
  \begin{cases}
    S^2 - d & \mbox{for $S < \tau$,} \\
    d+\tau^2 - 2(d-1) \tau S^{-1} & \mbox{for $S \geq \tau$.}
  \end{cases}
\end{displaymath}

Let $f_{\xi,d}(w)$ be the density function of 
$S^2 \sim \chi_d^2(\xi)$. This  satisfies
\begin{displaymath}
  \frac{\partial}{\partial \xi} f_{\xi,d}(w) 
  = - \frac{\partial}{\partial w} f_{\xi,d +2}(w)
  = \frac{1}{2}\, [f_{\xi,d+2}(w) - f_{\xi,d}(w)].
\end{displaymath}
Applying the first identity, integrating by parts, canceling terms,
and then using the second identity, we obtain
\begin{equation}
  \label{eq:deriv-formula}
  \frac{\partial\phantom{\mu}  }{\partial \mu^2} R(\mu^2;\tau) 
   = \int_0^{\tau^2} f_{\xi,d} (w)\, \de w + (d-1) \tau \int_{\tau^2}^\infty 
     w^{-3/2} f_{\xi,d+2}(w) \,\de w 
  \geq 0.
\end{equation}
For the upper bound (\ref{eq:bound1}), use the Poisson mixture
representation 
$f_{\xi,d+2}(w) = \sum_{j=0}^\infty p_{\xi/2}(j) f_{d+2+2j}(w)$ with
$p_{\lambda}(j)\equiv \lambda^je^{-\lambda}/j!$ and an
identity for the central $\chi^2$ density family to obtain
\begin{displaymath}
  \frac{f_{d+2+2j}(w)}{w} 
  = \frac{f_{d+2j}(w)}{d + 2j}
  \leq \frac{f_{d+2j}(w)}{d}  
\end{displaymath}
and so to conclude that $f_{\xi,d+2}(w) \leq (w/d) f_{\xi,d}(w)$. 
Hence the second term in (\ref{eq:deriv-formula}) is bounded by
$(1 - 1/\tau) \int_{\tau^2}^\infty f_{\xi,d}(w) \,\de w$, whence follows the
conclusion $(\partial R/ \partial \mu^2) \leq 1$. 
Property (\ref{eq:oracle1}) is immediate from (\ref{eq:mono}) and
(\ref{eq:bound1}) and the large-$\mu$ limit of $R$. 

To obtain the bound in (\ref{eq:oracle2}), write the risk function
using the unbiased risk formula as
\begin{displaymath}
R(\mu) = \xi + \E [ 2d + \tau^2 - S^2 - 2(d-1) \tau S^{-1}, S > \tau]
\end{displaymath}
and observe that the integrand is decreasing in $S$, and bounded above
by $2$ for $S > \tau$, so that $R(0;\tau) \leq 2 \prob( \chi_d^2 \geq \tau^2)$. 
This yields (\ref{eq:oracle2}). 
\end{proof}
%
%
\subsection{Positive-part James-Stein denoiser}
\label{sec:JSProperties}

As in the previous section, we set $\xi\equiv \mu^2$ and $d=B$.
HEre we consider the James-Stein denoiser
$\eJS:\bR^d\to\bR^d$ defined by
\begin{eqnarray}
\eJS(y) = \left(1 - \frac{(d-2)}{|y\|^2}\right)_+ \, y\,,
\end{eqnarray}
and we will write $R(\mu) = R(\mu;\eJS)$. 

We again let $S^2 = \| \mu+\bz \|^2_2$, with
$\bz\sim\normal(0,\id_{d\times d})$. 
We have the noncentral chi-squared distribution $S \sim \chi_d^2(\xi)$ 
with noncentrality $\xi$. 
Stein's unbiased estimate of risk is
\begin{displaymath}
  U(S) =
  \begin{cases}
    S^2 - d  &  \mbox{ for $S^2 < d - 2$}\, , \\
    d - (d-2)^2/S^2 & \mbox{for $S^2 > d-2$}\, ,
  \end{cases}
\end{displaymath}
and $R(\mu) = \E U(S)$.

\subsubsection{Risk  at $0$}
\label{sec:JS_at_0}

We will first develop an approximation of the risk at $0$ that was
used in Sections (\ref{sec:BlockJS}) and
(\ref{sec-PTExperiments-vector}).

Let $f_d(w)$ denote the density of a central chi-squared with $d$
degrees of freedom. We then have
 the density satisfies
\begin{align*}
 R(0) 
   & = \E \Big\{\Big\| \bigl( 1 - \frac{d-2}{S^2}\bigr)_+ \bz \Big\|^2\Big\} \\
   & = \int_{d-2}^\infty \Big[ w - 2(d-2) - \frac{(d-2)^2}{w}\Big]
   f_d(w)\, \de w \, .
\end{align*}
Using the identity $wf_{d-2}(w) = (d-2) f_{d}(w)$
and letting $D= d-2$, we can rewrite the last expression as
\begin{eqnarray}
  R(0) = \frac{1}{(d-2)}\int_{d-2}^\infty \big[w - (d-2) \big]^2
   f_{d-2}(w)\, \de w  = D^{-1} \E (\chi_D^2 - D)_+^2 . \label{eq:RiskAtZeroFormula}
\end{eqnarray}
We have the convergence in distribution, 
\[
D^{-1/2} (\chi_D^2 - D) \Rightarrow \normal(0,1), \qquad D \goto \infty.
\]
By a standard tightness argument, this implies that that
\[
  \lim_{d\to\infty} R(0) =2 \E Z_+^2 = 1\, ,
\] 
where $Z \sim \normal(0,1)$.

An Edgeworth series leads to the expansion $R(0) = 1+R_1\, d^{-1/2}+\Theta(d^{-2})$.
Indeed, one can integrate the   expression
(\ref{eq:RiskAtZeroFormula}) numerically, and the numerical values are
consistent with   $R(0) \approx 1 + 0.752 /\sqrt{d}$ for large $d$.

\subsubsection{Monotonicity of $R(\mu)$.}

We use the variation-diminishing version of total positivity,
developed in \cite{BrJoMa81}.
\begin{theorem}[Brown, Johnstone, and MacGibbon]
The non-central $\chi^2$ family is strictly variation diminishing of all
orders
\end{theorem}

For $g:[0,\infty)\to\bR$, let  $S^-(g)$ and $S^+(g)$ denote the number of sign changes and strict sign
changes of  $g$, and let
$IS(g)$ denote the  sign of $g(0)$ (assuming that $g(0)\neq 0$, the
more general definition being given in \cite{BrJoMa81}). 
Further define the function $\gamma:  [0,\infty) \mapsto \bR$ by
\begin{eqnarray}
    \gamma(\xi) = \int_0^{\infty} g(w) f_{d,\xi}(w)\,\de w\, ,\label{eq:GammaDef}
\end{eqnarray}
where $f_{d,\xi}(\,\cdot\,)$ is the density of the noncentral
chi-square with $d$ degrees of freedom  and noncentrality $\xi$.
By the SVR property we have that that $S^+(\gamma) \leq S^-(g)$ and that if
$S^+(\gamma) = S^-(g)$ then necessarily $IS(\gamma) = IS(g)$. 
In particular this implies that, if $g$ is strictly increasing, then
$\gamma$ is strictly increasing as well. Indeed this follows by
letting $g_a(w) = g(w)-a$ for $a\in\bR$ and
\[
    \gamma_a(\xi) \equiv \int_0^{\infty} g_a(w) f_{d,\xi}(w)\,\de w =\gamma(\theta)-a\, .
\]
If $g$ is strictly increasing, then $S^-(g_a)\le 1$ for all $a\in\bR$,
whence $S^{+}(\gamma_a)\le 1$ for all $a$, with $IS(\gamma) = IS(g)$
whenever $S^+(\gamma_a)=1$. This in turns implies that
$\gamma$ is increasing.

We now verify that the risk $R(\mu)$ of $\eta^{JS}$ is monotone increasing
in $\xi = \|\mu\|^2 \in [0,\infty)$. 
Let
\begin{displaymath}
  g(w)=U(w^{1/2}) =
  \begin{cases}
    w - d  &  \mbox{ for $w < d - 2$}\, , \\
    d - (d-2)^2/w & \mbox{for $w > d-2$}\, ,
  \end{cases}
\end{displaymath}
and define $\gamma(\xi)$ using Eq.~(\ref{eq:GammaDef}). 
Note that $g$ is strictly increasing and hence $\xi\mapsto\gamma(\xi)$
is
increasing as well by the above argument. 
But $U(S)$ is Stein's unbiased risk estimator and hence $R(\mu) =
\gamma(\xi=\|mu\|^2)$, which implies the claim.
 
\subsection{Proof of Lemma \ref{lemma:Star}}

Since any probability distribution is written as a convex combination
of point masses, it is sufficient to prove the claim for $\nu =
\delta_{\mu}$.
In this case, using the scaling relation
(\ref{eq:ScalingRelation}), we have
\begin{eqnarray}
\Psi(m;\delta,\tau,\delta_{\mu})  =
\frac{m}{\delta}\, R\big(\mu\sqrt{\delta/m}\big)\, ,
\end{eqnarray}
with $R(\,\cdot\,)$ the risk function.
Therefore the state evolution mapping is starshaped for all
distributions $\nu$ if and only if $\mu\mapsto R(\mu)$ is monotone
increasing.

The monotonicity of the risk function was proved in \cite{DMM09} for 
soft thresholding and positive soft thresholding. It is proved in
Section \ref{sec:Mono} and \ref{sec:TV} for monotone regression and
total variation denoising.
Finally, it is proved in Section \ref{sec:PropertiesOfBlock} and
\ref{sec:JSProperties} for block soft and James-Stein denoisers.

\bibliographystyle{amsalpha}

\newcommand{\etalchar}[1]{$^{#1}$}
\providecommand{\bysame}{\leavevmode\hbox to3em{\hrulefill}\thinspace}
\providecommand{\MR}{\relax\ifhmode\unskip\space\fi MR }
\providecommand{\MRhref}[2]{%
  \href{http://www.ams.org/mathscinet-getitem?mr=#1}{#2}
}
\providecommand{\href}[2]{#2}

\end{document}